\documentclass[12pt]{article}
\usepackage{amsmath,amsthm}
\usepackage{enumerate}
\usepackage{natbib}
\usepackage{setspace}
\renewcommand{\baselinestretch}{1.5} 
\usepackage{amsmath}
\usepackage{enumerate}
\usepackage{bbm}
\usepackage{pgf,tikz}
\usetikzlibrary{decorations.pathreplacing,angles,quotes}
\usepackage{mathrsfs}
\newtheorem{theorem}{Theorem}
\newtheorem{assumption}{Assumption}
\newtheorem{lemma}{Lemma}
\newtheorem{remark}{Remark}
\usepackage{color} 
\usepackage{enumitem} 
\usepackage{mathtools}
\usepackage{titlesec}
\titleformat{\section}{\large\bfseries}{\thesection}{1em}{}
\titleformat{\subsection}{\normalsize\bfseries}{\thesubsection}{1em}{}
\titleformat{\subsubsection}{\normalsize\bfseries}{\thesubsubsection}{1em}{}
\usepackage{pifont}% http://ctan.org/pkg/pifont
\usepackage{amssymb}
\usepackage{booktabs}
\usepackage{graphicx}
\usepackage{tabularx}
\usepackage{hyperref}
\usepackage{titlesec}

\titleformat{\paragraph}
{\normalfont\normalsize\bOseries}{\theparagraph}{1em}{}
\titlespacing*{\paragraph}
{0pt}{3.25ex plus 1ex minus .2ex}{1.5ex plus .2ex}

\usepackage{apptools}
\AtAppendix{\counterwithin{lemma1}{section}}

\newtheorem{prop}{Proposition}[section]
\usepackage[]{appendix}

\allowdisplaybreaks
\usepackage{hyperref}
\usepackage{color}
\usepackage{booktabs} 
\usepackage{xr}
\usepackage{amsmath}
\usepackage{amssymb}
\usepackage{caption}
\usepackage{subcaption}
\usepackage{booktabs}
\usepackage{graphicx}
\usepackage{centernot}
\usepackage{titlesec}
\usepackage{dsfont}
\usepackage{soul}
\usepackage{thmtools, thm-restate}

\usepackage{multirow}
\usepackage{caption}

\usepackage{bbm}
\usepackage{pdflscape, array, rotating, multirow}

\usepackage{centernot}

\usepackage{lscape}
\usepackage{color}
\usepackage{tikz}

\usetikzlibrary{shapes,decorations,arrows,calc,arrows.meta,fit,positioning}
\tikzset{
    -Latex,auto,node distance =1 cm and 1 cm,semithick,
    state/.style ={ellipse, draw, minimum width = 0.7 cm},
    point/.style = {circle, draw, inner sep=0.04cm,fill,node contents={}},
    bidirected/.style={Latex-Latex,dashed},
    el/.style = {inner sep=2pt, align=left, sloped}
}
\usepackage{environ} 
\NewEnviron{SMALLER}{% 
    \scalebox{0.93}{$\BODY$} 
}
\DeclareMathOperator{\E}{\textnormal{\mbox{E}}}
\def\bSig\mathbf{\Sigma}

\makeatletter
\newcommand*{\indep}{%
  \mathbin{%
    \mathpalette{\@indep}{}%
  }%
}
\newcommand*{\nindep}{%
  \mathbin{%                   % The final symbol is a binary math operator
    \mathpalette{\@indep}{\not}% \mathpalette helps for the adaptation
  }%
}
\newcommand*{\@indep}[2]{%
  \sbox0{$#1\perp\m@th$}%        box 0 contains \perp symbol
  \sbox2{$#1=$}%                 box 2 for the height of =
  \sbox4{$#1\vcenter{}$}%        box 4 for the height of the math axis
  \rlap{\copy0}%                 first \perp
  \dimen@=\dimexpr\ht2-\ht4-.2pt\relax
  \kern\dimen@
  {#2}%
  \kern\dimen@
  \copy0 %                       second \perp
} 
\makeatother

\makeatletter

\begin{document}
%% Can only be called _after_ \begin{document}

\def\spacingset#1{\renewcommand{\baselinestretch}%
{#1}\small\normalsize} \spacingset{1}

%%%%%%%%%%%%%%%%%%%%%%%%%%%%%%%%%%%%%%%%%%%%%%%%%%%%%%%%%%%%%%%%%%%%%%%%%%%%%%

% \if1\blind
% {
%   \title{\bf Causal inference under transportability assumptions for conditional relative effect measures}
%   \author{Guanbo Wang$^{1,2}$\thanks{
%     The authors gratefully acknowledge \textit{please remember to list all relevant funding sources in the unblinded version}}\hspace{.2cm} Alexander W. Levis$^{3}$, Jon A. Steingrimsson,$^{4}$ \\and Issa J. Dahabreh$^{1,2,5}$\\\\
% {\small $^{1}$CAUSALab, Harvard T.H. Chan School of Public Health, Boston, MA} \\
% {\small $^{2}$Department of Epidemiology, Harvard T.H. Chan School of Public Health, Boston, MA} \\
% {\small $^{3}$Department of Statistics \& Data Science, Carnegie Mellon University, Pittsburgh, PA}\\
% {\small $^{4}$Department of Biostatistics, Brown University, Providence, RI}\\
% {\small $^{5}$Department of Biostatistics, Harvard T.H. Chan School of Public Health, Boston, MA}
%     }
%   \maketitle
% } \fi

% \if0\blind
% {
%   \bigskip
%   \bigskip
%   \bigskip
%   \begin{center}
%     {\LARGE\bf Causal inference under transportability assumptions for conditional relative effect measures}
% \end{center}
%   \medskip
% } \fi
\begin{center}
\Large{\textbf{Causal inference under transportability assumptions for conditional relative effect measures}}\\
\vspace{1cm}
\small{Guanbo Wang$^ {1, 2,*}$
Alexander W. Levis$^{3}$, Jon A. Steingrimsson$^{4}$, and Issa J. Dahabreh$^ {1,2,5}$ }\\
\vspace{0.5cm}
\footnotesize{$^{1}$CAUSALab, Harvard T.H. Chan School of Public Health, Boston, MA }\\
\footnotesize{$^{2}$Department of Epidemiology, Harvard T.H. Chan School of Public Health, Boston, MA }\\
\footnotesize{$^{3}$Department of Statistics \& Data Science, Carnegie Mellon University, Pittsburgh, PA}\\
\footnotesize{$^{4}$Department of Biostatistics, Brown University, Providence, RI}\\
\footnotesize{$^{5}$Department of Biostatistics, Harvard T.H. Chan School of Public Health, Boston, MA  \\
}
\end{center}
\bigskip
\begin{abstract}
When extending inferences from a randomized trial to a new target population, the transportability condition for conditional difference effect measures is invoked to identify the marginal causal mean difference in the target population. However, many clinical investigators believe that conditional relative effect measures are more likely to be ``transportable'' between populations. Here, we examine the identification and estimation of the marginal counterfactual mean difference and ratio under the transportability condition for conditional relative effect measures. We obtain identification results for two scenarios that often arise in practice when individuals in the target population (1) only have access to the control treatment, and (2) have access to the control and other treatments but not necessarily the experimental treatment evaluated in the trial. We then propose model and rate multiply robust and nonparametric efficient estimators that allow for the use of data-adaptive methods to model the nuisance functions. We examine the performance of the methods in simulation studies and illustrate their use with data from two trials of paliperidone for patients with schizophrenia. We conclude that the proposed methods are attractive when background knowledge suggests that the transportability condition for conditional relative effect measures is more plausible than alternative conditions.
\end{abstract}

\noindent%
\vfill

\newpage
\spacingset{1.45} % DON'T change the spacing!

%%%%%%%%%%%%%%%%%%%%%%%%%%%%%%%%%%%%%%%%%%%%%%%%%%%%%%%%%%%%%%%%%%%%%%%%%%%%%%
\section{Introduction}\label{sec:introduction}
%%%%%%%%%%%%%%%%%%%%%%%%%%%%%%%%%%%%%%%%%%%%%%%%%%%%%%%%%%%%%%%%%%%%%%%%%%%%%%

In the growing literature on extending causal inferences from a randomized trial to a target population, conditions about the generalizability or transportability of difference effect measures (e.g., equality of the conditional average treatment effect between populations), or even stronger conditions of transportability in expectation or distribution of potential outcomes across populations, are invoked to identify the marginal causal mean difference \citep{pearl2015, westreich2017transportability, dahabreh2018generalizing,  westreich2019target, dahabreh2020extending, degtiar2023review}. These conditions are strong and untestable and require scrutiny on the basis of substantive knowledge. 

In many clinical contexts, investigators believe that relative effect measures, such as conditional risk and mean ratios, are more likely to be transportable between populations compared to absolute effect measures \citep{glasziou1995evidence, schwartz2006ratio, spiegelman2017evaluating,deeks2019analysing}. Based on these clinical intuitions, \citet{huitfeldt2019collapsibility, huitfeldt2019effect} described a formal transportability condition for relative effect measures. Under a non-nested study design where investigators can access individual-level data from a randomized trial and a sample of separately target population data \citep{dahabreh2021study}, \citet{dahabreh2024learning} used this condition to identify the marginal counterfactual mean difference and ratio in the target population. They also showed that the transportability condition for relative effect measures is incompatible with the transportability condition for difference effect measures unless the treatment has no effect on average or stronger transportability conditions hold. 

Here, we examine how the transportability condition of relative effect measures can be used to estimate the marginal counterfactual mean difference and ratio in two scenarios that often arise in practice when individuals in the target population (1) only have access to the control treatment, or (2) have access to the control and other treatments but not necessarily the experimental treatment evaluated in the trial. In the latter scenario, we consider the possibility of having to use covariates in the target population to control confounding, in addition to those needed for transportability. In both scenarios, it is not feasible to identify the target parameter using data only from the target population, because we do not assume that the experimental treatment is accessible within that population, nor do we assume that the treatment's association with the outcome is unconfounded. In this context, learning about the target population using information about the treatment effect from the trial may be useful both to inform decisions and to plan future studies. 

Under the transportability condition of relative effect measures, we develop model and rate multiply robust and non-parametric efficient estimators for the marginal counterfactual mean difference and ratio. The estimators allow for the use of data-adaptive methods (e.g., machine-learning techniques) to model the nuisance parameters. We also present the results of simulation studies that examine the finite sample performance of the estimators and an application of the methods using data from two trials that compare treatments for schizophrenia.

%%%%%%%%%%%%%%%%%%%%%%%%%%%%%%%%%%%%%%%%%%%%%%%%%%%%%%%%%%%%%%%%%%%%%%%%%%%%%%
\section{Uniform use of the control treatment in the target population} \label{sec:scenario1}
%%%%%%%%%%%%%%%%%%%%%%%%%%%%%%%%%%%%%%%%%%%%%%%%%%%%%%%%%%%%%%%%%%%%%%%%%%%%%%

We begin by considering the scenario where control treatment is the only option in the target population. This situation arises, for example, when a trial compares a new treatment that is unavailable outside the experimental setting versus a control treatment that is uniformly adopted in the target population \citep{huitfeldt2018choice}.

\subsection{Data and estimands} 

Denote the covariates, the binary population source indicator, binary treatment, and outcome by $X, S, A$, and $Y$, respectively, where $S = 1$ indexes the trial participants, and $S = 0$ indexes the target population; $A=1$ refers to the experimental treatment and $A = 0$ the control treatment. The observed data are represented by $O_i=(X_i, S_i, A_i, Y_i)$, $i = 1, \ldots, n_0+n_1$, where $n_1$ and $n_0$ are the sample sizes of the trial and target population, respectively. We assume that the data are the independent and identically distributed draws of the populations, and the baseline covariates may distribute differently in the two populations. See the data structure in the left panel of Figure \ref{fig:RR_Cases}.
\begin{figure}[h]
\centering
\includegraphics[width=\textwidth]{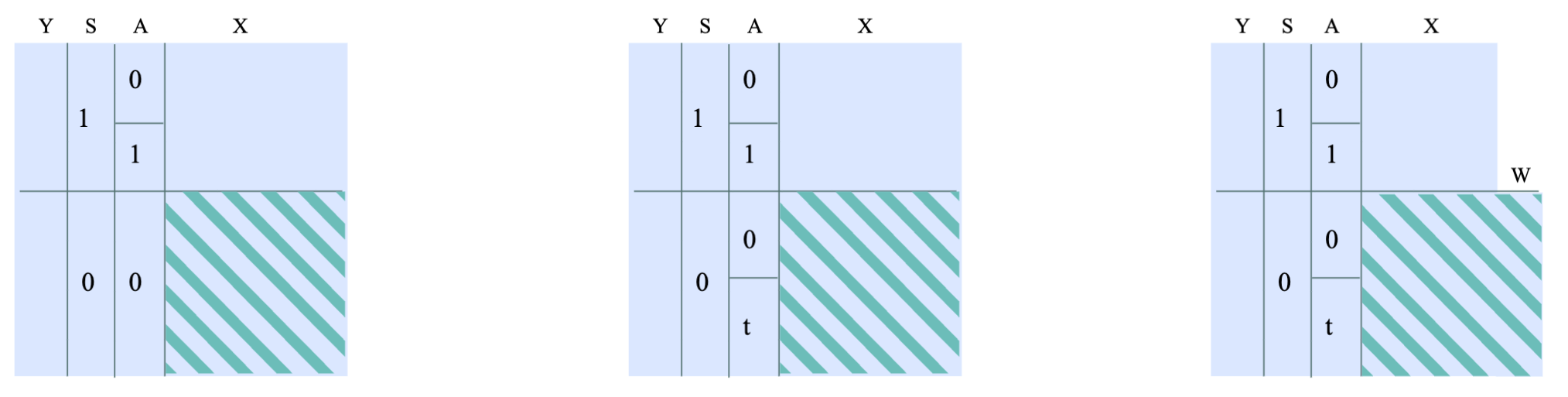}
\caption{The left, middle, and right panel represent the data structure of the case described in Section \ref{sec:scenario1}, \ref{sec:scenario2} and \ref{sec:scenario3} respectively. The shaded dataset represents the target population. The treatment $t$ can be any treatment other than the control treatment.}
\label{fig:RR_Cases}
\end{figure}

Denote the counterfactual outcome under intervention to set treatment $A$ to $a$ by $Y^a$ \citep{neyman1923application, rubin1974estimating, robins2000d}, possibly contrary to the observed outcome $Y$. Throughout this work, our target parameters are the marginal counterfactual means in the target populations $\E[Y^1 | S = 0]$ and $\E[Y^0 | S = 0]$, and their ratio and difference, $\E[Y^1/Y^0 | S = 0]$ and $\E[Y^1-Y^0 | S = 0]$. 

\subsection{Identification analysis}

Consider the following identifiability condition:\\ 
\noindent
\emph{(A1) Consistency:} For each $a \in \{0, 1\}$ and every individual $i$, if $A_i = a$, then $Y_i = Y^{a}_i$.\\
\noindent
\emph{(A2) Exchangeability over $A$ in the trial:} for each $a \in \{0, 1\}$ every $x$ with positive density in the trial $f(x, S = 1) \neq 0$, $\E[Y^a|X=x, S=1, A=a]=\E[Y^a|X=x, S=1]$.\\
\noindent
\emph{(A3) Positivity of treatment assignment in the trial:} $\Pr[A=a | X = x, S=1]>0$ for each $a \in \{0, 1\}$ and every $x$ with positive density in the trial $f(x, S = 1) \neq 0$.\\
\noindent
\emph{(A4) Conditional transportability of relative effect measures over $S$:}  $\E[Y^0 | X = x , S = 1] \neq 0$, and $\E[ Y^0 | X = x , S = 0] \neq 0$, and
\begin{equation*}
        \dfrac{\E [Y^1 | X = x , S = 1]}{\E[Y^0 | X = x , S = 1]} = \dfrac{\E[Y^1  | X = x , S = 0]}{\E[Y^0 | X = x , S = 0]},
\end{equation*} 
for every $x$ with positive density in the target population $f(x, S = 0) \neq 0$.\\
\noindent
\emph{(A5) Positivity of trial participation:} $\Pr[S=1 | X = x]>0$ for every $x$ with positive density in the target population $f(x, S = 0) \neq 0$.\\
\noindent
\emph{(A6) Uniform use of the control treatment in the target
population:} $S=0\implies A=0$, that is, if $S_i=0$, then $A_i=0$
for every individual $i$.\\
\noindent
In particular, condition (A4) connects conditional relative effect measures in the trial and the target population. It can be understood as a formalization of the common intuition that relative treatment effects are reasonably stable across populations, conditional on covariates. The relative common transportability condition is \\
\noindent
\emph{(A4*) Conditional transportability of difference effect measures over $S$:} for each $a\in\{0,1\}$ and every $x$ such that $f(x, S = 0) \neq 0$, $\E[Y^a|X, S=1]= \E[Y^a|X, S=0]$ (to identify marginal counterfactual means and their ratio) and $\E[Y^1-Y^0|X, S=1]= \E[Y^1-Y^0|X, S=0]$ (to identify the difference of the marginal counterfactual means). 

Unless the treatment has no effect on average, condition (A4*) implies but is not necessarily implied by condition (A4) \citep{dahabreh2024learning}. More specifically, conditional on $X$, condition (A4) allows $S$ to be associated with $Y$. That is, in a causally directed acyclic graph, condition (A4) allows a causal path from $S$ to $Y$ not through $X$. See more discussions in Supplementary Materials A (\ref{appendix:CausalStructure}). Such situations may occur, for example, when trial participants have more favored outcomes due to the potential higher adherence or medical advice acquired during follow-up \citep{van2011participation, dahabreh2019generalizing}, or recruiters selectively enroll patients into the trial based on what the next treatment allocation is likely to be \citep{kahan2015risk}. 

Conditions (A1)-(A5) will be assumed throughout this paper; in this section, we consider the special case where condition (A6) holds. The next theorem gives the identification results; the proof is given in \ref{appendix:identification1}, along with other identification strategies.
\begin{theorem}
\label{thm:identification1}
Under conditions (A1) through (A6), $\E[Y^1|S=0]$ and $\E[Y^0| S = 0]$ can be identified by 
$$
\alpha_1 \equiv \E \left[\dfrac{\E[Y|X, S=1, A=1]}{\E[Y|X, S=1, A=1]}Y \Big\vert S = 0\right]
$$
and $\beta_{1}\equiv\E(Y| S = 0)$ respectively. Then $\E[Y^1| S = 0]/\E[Y^0| S = 0]$ and $\E[Y^1-Y^0| S = 0]$ can be identified by  $\phi_{1}\equiv\alpha_{1}/\beta_{1}$ (assuming $\beta_1\neq 0$) and $\psi_{1}\equiv\alpha_{1}-\beta_{1}$, respectively.
\end{theorem}

Due to condition (A6), the outcome and treatment in $\{S=0\}$ are not confounded. Therefore, it is not necessary to model the outcomes in $\{S=0\}$ to identify the target parameters. Because $\beta_1$ can be estimated by the empirical mean of the outcomes in $\{S=0\}$, the challenging component to estimate is $\alpha_1$ in the estimation of $\phi_1$ and $\psi_1$.

\subsection{Estimation}\label{sec:estimation}
To construct asymptotically normal and non-parametric efficient estimators while allowing flexible estimation of nuisance functions, we will base the estimation of target parameters on their influence functions \citep{bickel1993efficient}. For brevity, we use ``the influence function'' to indicate ``the non-parametric influence function'' in this paper without further notice. 
Define the following nuisance functions, $\kappa\equiv\Pr[S=0]^{-1}, \mu_{s,a}(X)\equiv\E[Y|X, S=s, A=a], q(X)\equiv\Pr[A=1|X, S=1]$, and $ \tau(X)\equiv\Pr[S=0|X]/\Pr[S=1|X]$. In addition, denote the true values of all these nuisance functions by $\boldsymbol{\eta_{1}}=\{\kappa, \mu_{1, 1} (X), \mu_{1, 0} (X), \mu_{0, 0}(X), q(X), \tau(X)\}$, and their estimates by $\boldsymbol{\widehat \eta_{1}}=\{\widehat \kappa, \widehat \mu_{1, 1} (X), \widehat \mu_{1, 0} (X), \widehat \mu_{0, 0}(X), \widehat q(X), \widehat \tau(X)\}$. The following lemma gives the influence functions of the target parameters; the proof is given in \ref{appendix:IF1}. 
\begin{lemma}
\label{thm:IF1}
The influence functions of $\alpha_{1}$ and $\beta_{1}$ are
\begin{align*}
    A_{1}(O, \boldsymbol{\eta_{1}})=&\kappa\Bigg(
    S\tau(X)\dfrac{\mu_{0,0}(X)}{\mu_{1,0}(X)}
    \Big[
    \dfrac{A}{q(X)}\{Y-\mu_{1, 1}(X)\}-\dfrac{1-A}{1-q(X)}\dfrac{\mu_{1,1}(X)}{\mu_{1,0}(X)}\{Y-\mu_{1, 0}(X)\}
    \Big]\\
    &\hspace{0.25in}+(1-S)
    \left\{\dfrac{\mu_{1,1}(X)}{\mu_{1,0}(X)}Y-\alpha_{1}\right\}
    \Bigg),
\end{align*}
\noindent
and $B_{1}(O)=Y(1-S)-\beta_{1}$ respectively. The influence functions of $\phi_{1}$ and $\psi_{1}$ are $\Phi_{1}(O, \boldsymbol{\eta_{1}})=\frac{1}{\beta_{1}}\{A_{1}(O, \boldsymbol{\eta_{1}})-\phi_{1} B_{1}(O)\}$ and $\Psi_{1}(O, \boldsymbol{\eta_{1}})=A_{1}(O, \boldsymbol{\eta_{1}})-B_{1}(O)-\psi_1$ respectively.
\end{lemma}

We propose estimating the target parameters by solving an estimating equation based on an estimated version of the influence functions. For example, we can
use the empirical average of $A_{1}(O, \boldsymbol{\widehat \eta_{1}})$ as an estimating function. Throughout, we assume the use of sample splitting and cross-fitting \citep{schick1986, pfanzagl2012, chernozhukov2018double, kennedy2020sharp} for the estimation, so that the nuisance functions do not rely on empirical process conditions (e.g., Donsker-type conditions) and are allowed to be modeled by data-adaptive methods (e.g., machine learning methods). For example, one can split the data into two halves (each half consists of stratified samples according to the data source), use one to estimate nuisance functions, and the other to compute $\widehat \phi_{s, a}(\widetilde{x})$. The samples can then be swapped, the same procedure is performed, and then averaged to obtain the final estimate.

\subsection{Inference}\label{sec:inference}
We next give the asymptotic properties of the proposed estimator. In particular, we show that, with certain conditions, our estimator is model and rate doubly robust, asymptotically normal, and nonparametrically efficient. 
% Notably, the properties do not rely on the condition that the functions to estimate the nuisance parameters fall into the Donsker class, which enables users to utilize flexible machine-learning methods to adjust for high-dimensional covariates.

Since $\beta_1$ can be estimated by the sample mean of $Y(1-S)$, by central limit theory, we know that $\widehat \beta_1$ is $\sqrt{n}$-consistent and asymptotically normal. In the estimation of $\alpha_1$, and thus $\phi$ and $\psi$, we first note that the marginal probability $\kappa$ can be estimated by $\widehat \kappa = \{\mathbb{P}_{n}(1-S)\}^{-1}$, and the estimator satisfies i) $I(\widehat \kappa^{-1}=0)=o_p(n^{-1/2})$, and ii) $\widehat \kappa\xrightarrow{P} \kappa$; see proof in \ref{appendix: kappa}. Since $q(X)$ is the propensity score in the trial population, it is typically known. It is thus reasonable to assume that  $q(X)$ can be estimated at a $\sqrt{n}$-rate. 
To establish the asymptotic properties of $\widehat \alpha_{1}$, thus $\widehat \phi_{1}$ and $\widehat \psi_{1}$, we make the following conditions:
\begin{enumerate}
\item[(a1)] $\exists \varepsilon>0, \quad s.t. \quad \Pr[\varepsilon\leqslant \kappa\leqslant1-\varepsilon]=
\Pr[\varepsilon\leqslant \widehat \kappa\leqslant1-\varepsilon]=\Pr[\varepsilon\leqslant q(X)\leqslant1-\varepsilon]=
\Pr[\varepsilon\leqslant \widehat q(X)\leqslant1-\varepsilon]=
\Pr[\varepsilon\leqslant \tau(X)]=
\Pr[\varepsilon\leqslant \widehat \tau(X)]=1$,
\item[(a2)] $ \E[Y^{2}]<\infty$,
\item[(a3)] $\left\lVert \boldsymbol{\widehat \eta_1}-\boldsymbol{\eta_1} \right\rVert =o_p(1)$ 
%All the nuisance functions can be estimated with a rate $o_p(1)$.
\end{enumerate}
Denote $L_2$ norm by $||\cdot||$. The following theorem gives the asymptotic properties of $\widehat \alpha_1$; a detailed proof is given in ~\ref{appendix:inference1}. 
\begin{theorem}
\label{thm:inference1}
If conditions \emph{(a1)} through \emph{(a3)} hold, then $\widehat \alpha_1$ is consistent with rate of convergence $O_p(R_{1, n}+n^{-1/2})$, where  
% $
%      \widehat \alpha_{1}  -\alpha_{1}  
%     =\mathbb{P}_n\{A_{1}(O, \boldsymbol{\eta_{1}})\}+R_{1, n}+o_p(n^{-1/2}),
% $ 
% \begin{equation*}
% \label{Decompose11}
$
    R_{1, n}=
    \left\lVert \widehat \mu_{1, 1}(X)/\widehat \mu_{1, 0}(X)-\mu_{0, 0}(X)/ \mu_{1, 0}(X)\right\rVert 
    (
    \big\lVert \widehat \mu_{0, 0}(X)$- $\mu_{0, 0}(X) \big\rVert+
    \left\lVert \widehat \mu_{1, 0}(X)-\mu_{1, 0}(X) \right\rVert
    +
    \left\lVert \widehat \tau(X)-\tau(X) \right\rVert
    ).
$
% \end{equation*}
% or equivalently,
% \begin{equation*}
% % \label{Decompose12}  
% \SMALLER{
%     R_{1, n}\lesssim O_p\Big\{
%     \big(
%     \left\lVert \widehat \mu_{1, 1}(X)-\mu_{1, 1}(X) \right\rVert+
%     \left\lVert \widehat \mu_{1, 0}(X)-\mu_{1, 0}(X) \right\rVert\big)
%     \big(
%     \left\lVert \widehat r_S(X)-r_S(X) \right\rVert+
%     \left\lVert \widehat \tau(X)-\tau(X) \right\rVert\big)
%     \Big\}.
%     }%
% \end{equation*}
When $R_{1, n}= o_p(n^{-1/2})$, $\widehat \alpha_1$ is asymptotically normal and non-parametric efficient.
% $
%     \sqrt{n}(\widehat \alpha_{1}  -\alpha_{1})\rightsquigarrow\mathcal{N}\big[0, \E\{A_{1}(O, \boldsymbol{\eta_{1}})^{2}\}\big],
% $
% % For $\widehat \beta_{1}$, we have
% $
%     \sqrt{n}(\widehat \beta_{1}  -\beta_{1})\rightsquigarrow\mathcal{N}\big[0, \E\{B_{1}(O)^{2}\}\big],
% $
% $ 
%       \sqrt{n}(\widehat \phi_{1}  -\phi_{1})\rightsquigarrow\mathcal{N}\big[0, \E\{\Phi_{1}(O, \boldsymbol{\eta_{1}})^{2}\}\big],
% $
% and 
% $
%       \sqrt{n}(\widehat \psi_{1}  -\psi_{1})\rightsquigarrow\mathcal{N}\big[0, \E\{\Psi_{1}(O, \boldsymbol{\eta_{1}})^{2}\}\big].
% $
\end{theorem}
\noindent
Because $\widehat \beta_1$ is $\sqrt{n}$-consistent and asymptotically normal, the estimators $\widehat \phi_1$ and $\widehat \psi_1$ share the same properties as $\widehat \alpha_1$.
Theorem \ref{thm:inference1} states that to obtain  $o_p(1)$-consistent estimators, we only require that all nuisance functions are consistently estimated---though potentially quite slowly---at an $o_p(1)$ convergence rate. Many flexible and data-adaptive machine learning methods satisfy the requirement, which creates opportunities to flexibly adjust for high-dimensional covariates. More importantly, the theorem lays out the condition for the estimator to be $\sqrt{n}$-consistent, asymptotically normal, and nonparametric efficient--it is sufficient that $R_{1, n} = o_p(n^{-1/2})$. Next, we will analyze this condition in several aspects.

Firstly, through this condition, we characterize the convergence rate requirements of nuisance functions estimation for the estimators to be $\sqrt{n}$-consistent and asymptotically normal. Because $R_{1, n}$ is a product of the errors of two sets of nuisance functions, $\widehat \alpha_1$ doubly robust. More specifically, depending on whether the nuisance functions are estimated parametrically or non-parametrically, the estimator can be said to be \textit{model doubly robust} or \textit{rate doubly robust} \citep{chernozhukov2018double, smucler2019unifying, rotnitzky2021characterization, ying2024geometric}. Formally, $\widehat \alpha_1$ is model double robust in terms of $\{\widehat \mu_{1, 1}(X)/\widehat \mu_{1, 0}(X)\}$ and $\{ \widehat \mu_{0, 0}(X), \widehat \mu_{1, 0}(X), \widehat \tau(X)\}$ because if either set of nuisance functions can be correctly specified by parametric models and converge at a $\sqrt{n}$-rate, the condition $R_{1, n} = o_p(n^{-1/2})$ holds, and $\widehat \alpha_1$ is $\sqrt{n}$-consistent and asymptotically normal. On the other hand, $\widehat \alpha_1$ is rate double robust in terms of the two sets of nuisance functions because when the product of the error rates of the (non-parametric) estimates of the two sets of the nuisance functions converges faster than $\sqrt{n}$, the condition $R_{1, n} = o_p(n^{-1/2})$ holds, and $\widehat \alpha_1$ is $\sqrt{n}$-consistent and asymptotically normal. For example, when $\widehat \mu_{1, 1}(X), \widehat \mu_{1, 0}(X), \widehat \mu_{0, 0}(X), \widehat \tau(X)$ can be estimated by $n^{1/4}$-rates,
$R_{1, n} = o_p(n^{-1/2})$. This is possible under sparsity, smoothness, or other assumptions about the underlying models. For instance, under some conditions \citep{horowitz2009semiparametric}, generalized additive model estimators can attain rates of $O_p(n^{-2/5})$, which is $o_p(n^{-1/4})$. In general, the convergence rate for $\widehat \alpha_1$ can be much faster than that of any component nuisance function \citep{kennedy2019robust}. If the condition $R_n = o_p(n^{-1/2})$ is satisfied, then confidence intervals can be constructed using either bootstrap or a direct estimate of the asymptotic variance, which can be computed by the sample variance of the influence functions. 

Secondly, we examine the first component of $R_{1, n}$, $\left\lVert \widehat \mu_{1, 1}(X)/\widehat \mu_{1, 0}(X)-\mu_{1, 1}(X)/ \mu_{1, 0}(X)\right\rVert$. Two caveats need our attention. The first is its relationship to $\big\lVert \widehat \mu_{1, 1}(X)-\mu_{1, 1}(X) \big\rVert+\left\lVert \widehat \mu_{1, 0}(X)-\mu_{1, 0}(X) \right\rVert$. Note that, $\forall a\leqslant 0$,
\begin{align}
    &\big\lVert \widehat \mu_{1, 1}(X)-\mu_{1, 1}(X) \big\rVert+\left\lVert \widehat \mu_{1, 0}(X)-\mu_{1, 0}(X) \right\rVert=o_p(n^{a}) \label{stronger condition}\\
    \Rightarrow
    &\left\lVert \widehat \mu_{1, 1}(X)/\widehat \mu_{1, 0}(X)-\mu_{1, 1}(X)/ \mu_{1, 0}(X)\right\rVert=o_p(n^{a}) \label{original condition}\\
    \Leftrightarrow
    &\exists\delta(X)\neq 0 \text{ s.t. }\left\lVert \widehat \mu_{1, 1}(X)-\mu_{1, 1}(X)\delta(X)\right\rVert+ \left\lVert \widehat \mu_{1, 1}(X)-\mu_{1, 1}(X)\delta(X)\right\rVert=o_p(n^{a}),\label{equivalent condition}
\end{align}
where the arrow means ``implies''. Conditions (\ref{original condition}) and (\ref{equivalent condition}) are weaker than condition (\ref{stronger condition}) because they do not require $\widehat \mu_{1, 1}(X)$ and $\widehat \mu_{1, 0}(X)$ to converge to their true values. However, they require that the two estimates have the same amount of bias. This may happen, for instance, when the observed outcomes (with measurement error) in the trial are systematically larger than the true values. 
\begin{assumption}\label{assumption}
    For two functions of $X$, $f(X)$ and $g(X)\neq 0$, and their estimates $\widehat f(X)$ and $\widehat g(X)\neq 0$, $||\widehat f(X)/\widehat g(X)-f(X)/g(X)||=o_p(n^{a})\Leftrightarrow||\widehat f(X)-f(X)||+||\widehat g(X)-g(X)||=o_p(n^{a})$ for any $a\leqslant 0$.
\end{assumption}
When such nuance is ignored, i.e., conditions (\ref{stronger condition}) and (\ref{original condition}) are exchangeable, or when assumption \ref{assumption} holds, $R_{1, n}$ can be written as
\begin{align}\label{equ:alternativeR1n}
    \left\lVert \widehat \mu_{1, 0}(X)- \mu_{1, 0}(X)\right\rVert
    +\left\lVert \widehat \mu_{1, 1}(X)- \mu_{1, 1}(X)\right\rVert
    (
    \big\lVert \widehat \mu_{0, 0}(X)-\mu_{0, 0}(X) \big\rVert
    +
    \left\lVert \widehat \tau(X)-\tau(X) \right\rVert
    ).
\end{align}
The change in $R_{1, n}$ will change the convergence rate requirements for the nuisance functions estimation. For example, when nuisance functions are estimated via parametric models, to obtain $\sqrt{n}$-consistent estimators, $\widehat \mu_{1, 0}(X)$ must converge at a $\sqrt{n}$-rate. In addition, $\widehat \mu_{1, 1}(X)$ or $\{\widehat \mu_{0, 0}(X), \tau(X)\}$ has to converge at a $\sqrt{n}$-rate. The quantity $\widehat \mu_{1, 0}(X)$ is critical to the estimation of $\alpha_{1}$ because we do not invoke the condition (A4*) but (A4) to identify $\alpha_{1}$. This results in the fact that $\widehat \mu_{1, 0}(X)$ appears in each of the terms in the influence function, which means that $\widehat \mu_{1, 0}(X)$ heavily influences the estimation.

Another caveat to pay attention to is that 
$$
\left\lVert \widehat \mu_{1, 1}(X)/\widehat \mu_{1, 0}(X)-\mu_{1, 1}(X)/ \mu_{1, 0}(X)\right\rVert=o_p(n^{a}) 
    \nLeftrightarrow
    \left\lVert \widehat r(X)-\mu_{1, 1}(X)/ \mu_{1, 0}(X)\right\rVert=o_p(n^{a}), 
$$
where $\widehat r(X)$ is a direct estimate of $\mu_{1, 1}(X)/ \mu_{1, 0}(X)$  (e.g., using a Poisson model when the outcome is count or more flexible methods \citep{yadlowsky2021estimation}) rather than the estimate that obtained as a ratio of two outcome model estimates $\widehat \mu_{1, 1}(X)/\widehat \mu_{1, 0}(X)$. That is, even if $\mu_{1, 1}(X)/ \mu_{1, 0}(X)$ can be estimated using a correctly specified parametric model and converge at a $\sqrt{n}$-rate, it does not imply that $\widehat \mu_{1, 1}(X)/\widehat \mu_{1, 0}(X)$ can converge at a $\sqrt{n}$-rate, and thus $\widehat \alpha_1$ may not be $\sqrt{n}$-consistent. Because $R_{1, n}$ cannot be written as a function of $\left\lVert \widehat r(X)-\mu_{1, 1}(X)/ \mu_{1, 0}(X)\right\rVert$, the estimator that involves estimating $\mu_{1, 1}(X)/ \mu_{1, 0}(X)$ directly does not share the properties stated in Theorem \ref{thm:inference1}.

Thirdly, the $R_{1, n}$ in Theorem \ref{thm:inference1} can be also written as $\big(
    || \widehat \mu_{1, 1}(X)-\mu_{1, 1}(X) ||+
    || \widehat \mu_{1, 0}(X)-\mu_{1, 0}(X) ||\big)
    \big(
    || \widehat \mu_{0, 0}(X)/\widehat \mu_{1, 0}(X)-\mu_{0, 0}(X)/\mu_{1, 0}(X) ||+
    \left\lVert \widehat \tau(X)-\tau(X) \right\rVert\big)$; see the proof in \ref{appendix:inference1}. It implies that, if $\{\mu_{0, 0}(X)/\mu_{1, 0}(X), \mu_{1, 1}(X), \mu_{1, 0}(X), \tau(X)\}$ (instead of \\$\{\mu_{1, 1}(X)/\mu_{1, 0}(X),\mu_{0, 0}(X), \mu_{1, 0}(X), \tau(X)\}$) can be estimated at, for example, $n^{1/4}$ rates, then $R_{1, n}=o_p(n^{-1/2})$ as well. However, when assumption \ref{assumption} holds, this alternative representation of $R_{1, n}$ equals (\ref{equ:alternativeR1n}). That is, the model and the rate double robustness property will not change due to this alternative representation when assumption \ref{assumption} holds.

\subsection{Estimators under the condition (A4*)}\label{sec:A4*}
We provide the identification results, the influence function, and the corresponding estimators of the target parameters under the relatively common condition (A4*) (instead of condition (A4)); see details in \ref{appendix:A4*1}. Applying the same estimation strategy, we denote the estimator under the condition (A4*) by $\widehat \alpha'_{1}$. The corresponding bias term $R'_{1n}$ is $\left\lVert \widehat \mu_{1, 1}(X)-\mu_{1, 1}(X) \right\rVert \left\lVert \widehat \tau(X)-\tau(X) \right\rVert$. The estimator is rate and model doubly robust in terms of $\widehat \mu_{1, 1}(X)$ and $\widehat \tau(X)$. However, when condition (A4) holds but condition (A4*) not, $\widehat \alpha'_{1}$ is biased even if all the nuisance functions are correctly specified. We further verify this claim in the simulation study.

%%%%%%%%%%%%%%%%%%%%%%%%%%%%%%%%%%%%%%%%%%%%%%%
\section{Treatment variation in the target population} \label{sec:scenario2and3}
%%%%%%%%%%%%%%%%%%%%%%%%%%%%%%%%%%%%%%%%%%%%%%%

In this section, we examine the scenario where the sample from the target population has access to treatments beyond the control group. To illustrate, consider a trial comparing the effects between an experimental treatment and control treatment, alongside another previous study (target population) investigating the effect comparison between another now-existing treatment (or multiple treatments) and control treatment. In this particular context, our objective is to investigate the effect comparison between experimental treatment and control treatment within this specific target population, which may be of interest in the drug approval process. 

% In what follows, we first assume that the covariates collected in both the trial and the target population are sufficient to control for confounding in Section \ref{sec:scenario2}. In Section \ref{sec:scenario3}, we investigate the case where additional covariates are collected in the target population and are necessary to control the confounding.

\subsection{Sufficient confounding control}\label{sec:scenario2}
In this subsection, we assume that the covariates collected in both the trial and the target population are sufficient to control for confounding. See the corresponding data structure depicted in the middle panel of Figure \ref{fig:RR_Cases}. In such a data structure, unlike the data structure described in Section \ref{sec:scenario1}, the outcome in the target population is confounded.

To identify the target parameters, we make the following additional conditions.

\noindent
\emph{(B1) Partial exchangeability:} for every $x$ with positive density in the target population $f(x, S=0)\neq 0$, $\E [ Y^0 | X = x , S = 0, A = 0] = \E[Y^0 | X = x, S = 0]$.

\vspace{0.1in}
\noindent
\emph{(B2) Positivity of receiving control:}  for every $x$ with positive density in the target population $f(x, S=0)\neq 0$, $0<\Pr[A = 0 | X = x, S = 0]<1$.

\begin{remark}
   We do not require the full exchangeability over $A$ in $\{S=0\}$ (i.e., $\E[Y^a | X = x , S = 0, A = 0]= \E [Y^a | X = x, S = 0]$ for all possible $a$ in the target population). Instead, we only require a weaker condition--Condition (B1) to identify the target parameters; interested readers can refer to \citep{sarvet2020graphical}. 
\end{remark}
Condition (B1) is untestable and requires a case-by-case examination based on domain knowledge. Condition (B2) is testable in principle, but formal evaluation may be challenging if $X$ is high-dimensional \citep{petersen2012diagnosing}.

The next theorem gives the identification results; the proof is given in \ref{appendix:identification2}, along with other identification strategies.
\begin{theorem}
\label{thm:identification2}
Under conditions (A1) through (A5), (B1), and (B2), $\E[Y^1|S=0]$ and $\E[Y^0|S=0]$ can be identified by $\alpha_2\equiv\E[\mu_{1, 1}(X) \mu_{0, 0}(X) /\mu_{1, 0}(X)| S = 0]$ and 
% \begin{align*} 
%    &\E \left\{ r_A(X) \mu_{0, 0}(X) | S = 0 \right\},\qquad\hspace{1.425cm} 
%    \E \left\{ r_A(X) (1-A)/\pi(X) Y| S = 0 \right\}, \\
%    &\dfrac{1}{\Pr(S=0)}\E \left\{S\tau(X)A/q(X)r_S(X)Y  \right\}, 
%    \dfrac{1}{\Pr(S=0)}\E \left[S\tau(X)(1-A)/\{1-q(X)\}r_A(X)r_S(X)Y \right],
% \end{align*}
and $\beta_{2}\equiv\E[\mu_{0, 0}(X) | S = 0]$, respectively. Then $\E[Y^1| S = 0]/\E[Y^0| S = 0]$ and $\E[Y^1-Y^0| S = 0]$ can be identified by  $\phi_{2}\equiv\alpha_{2}/\beta_{2}$ (assuming $\beta_2\neq 0$) and $\psi_{2}\equiv\alpha_{2}-\beta_{2}$, respectively.
\end{theorem}
Compared to the identification results showed in Theorem \ref{thm:identification1}, because the outcome in the target population are confounded, the outcome model in the target population $\mu_{0, 0}(X)$ is necessary for identifying the target parameters. 

Denote a new nuisance function $\Pr[A=0|X, S=0]$ by $\pi(X)$. Furthermore, for the estimation of $\alpha_2$, let $\boldsymbol{\eta_{2}}$ be the new collection of nuisance functions $\{\boldsymbol{\eta_{1}}, \pi(X)\}$, and let $\boldsymbol{\widehat \eta_{2}}=\{\boldsymbol{\widehat \eta_{1}}, \widehat \pi(X)\}$ be its estimates; for the estimation of $\beta_1$, denote $\boldsymbol{\eta'_{2}}=\{\kappa, \mu_{0,0} (X), \pi(X)\}$ and $\boldsymbol{\widehat \eta_{2}}=\{\widehat \kappa, \widehat \mu_{0,0} (X), \widehat \pi(X)\}$. Similar to Section \ref{sec:estimation}, we will base the estimation of target parameters on their influence functions. The following lemma gives the influence functions of the target parameters; the proof is given in \ref{appendix:IF2}. 
\begin{lemma}
\label{thm:IF2}
The influence functions of $\alpha_{2}$, $\beta_{2}$, $\phi_{2}$, and $\psi_{2}$ are
\begin{align*}
    A_{2}(O, \boldsymbol{\eta_{2}})=&\kappa\Bigg(
    (1-S)
    \dfrac{\mu_{1, 1}(X)}{\mu_{1, 0}(X)}\Bigg[\mu_{0, 0}(X)-\alpha_{2}  +\dfrac{1-A}{\pi(X)}\big\{Y-\mu_{0, 0}(X)\big\}
    \Bigg]\\
    &+
    S\tau(X)\dfrac{\mu_{0, 0}(X)}{\mu_{1, 0}(X)}\Bigg[
    \dfrac{A}{q(X)}\big\{Y-\mu_{1, 1}(X)\big\}-\dfrac{1-A}{1-q(X)}\dfrac{\mu_{1, 1}(X)}{\mu_{1, 0}(X)}\big\{Y-\mu_{1, 0}(X)\big\}
    \Bigg]
    \Bigg),\\
    B_{2}(O, \boldsymbol{\eta'_{2}}) =&\kappa(1-S)  \Bigg[
    \mu_{0, 0}(X)-\beta_{2}  +\dfrac{1-A}{\pi(X)}\big\{Y-\mu_{0, 0}(X)\big\} 
    \Bigg],\\
    \Phi_{2}(O; \boldsymbol{\eta_{2}})=&\dfrac{1}{\beta_{2}}\{A_{2}(O; \boldsymbol{\eta_{2}})-\phi_{2} B_{2}(O; \boldsymbol{\eta'_{2}})\},
    \Psi_{2}(O; \boldsymbol{\eta_{2}})=A_{2}(O; \boldsymbol{\eta_{2}})-B_{2}(O; \boldsymbol{\eta'_{2}}).
\end{align*}
These are also the semiparametric efficient influence functions of the target parameters when $\pi(X)$ is known.
\end{lemma}
The lemma shows that when the target population is a trial population, the estimators based on the above influence functions are also the efficient estimators that share the same properties as described below.

Once the influence functions are known, we apply the same estimation strategy as described in Section \ref{sec:estimation}, and obtain the estimators $\widehat \alpha_2, \widehat \beta_2, \widehat \phi_2$ and $\widehat \psi_2$. 
The next theorem gives the asymptotic properties of the proposed estimators; the proof is given in \ref{appendix:inference2}.
\begin{enumerate}
\item[(b1)] $\exists \varepsilon>0, \quad s.t. \quad \Pr[\varepsilon\leqslant \kappa\leqslant1-\varepsilon]=
\Pr[\varepsilon\leqslant \widehat \kappa\leqslant1-\varepsilon]=\Pr[\varepsilon\leqslant q(X)\leqslant1-\varepsilon]=
\Pr[\varepsilon\leqslant \widehat q(X)\leqslant1-\varepsilon]=
\Pr[\varepsilon\leqslant \tau(X)]=
\Pr[\varepsilon\leqslant \widehat \tau(X)]=
\Pr\{\varepsilon\leqslant \pi(X)\leqslant1-\varepsilon\}=
\Pr\{\varepsilon\leqslant \widehat \pi(X)\leqslant1-\varepsilon\}=1$,
\item[(b2)] $ \E(Y^{2})<\infty$,
\item[(b3)] $\left\lVert \boldsymbol{\widehat \eta_2}-\boldsymbol{\eta_2} \right\rVert =o_p(1)$ 
\end{enumerate}
% i.e., bounded nuisance functions and their estimates and outcomes, as well as consistent estimation of nuisance functions. The difference is those conditions also apply to the newly defined nuisance function $\pi(X)$. We refer to these conditions as conditions (b1) through (b3), which are given in \ref{appendix:inference2}. 
\begin{theorem}\label{thm:inference2}
If conditions \emph{(b1)} through \emph{(b3)} hold, then $\widehat \beta_2$ is consistent with rate of convergence $O_p(R_{2, n}^{\beta}+n^{-1/2})$, where  
$
    R_{2, n}^{\beta}=\left\lVert \widehat\mu_{0, 0}(X)-\mu_{0, 0}(X)\right\rVert
     \left\lVert \widehat\pi(X)-\pi(X)\right\rVert.
$
When $R_{2, n}^{\beta}= o_p(n^{-1/2})$, $\widehat \beta_2$ is asymptotically normal and non-parametric efficient. In addition, 
$\widehat \alpha_2, \widehat \phi_2, \widehat \psi_2$ are consistent with rate of convergence $O_p(R_{2, n}^{\alpha}+n^{-1/2})$, where  
$
    R_{2, n}^{\alpha}=R_{1, n}+R_{2, n}^{\beta}.
$
When $R_{2, n}^{\alpha}= o_p(n^{-1/2})$, $\widehat \alpha_2, \widehat \phi_2, \widehat \psi_2$ are asymptotically normal and non-parametric efficient.
\end{theorem}
\noindent
Unlike $\widehat \beta_1$, which can be estimated by the empirical average of $Y(1-S)$, the estimation of $\beta_2$ involves nuisance functions estimation. In observation of $R_{2, n}^{\beta}$, $\beta_2$ is rate and model doubly robust in terms of $\widehat \mu_{0, 0}$ and $\widehat \pi(X)$. On the other hand, $R_{2, n}^{\alpha}$ consists of two product terms, which results in multiple robustness property of $\widehat \alpha_2$. The estimator $\widehat \alpha_2$ is \textit{model multiply robust} in the sense that if 
\begin{align*}
    \Big[\widehat \mu_{1, 1}(X)/\widehat \mu_{1, 0}(X) \text{ or } \{ \widehat \mu_{0, 0}(X), \widehat \mu_{1, 0}(X), \widehat \tau(X)\}\Big] \text{ and } \Big[\widehat \mu_{0, 0} \text{ or } \widehat \pi(X)\Big]
\end{align*}
are correctly specified and converge at $\sqrt{n}$-rates, then $\widehat \alpha_2$ is $\sqrt{n}$-consistent and asymptotically normal. In other words, when one of the following sets of nuisance functions 
\begin{align*}
    (i) \{\widehat \mu_{1, 1}(X)/\widehat \mu_{1, 0}(X),  \widehat \mu_{0, 0}(X)\}, \quad
    (ii) \{\widehat \mu_{1, 1}(X)/\widehat \mu_{1, 0}(X),  \widehat\pi(X)\}, \quad
    (iii) \{\widehat \mu_{0, 0}(X), \widehat \mu_{1, 0}(X), \widehat \tau(X)\}
\end{align*}
is correctly specified and converges at $\sqrt{n}$-rate, $\widehat \alpha_2$ is $\sqrt{n}$-consistent and asymptotically normal. When assumption \ref{assumption} holds and $\widehat \mu_{1, 0}(X)$ can be correctly specified and converge at a $\sqrt{n}$-rate, $\widehat \alpha_2$ is multiply robust in terms of
\begin{align*}
    (i) \{\widehat \mu_{1, 1}(X),  \widehat \mu_{0, 0}(X)\}, \quad
    (ii) \{\widehat \mu_{1, 1}(X),  \widehat\pi(X)\}, \quad
    (iii) \{\widehat \mu_{0, 0}(X), \widehat \tau(X)\}.
\end{align*}
Heuristically, when the three outcome models are correctly specified, then $\widehat \alpha_{2}$ is consistent; when $\widehat \mu_{1,1}(X)$ cannot be correctly specified, $\widehat \tau(X)$ can be used to debias the residual if it can be correctly specified; when $\widehat \mu_{0,0}(X)$ cannot be correctly specified, $\pi(X)$ can be used to debias the residual if it can be correctly specified. In addition, $\widehat \alpha_2$ is \textit{rate multiply robust} in the sense that if both the products of the error rates of the nuisance function estimators converge faster than $\sqrt{n}$, then $\widehat \alpha_2$ is $\sqrt{n}$-consistent and asymptotically normal. In other words, when 
$\widehat \mu_{1, 1}(X), \widehat \mu_{1, 0}(X), \widehat \mu_{0, 0}(X), \widehat \tau(X), \widehat \pi(X)$ can be estimated by $n^{1/4}$-rates (e.g., generalized additive models), $\widehat \alpha_2$ is $\sqrt{n}$-consistent and asymptotically normal. We also present other forms of the decomposition of $R_{2, n}$ in \ref{appendix:inference2}, which admits the same conclusion for robustness and efficiency when assumption \ref{assumption} holds. Because the double robustness of $\widehat \beta_{2}  $ is ``nested'' in the multiple robustness of $\widehat \alpha_{2}$ (when $\widehat \alpha_{2}$ is $\sqrt{n}$-consistent, $\widehat \beta_{2}$ is $\sqrt{n}$-consistent), $\widehat \phi_{2}$ and $\widehat \psi_{2}$ are also multiply robust.

In summary, compared to the scenarios outlined in Section \ref{sec:scenario1}, when alternative treatment options exist within the target population, an additional layer of complexity emerges as we seek to adjust for confounding of the outcomes in the target population. This complexity necessitates modeling either $\mu_{0, 0}(X)$ or $\pi(X)$ with a required (though relatively slow) convergence rate to ensure the consistency and efficiency of the estimators.

%%%%%%%%%%%%%%%%%%%%%%%%%%%%%
\subsection{Additional confounders in the target population}\label{sec:scenario3}
%%%%%%%%%%%%%%%%%%%%%%%%%%%%%
Many applied statisticians and epidemiologists believe that statistical interaction on the multiplicative scale is uncommon \citep{spiegelman2017evaluating}. Therefore, it is motivated to assume that the set of covariates to render the transportability between the trial and target population (the covariates in condition (A4)) is a subset of the covariates to adjust for confounders in the target population (the covariates in condition (B1)).

Denote the additional confounder by $W$, we assume the data structure to be $(X_i, S_i = 1, A_i, Y_i)$, $i = 1, \ldots, n_1$, and $(X_i, W_i, S_i = 0, A_i, Y_i)$, $i = 1, \ldots, n_0$. See the data structure in the right panel of Figure \ref{fig:RR_Cases}. Conditions (B1) and (B2) can be replaced with the following two conditions.

\noindent
\emph{(C1) Partial exchangeability:} for every $x$ with positive density in the target population $f(x, S=0)\neq 0$, $\E [ Y^0 | X = x , W=w, S = 0, A = 0] = \E[Y^0 | X = x, W=w, S = 0]$.

\vspace{0.1in}
\noindent
\emph{(C2) Positivity of receiving control:}  for every $x$ with positive density in the target population $f(x, S=0)\neq 0$, $0<\Pr[A = 0 | X = x, W=w, S = 0]<1$.

When conditions (A1) through (A5) and (C1) through (C2) hold, the target parameters are identifiable and can be estimated using the same strategy as described in previous sections; the resulting estimators are also multiply robust. Theorem S J.1 gives the identification results, Lemma S J.1 gives the influence functions, and Theorem S J.2 gives the asymptotic properties of the estimators. The theorems and the lemma are given in \ref{appendix:scenario3}; their proofs are given in \ref{appendix:identification3}-\ref{appendix:inference3}. We highlight the interesting results below.

Denote the statistics for identifying $\E[Y^1|S=0], \E[Y^0|S=0], \E[Y^1/Y^0|S=0]$ and $\E[Y^1-Y^0|S=0]$ by $\alpha_{3}, \beta_{3}, \phi_3$ and $\psi_3$, respectively. Because the outcome in the target population is confounded and needs additional covariates to adjust for confounding, estimating these target parameters requires estimating new nuisance functions: $\mu'_{0, 0}(X, W)=\E[Y | X, W, S = 0, A = 0]$, 
$M(X)=\E \left[ \E[Y | X, W, S = 0, A = 0] \big | X, S = 0 \right]$, and 
$\pi'(X, W)=\Pr[A=0|X, W, S=0]$. 

The asymptotic properties of the estimators are heavily based on the bias terms of the estimators; similar to Theorems \ref{thm:inference1} and \ref{thm:inference2}, we present the bias terms in estimating $\beta_3$ and $\alpha_3$ below.
\begin{align*}
     &R_{3, n}^{\beta}=\left\lVert \widehat\mu'_{0, 0}(X, W)-\mu'_{0, 0}(X, W)\right\rVert
     \left\lVert \widehat\pi'(X, W)-\pi'(X, W)\right\rVert,\\
     &R_{3, n}^{\alpha}=\left\lVert \widehat \mu_{1, 1}(X)/\widehat \mu_{1, 0}(X)-\mu_{1, 1}(X)/ \mu_{1, 0}(X)\right\rVert
    (
    || \widehat M(X)-M(X) ||+\\
    &\hspace{2.5cm}\left\lVert \widehat \mu_{1, 0}(X)-\mu_{1, 0}(X) \right\rVert+
    || \widehat \tau(X)-\tau(X) ||
    )+R_{3, n}^{\beta}.
\end{align*}
Analyzing $R_{3, n}^{\beta}$, we see that $\widehat \beta_3$ is doubly robust in terms of $\widehat\mu'_{0, 0}(X)$ and $\widehat\pi'(X)$. Compared to $\widehat \beta_2$, the nuisance functions for estimating $\widehat \beta_3$ are functions conditional on $W$. Next, we analyze $R_{3, n}^{\alpha}$.We first observe that $\mu'_{0, 0}(X, W)$ is nested in $M(X)$ in the sense that $\mu'_{0, 0}(X, W)$ is the outcome of the model to estimate $M(X)$. Therefore, the following assumption is reasonable for any $a\leqslant 0$. 
\begin{assumption}\label{assumption2}
    $||\widehat M(X)-M(X)||=o_p(n^{a})\Rightarrow||\widehat \mu'_{0, 0}(X, W)-\mu'_{0, 0}(X, W)||=o_p(n^{a})$.
\end{assumption}
The estimator $\widehat \alpha_3$ is model multiply robust in the sense that when assumptions \ref{assumption} and \ref{assumption2} hold, and $\widehat \mu_{1, 0}(X)$ can be correctly specified and converge at a $\sqrt{n}$-rate, if one of the following sets of nuisance functions 
\begin{align*}
    (i) \{\widehat \mu_{1, 1}(X),  \widehat \mu'_{0, 0}(X, W)\}, \quad
    (ii) \{\widehat \mu_{1, 1}(X),  \widehat \pi'(X, W)\}, \quad
    (iii) \{\widehat M(X), \widehat \tau(X)\}.
\end{align*}
is correctly specified and converges at $\sqrt{n}$-rate, $\widehat \alpha_3$ is $\sqrt{n}$-consistent and asymptotically normal. We see that a correctly specified $\widehat \pi'(X, W)$ can debias the bias caused by a misspecified model for $\widehat \mu'_{0, 0}(X, W)$. However, to debias the bias caused by a misspecified model for $\widehat \mu_{1, 1}(X)$, a correctly specified $\widehat \tau(X)$ alone is not enough--$\widehat M(X)$ has also to be correctly specified. This is because, again, the covariates to adjust for confounders between two populations are not sufficient to control for confounders in the target population. Similar to $\widehat \alpha_2$, $\widehat \alpha_3$ is also rate multiply robust.

%%%%%%%%%%%%%%%%%%%%%%%%%%%%%%%%%%%%%%%%%%%%%%%
\section{Simulation}\label{sec:Simulations}
%%%%%%%%%%%%%%%%%%%%%%%%%%%%%%%%%%%%%%%%%%%%%%%

Since the estimators proposed in Sections \ref{sec:scenario2and3} are extensions of the estimator described in Section \ref{sec:scenario1}, $\widehat \alpha_{1}$; for illustration purposes, we focused on the performance of $\widehat \alpha_{1}$ in this simulation study. We conducted simulations to evaluate the finite sample performance of $\widehat \alpha_{1}$ and compared it with alternative estimators (described below). The simulation code can be found in \url{https://github.com/Guanbo-W/Transportability_Weakercondition}.

\subsection{Assessment of finite-sample bias, variance, and coverage}
We generated data consisting of samples from a trial (denoted by $S=1$) and a target population (denoted by $S=0$). We evaluated the performance of the estimator with different sample sizes: the total sample size of the trial $n_{1}$ was 250, 500 or 1000, and the sample size of the target population $n_{0}$ was 5000. 

We generated five uniformly distributed variables $X_{j}\sim U(0, 1), j=1,\dots,5$.  We assumed the data source indicator followed a Bernoulli distribution, $S\sim \text{Bernoulli}\{\Pr[S=1|X]\}$ with $\Pr[S=1|X]=\text{exp}(\lambda X^T) /\{1+ \text{exp}(\lambda X^T)\}$, $X = (1, X_{1}, \ldots, X_{5})$ and $\lambda$ was a vector length of six containing $\lambda_0$ and all others $\ln(1.05)$, where we solved for $\lambda_0$ to result (on average) in the desired total sample size \citep{robertson2021intercept}. We assumed the propensity score $\Pr[A=1|X, S=1]=0.5$ in the trial and $\Pr[A=0|X, S=0] = 1$ in the target population. For the outcomes, we first generated $Y^0$ and $Y^1$, which followed Bernoulli distributions with probabilities $p_0$ and $p_1$, respectively, where $p_0=\exp{\{m(X)+Sg(X)\}}$ and $p_1=\exp{\{m(X)+r(X)+Sg(X)\}}$ (both $p_0$ and $p_1$ were bounded by 0 and 1). Then we generated $Y$ as $AY^1+(1-A)Y^0$. That is, the data-generating mechanism satisfied condition (A4) but not (A4*) (corresponding to case 2 in Table \ref{tab:DGM} in \ref{appendix:CausalStructure}). In all settings, the true values were generated by averaging $10^{5}$ empirical means of $Y^1$ in the subset of $\{S=1\}$. All simulation results were an average of 5000 runs, unless otherwise stated, $\widehat \mu_{1, 0}(X)$ and $\widehat q(X)$ were correctly specified. See \ref{appendix:Sim} for more detailed model specification and misspecification. 
\begin{figure}[h]
    \centering
    \includegraphics[width=0.95\textwidth]{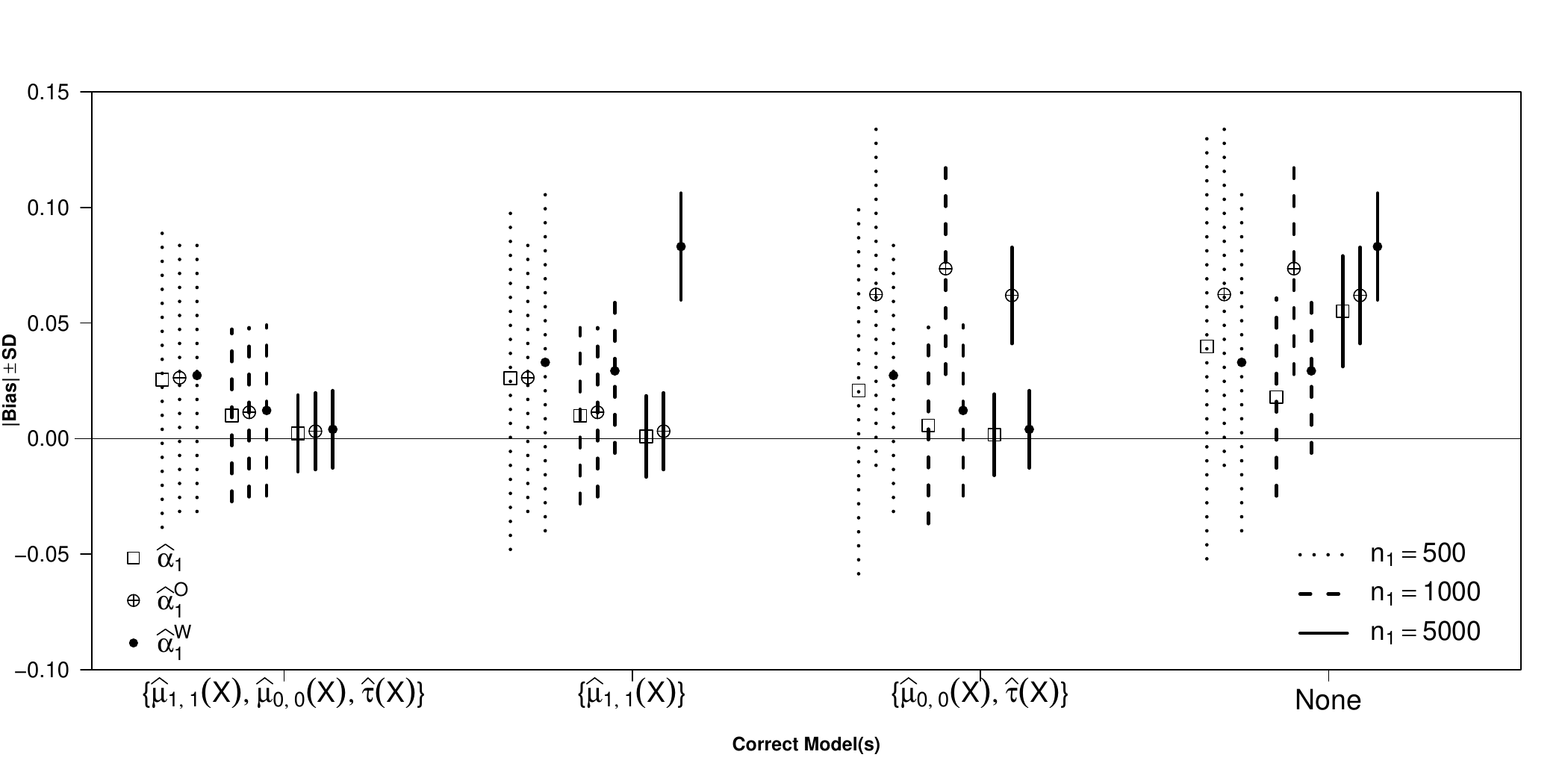}
    \caption{Absolute biases and standard deviations (SD) of the three estimators $\widehat \alpha_{1}, \widehat \alpha_{1}^O$ and $\widehat \alpha_{1}^W$ under different correctly specified nuisance functions and sample sizes.}
    \label{fig:DR_Sim}
\end{figure}

Figure \ref{fig:DR_Sim} shows the averaged absolute biases and standard deviations of the proposed estimator $\widehat \alpha_{1}$, the outcome regression estimator $\widehat \alpha_{1}^O$, and inverse probability weighting (IPW) estimator $\widehat \alpha_{1}^W$ under different sets of correctly specified nuisance functions and sample sizes, where
\begin{align*}
    \widehat \alpha_{1}^O=\widehat\kappa\mathbb{P}_{n}\Big\{\dfrac{\widehat\mu_{1, 1}(X)}{\widehat\mu_{1, 0}(X)}(1-S)Y\Big\}, 
    \qquad\text{and}\qquad
\widehat \alpha_{1}^W=\widehat\kappa\mathbb{P}_{n}\Big\{S\widehat\tau(X)\dfrac{A}{\widehat q(X)}\dfrac{\widehat\mu_{0, 0}(X)}{\widehat\mu_{1, 0}(X)}Y\Big\}.
\end{align*}
The figure shows that, with the increasing sample size of the trial, the biases of $\widehat \alpha_{1}$ converged to zero if either $\{\widehat \mu_{1, 1}(X)\}$ or $\{\widehat \mu_{0, 0}(X), \widehat \tau(X)\}$ was correctly specified. That is, $\widehat \alpha_{1}$ was model doubly robust. In contrast, when $\{\widehat \mu_{1, 1}(X)\}$ is misspecified, $\widehat \alpha_{1}^O$ is not consistent; when $\{\widehat \mu_{0, 0}(X), \widehat \tau(X)\}$ is misspecified, $\widehat \alpha_{1}^W$ is not consistent. As expected, when all the nuisance functions were misspecified, all three estimators were inconsistent.
\begin{figure}[h]
  \begin{subfigure}{0.5\textwidth}
    \centering
    \includegraphics[width=\linewidth]{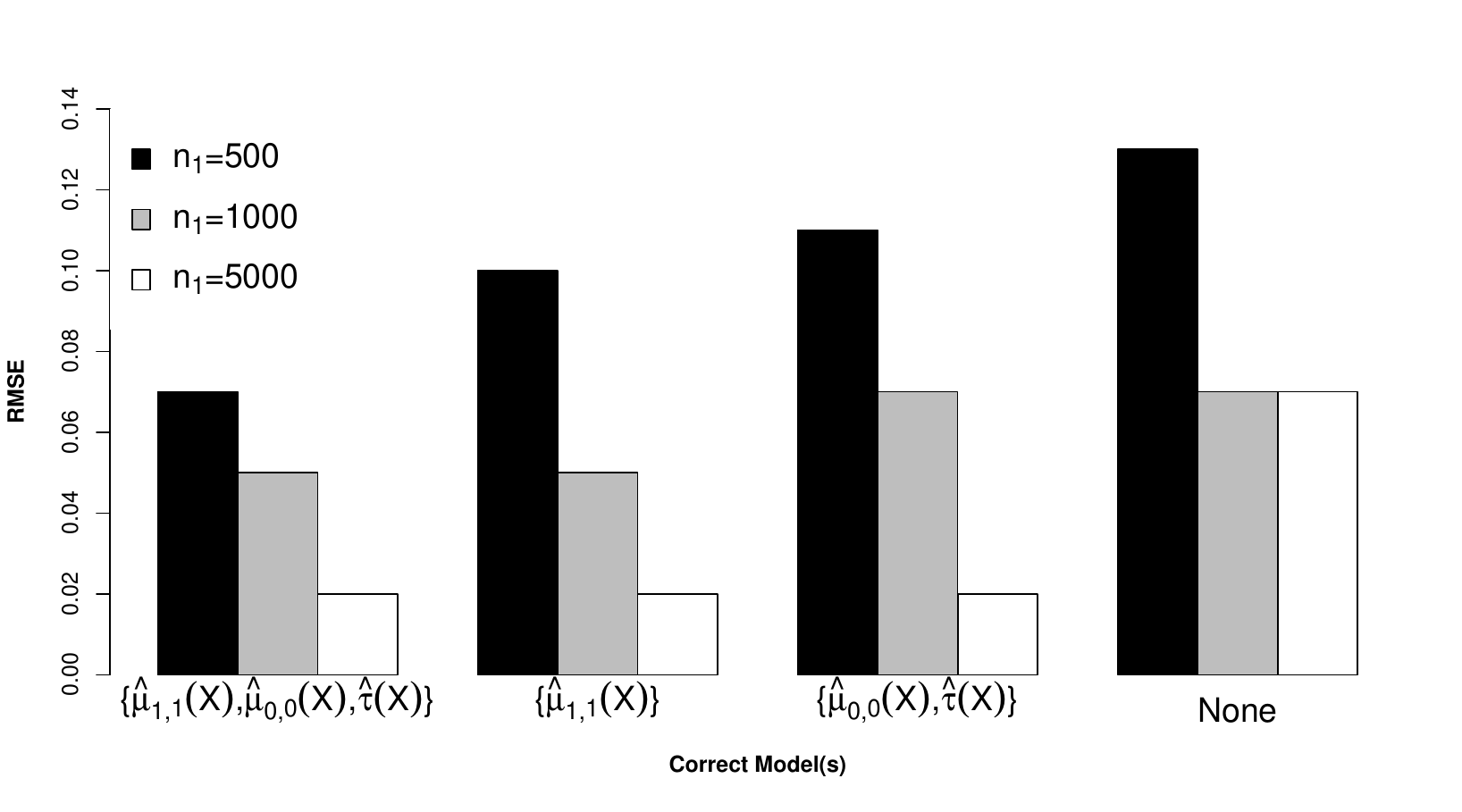}
  \end{subfigure}%
  \begin{subfigure}{0.5\textwidth}
    \centering
    \includegraphics[width=\linewidth]{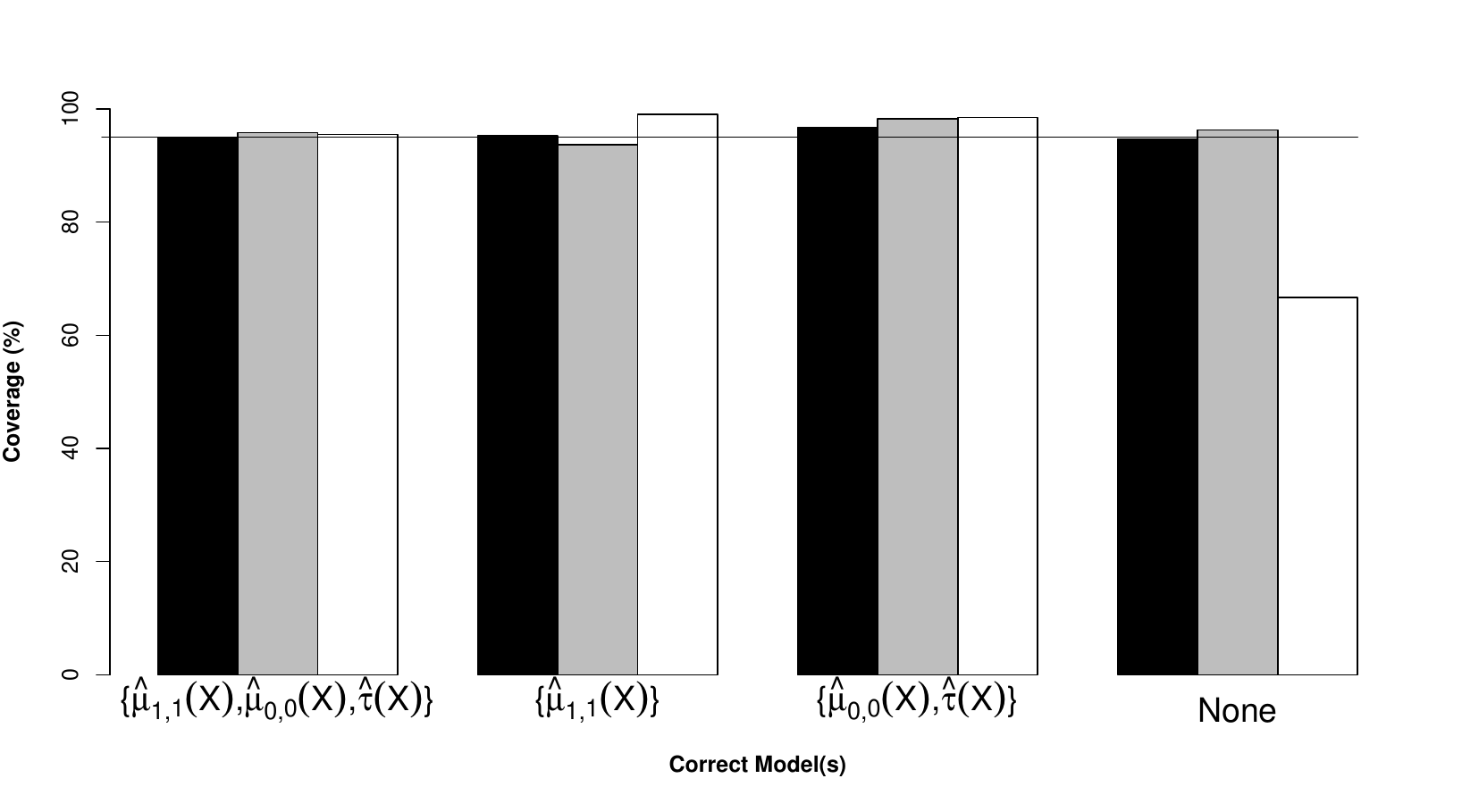}
  \end{subfigure}
  \caption{Scaled RMSE and coverage of $\widehat \alpha_{1}$ under different correctly specified nuisance functions and sample sizes.}
  \label{fig:Sim_efficiency}
\end{figure}

To further assess the efficiency of the proposed estimator $\widehat \alpha_{1}$, we showed, in Figure \ref{fig:Sim_efficiency} the root mean squared error (RMSE) (scaled by $\sqrt{n}$) and coverage of $\widehat \alpha_{1}$ in different settings. The figure showes that when all the nuisance functions were correctly specified, $\widehat \alpha_{1}$ had the least RMSE and was well covered. When one of them was correctly specified, $\widehat \alpha_{1}$ had a larger RMSE, and it could be slightly under or over-covered. When none of the nuisance functions was correctly specified, $\widehat \alpha_{1}$ had the largest RMSE and was significantly under-covered when the sample size of the trial was 5000.
\begin{figure}[h]
    \centering
    \includegraphics[width=0.95\textwidth]{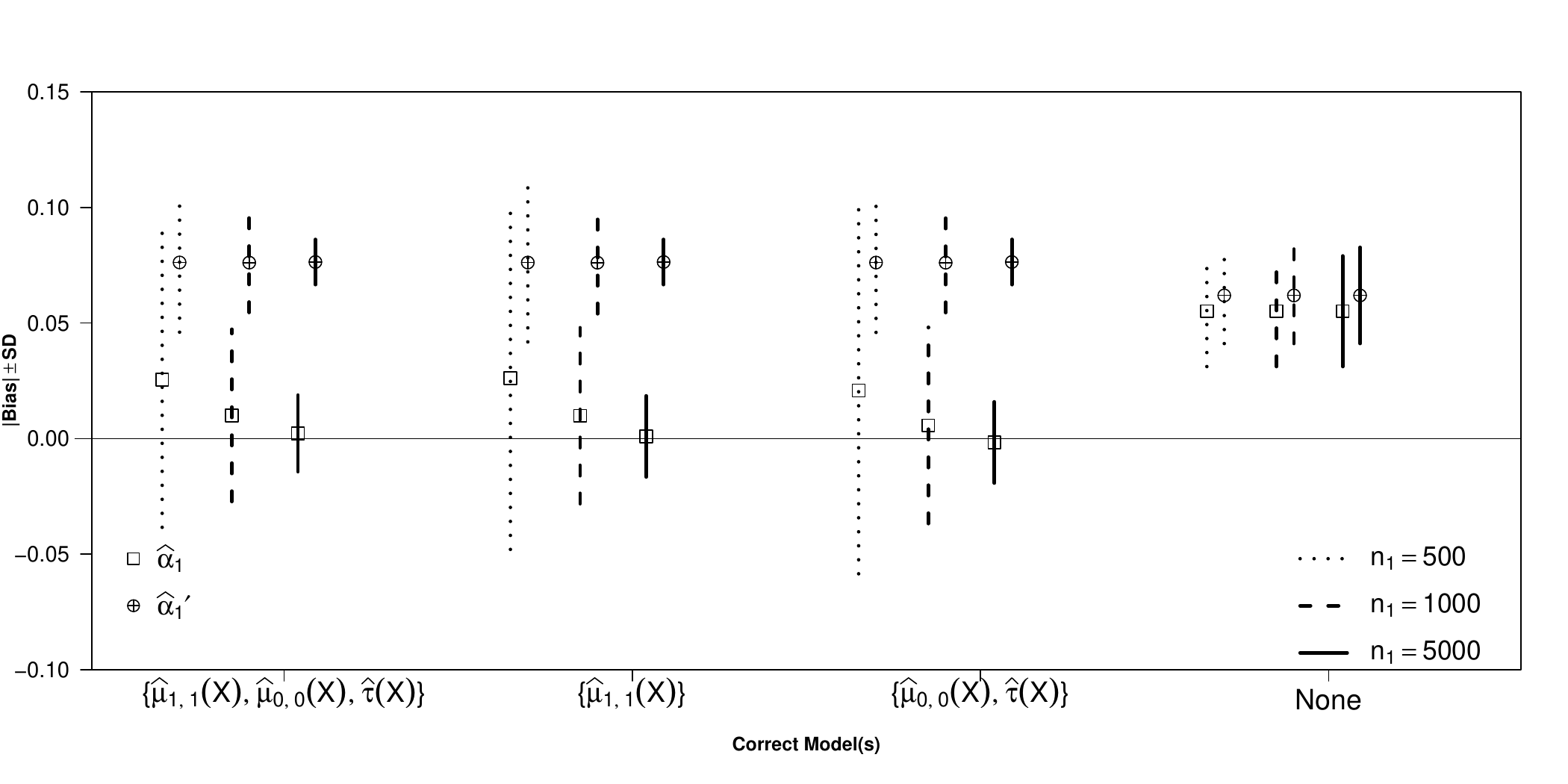}
    \caption{Absolute biases and standard deviations (SD) of the two estimators $\widehat \alpha_{1}$ and $\widehat \alpha_{1}'$ under different correctly specified nuisance functions and sample sizes.}
    \label{fig:Comp_Conv}
\end{figure}

In addition, to evaluate the impact of imposing the condition (A4*) when the data generating mechanism only satisfies condition (A4), we also compared $\widehat \alpha_{1}$ with the efficient estimator when condition (A4*) holds (see Section \ref{sec:A4*} and \ref{appendix:A4*1})
$$
\widehat \alpha_{1}'=
\widehat\kappa\mathbb{P}_{n}\Big[
(1-S)\widehat\mu_{1, 1}(X)+
S\widehat\tau(X)\dfrac{A}{\widehat q(X)}\{Y-\widehat\mu_{1, 1}(X)\}
    \Big].
$$
Figure \ref{fig:Comp_Conv} shows that, even if when all the nuisance functions were correctly specified, $\widehat \alpha'_{1}$ was not consistent. That is, when condition (A4*) does not hold in the data generating mechanism, the estimator that assumes that the condition holds may generate a large bias. In addition, in \ref{appendix:AdditionalSim}, we also compared $\widehat \alpha_{1}$ and $\widehat \alpha'_{1}$ when $\mu_{0, 1}(X)$ was misspecified; both estimators were biased.

\subsection{Assessment of rate robustness}
According to Theorem~\ref{thm:inference1}, $\widehat \alpha_1$ is also rate double robust; whereas simpler outcome regression (plug-in) estimator would generally converge only as fast as the component nuisance function estimates.

In order to illustrate the above advantage, we conducted an additional simulation study. Each of the nuisance functions in this simulation study was ``estimated'' by adding the designed noise to perturb the true underlying data-generating functions so that we could guarantee the convergence rate of the nuisance functions. The details and results are given in \ref{appendix:sim_RateRobustness}.

The results show that the proposed estimator $\widehat \alpha_{1}$ outperformed the outcome regression estimator $\widehat \alpha_{1}^O$, especially when nuisance error was of high order. Assume the underlying nuisance errors are on the order of $O_p(n^{-r})$, the near-optimal performance of the proposed estimator was achieved once $r \geq 0.25$; whereas this was only achieved by the outcome regression estimator\ when $r \approx 0.5$. One would only expect such low nuisance error, however, under correctly specified parametric models for $\widehat{\mu}_{s,a}$, which would be unlikely in practical applications.

%%%%%%%%%%%%%%%%%%%%%%%%%%%%%%%%%%%%%%%%%%%%%%%%%%%%%%%%%%%%%%%%%%%%%%%%%%%%%%%%%%%%
\section{Effects evaluation of treatments for schizophrenia}\label{sec:Application}
%%%%%%%%%%%%%%%%%%%%%%%%%%%%%%%%%%%%%%%%%%%%%%%%%%%%%%%%%%%%%%%%%%%%%%%%%%%%%%%%%%%%

Schizophrenia is a syndrome with psychotic, negative, and cognitive symptoms that generally appear in early adulthood. It is often a chronic, severe, and highly disabling condition that results in significant social and occupational impairment \citep{schizophernia2023nimh}. Researchers often use the Positive and Negative Syndrome Scale (PANSS) total score to evaluate the severity of schizophrenia symptoms. This score ranges from 30 to 200, with higher scores reflecting more severe symptoms \citep{leucht2005does}. For example, patients with PANSS total score over 90 are considered to be ``moderately ill'' \citep{NLM}. The American Psychiatric Association recommends treating schizophrenia with antipsychotic medications \citep{american2020american}. Paliperidone extended-release (ER) is an oral antipsychotic drug approved for the treatment of schizophrenia by both the US Food and Drug Administration and the European Medicines Agency.

\citet{marder2007efficacy} and \citet{davidson2007efficacy} separately evaluated the efficacy and safety of paliperidone ER tablets (3 – 12 mg/day), as compared to the placebo, for treating schizophrenia by two 6-week double-blind trials (NCT00077714 and NCT00668837). To illustrate the proposed methods, we designated one trial \citep{marder2007efficacy} as the index trial (denoted by $S=1$), and used data from the second trial \citep{davidson2007efficacy} as the target population (denoted by $S=0$). We included the patients assigned to either paliperidone ER or placebo and whose PANSS total scores were evaluated at baseline and week 6. In total, 732 patients were included, of which 284 patients were in the target population. We included patient gender, age, race, and PANSS baseline score as covariates. The summary statistics for each variable stratified by the trials is given in \ref{appendix:application1}. 
% The two trials represent different populations. For instance, the population underlying the target population was dominated by White and Black, whereas 31\% of the population underlying the index trial were not White or Black. 

In this study, we aimed to estimate the effects of paliperidone ER tablets in the target population where the outcome was if the patients' PANSS score at week 6 was greater than 90. We used the proposed method and the data from both trials to improve the efficiency of the estimation. This is a special case described in Section \ref{sec:scenario2}, wherein the target population, beyond the control, other treatment (i.e., the experimental treatment in the index trial) was also available. In addition, when the target population is also a trial, it is reasonable to assume that $a=\{0, 1\}, \E(Y^a|S=0, X)=\E(Y^a|S=0, A=a, X)$. This enabled us to identify the target parameters using the target population only. We give the corresponding consistent and efficient estimator in \ref{appendix:application_estimator}. Though this estimator is consistent, because this estimator only uses the data from the target population, we expect this estimator would be less efficient than the proposed estimator. Therefore, we benchmark the proposed estimator with this estimator. We used generalized linear models to model the nuisance functions and estimated marginal counterfactual means, and their difference and ratio in the target population. 
\begin{table}[h]
    \centering
    \begin{tabular}{ccccccc}
    \hline
      \multirow{3}{*}{\textbf{Target parameter}} & \multicolumn{2}{c}{\textbf{Estimate}}  & \multicolumn{2}{c}{\textbf{SD}} & \multicolumn{2}{c}{\textbf{95\% CI}}\\ 
     \cline{2-7}\cline{2-7}
     &\multicolumn{2}{c}{Method}  & \multicolumn{2}{c}{Method} & \multicolumn{2}{c}{Method}\\
     \cline{2-7}
     &1&2&1&2&1&2\\
     \hline
     $\E(Y^{1}|S=0)$               & 0.14 &0.15 &0.03   &0.03 & (0.09, 0.20)& (0.09, 0.21)\\
     $\E(Y^{0}|S=0)$               & 0.31 &0.31 &0.04   &0.05 & (0.22, 0.39)& (0.21, 0.40)\\
     $\E(Y^{1}|S=0)/\E(Y^{0}|S=0)$ & 0.46 &0.49 &0.09   &0.10 & (0.36, 0.68)& (0.31, 0.68)\\
     $\E(Y^{1}|S=0)-\E(Y^{0}|S=0)$ &-0.17 &-0.15&0.04   &0.05 & (-0.25, -0.05)& (-0.26, -0.05)\\
     \hline
    \end{tabular}\\
    \caption{Estimates, standard deviations (SDs), and 95\% confidence intervals (CIs) of the target parameters. Method 1 is described in Section \ref{sec:scenario2}, which uses both trials for the estimation; Method 2, described in \ref{appendix:application_estimator}, only uses the target population for the estimation.}
    \label{tab:ResultsAnalysis}
\end{table}

Table \ref{tab:ResultsAnalysis} shows the estimates, standard deviations, and 95\% confidence intervals of the target parameters with and without transporting the effects estimation from the index trial. The tables shows that the point estimates generated by the proposed estimator closely aligned with those generated from the estimator that only uses the target population. However, the proposed estimator is more efficient. 

Overall, the analyses shows that in the target population, the risk of a patient having a PANSS score higher than 90 at week 6 who took Paliperidone ER was about half of the risk who took placebo; the risk could be reduced 15\%-17\% by taking Paliperidone ER.

%%%%%%%%%%%%%%%%%%%%%%%%%%%%%%%%%%%%%%%%%%%%%%%%%%%%%%%%%%%%%%%%%%%%%%%%%%%%%%%%%%%%
\section{Discussion}\label{sec:Discussion}
%%%%%%%%%%%%%%%%%%%%%%%%%%%%%%%%%%%%%%%%%%%%%%%%%%%%%%%%%%%%%%%%%%%%%%%%%%%%%%%%%%%%

Causal conditions play a pivotal role in identifying and estimating causal target parameters. If these conditions are not met, the resulting estimator becomes difficult to interpret and may be biased for the target parameters. Within the realm of transportability, the relatively common condition often employed for transportability analysis is the transportability condition for difference effect measures. However, sometimes this condition may not align with real-world scenarios. In this work, from the common belief, we propose an alternative condition--the transportability condition of relative effect measure--for identifying and estimating the target parameters. This condition may be more reasonable to make in certain real-world scenarios. With this condition, we develop multiply robust and efficient estimators tailored to three common data structures, which enable the utilization of flexible machine-learning methods for estimating the target parameters.

The work identifies three distinct scenarios related to the presence or absence of confounding in the target population. The first two scenarios deal with situations where confounding is either absent (so modeling the conditional outcome means in the target population is not necessary for identifying the target parameter) or sufficiently controlled (so a nested regression model taking additional confounders into account for the outcomes in the target population is not necessary), thus simplifying the process of identifying target parameters. The third scenario, however, accounts for treatment variation and additional confounders, necessitating more complex models. The three scenarios require different causal conditions for identifying the target parameters, resulting in different conditions for deriving a robust and efficient estimator. 

% Throughout the course of this development, numerous intriguing statistical challenges surfaced and were elucidated. For instance, 1) when it is necessary to model the outcomes in the target population, 2) the relationship of convergence requirements for the ratio of two outcome models and the ones for the two outcome models, 3) the interplay of convergence rates of the nested and outer regression models, 4) the connection and difference between parametric model robustness and rate robustness, and 5) the comprehension of multiple robustness. These challenges added depth and complexity to our statistical inquiry, shedding light on the nuanced intricacies inherent in our research pursuits.

Looking ahead, it would be beneficial to develop a hybrid estimator based on the hypothesis testing of the validity of the transportability condition for difference effect measures, where the estimator is our proposed estimator when the test fails. We leave this as future work. Expanding this work to include survival outcomes and investigating heterogeneous treatment effects based on the proposed transportability condition are also identified as valuable future directions.

%%%%%%%%%%%%%%%%%%%%%%%%%%%%%%%%%%%%%%%%%%%%%%%%%%%%%%%%%%%%%%%%%%%%%%%%%%%%%%%%%%%%
\section*{Acknowledgments}
%%%%%%%%%%%%%%%%%%%%%%%%%%%%%%%%%%%%%%%%%%%%%%%%%%%%%%%%%%%%%%%%%%%%%%%%%%%%%%%%%%%%

This work was supported in part by Patient-Centered Outcomes Research Institute (PCORI) awards ME-2021C2-22365. The content is solely the responsibility of the authors and does not necessarily represent the official views of PCORI, PCORI's Board of Governors or PCORI's Methodology Committee. \\
This study, carried out under The Yale University Open Data Access  (YODA) \url{https://yoda.yale.edu/} Project number 2022-5062, used data obtained from the Yale University Open Data Access Project, which has an agreement with Jassen Research \& Development, L.L.C.. The interpretation and reporting of research using this data are solely the responsibility of the authors and do not necessarily represent the official views of the Yale University Open Data Access Project or Jassen Research \& Development, L.L.C..
\vspace*{-8pt}

%%%%%%%%%%%%%%%%%%%%%%%%%%%%%%%%%%%%%%%%%%%%%%%
\section*{Data availability statement}
%%%%%%%%%%%%%%%%%%%%%%%%%%%%%%%%%%%%%%%%%%%%%%%

Code to reproduce our simulations and the data analyses are provided on GitHub: \\\href{https://github.com/Guanbo-W/Transportability_RRcondition}{https://github.com/Guanbo-W/Transportability\_RRcondition}. \\
The data analyzed can be obtained from the YODA, subject to approval by the YODA \url{https://yoda.yale.edu/}

\bigskip
\begin{center}
{\large\bf Supplementary Materials}
\end{center}

\begin{description}

\item Supplementary materials contain all the proofs and related theorems of this paper.

\end{description}

%%%%%%%%%%%%%%%%%%%%%%%%%REFERENCES%%%%%%%%%%%%%%%%%%%%%%%%%
%%%%%%%%%%%%%%%%%%%%%%%%%%%%%%%%%%%%%%%%%%%%%%%%%%%%%%%%%%%%
\bibliographystyle{apalike}
\bibliography{refs}

\clearpage
\textbf{\large\bf Appendices for ``Efficient estimation of subgroup treatment effects using multi-source data''}
\setcounter{page}{1}
\renewcommand{\thepage}{S\arabic{page}}
\begin{appendices}
\numberwithin{equation}{section}
\renewcommand{\theequation}{\thesection.\arabic{equation}}
\numberwithin{table}{section}
\renewcommand{\thetable}{\thesection.\arabic{table}}
\numberwithin{figure}{section}
\renewcommand{\thefigure}{\thesection.\arabic{figure}}
\numberwithin{assumption}{section}
\renewcommand{\theassumption}{\thesection.\arabic{assumption}}
\numberwithin{lemma}{section}
\renewcommand{\thelemma}{\thesection.\arabic{lemma}}
\numberwithin{prop}{section}
\renewcommand{\theprop}{\thesection.\arabic{prop}}

\renewcommand{\thesection}{Appendix \Alph{section}}
\renewcommand{\thesubsection}{\Alph{section}.\arabic{subsection}}

\titleformat{\subsection}
      {\normalfont\fontsize{10}{12}\bfseries}{\thesubsection}{1em}{}
      
% \renewcommand{\figurename}{Web Figure}
% \addto\captionsenglish{%
% \renewcommand{\tablename}{Web Table}
% }
% \captionsetup[table]{name=Web Tabl}

\renewcommand{\thetable}{S\arabic{table}}
\renewcommand{\thefigure}{S\arabic{figure}}
\renewcommand{\theequation}{S\arabic{equation}}
\renewcommand{\thetheorem}{S\arabic{theorem}}
\renewcommand{\thelemma}{S\arabic{lemma}}
\renewcommand{\theprop}{S\arabic{prop}}
%%%%%%%%%%%%%%%%%%%%%%%%%%%%%%%%%%%%%%%%%%%%%%%%%%%%%%%%%%%%
\section{Implications of the condition (A4)}\label{appendix:CausalStructure}
%%%%%%%%%%%%%%%%%%%%%%%%%%%%%%%%%%%%%%%%%%%%%%%%%%%%%%%%%%%%

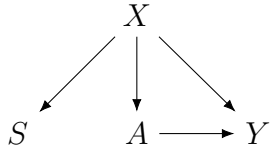
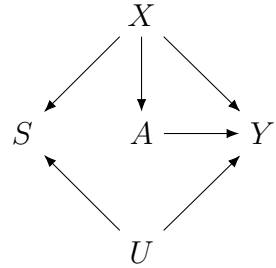
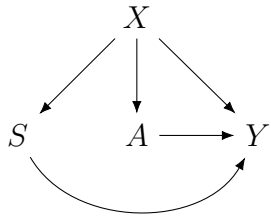
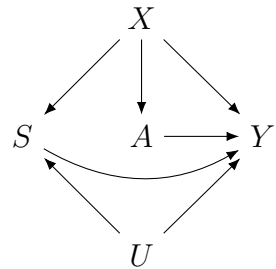
\begin{figure}[hbt!]
\begin{minipage}[h]{0.47\linewidth}
\begin{center}
  \begin{tikzpicture}
    % x node set with absolute coordinates
    \node[] (A) at (0,0) {$A$};
    \node[] (Y) [right =of A] {$Y$};
    \node[] (X) [above =of A] {$X$};
    \node[] (S) [left =of A] {$S$};
    \node[] (0) at (0,-1.8) {};
    % Directed edge
    \path (A) edge (Y);
    \path (X) edge (Y);
    \path (X) edge (A);
    \path (X) edge (S);
\end{tikzpicture}
  \subcaption{DAG that satisfies (A4*)}
  \label{fig:DAGA4*}
\end{center} 
\end{minipage}
\hfill
\vspace{0.2 cm}
\begin{minipage}[h]{0.47\linewidth}
\begin{center}
  \begin{tikzpicture}
    % x node set with absolute coordinates
    \node[] (A) at (0,0) {$A$};
    \node[] (Y) [right =of A] {$Y$};
    \node[] (X) [above =of A] {$X$};
    \node[] (S) [left =of A] {$S$};
    \node[] (U) [below =of A] {$U$};
    % Directed edge
    \path (A) edge (Y);
    \path (X) edge (Y);
    \path (X) edge (A);
    \path (X) edge (S);
    \path (U) edge (S);
    \path (U) edge (Y);
\end{tikzpicture}
  \subcaption{DAG that satisfies (A4), case 1}
  \label{fig:DAGA4Case1}
\end{center}
\end{minipage}
\vfill
\vspace{0.2 cm}
\begin{minipage}[h]{0.47\linewidth}
\begin{center}
   \begin{tikzpicture}
    % x node set with absolute coordinates
    \node[] (A) at (0,0) {$A$};
    \node[] (Y) [right =of A] {$Y$};
    \node[] (X) [above =of A] {$X$};
    \node[] (S) [left =of A] {$S$};
    \node[] (0) at (0,-1.8) {};
    % Directed edge
    \path (A) edge (Y);
    \path (X) edge (Y);
    \path (X) edge (A);
    \path (X) edge (S);
    \path (S) edge[bend right=60] (Y);
\end{tikzpicture}
  \subcaption{DAG that satisfies (A4), case 2}
  \label{fig:DAGA4Case2}
\end{center}
\end{minipage}
\hfill
\begin{minipage}[h]{0.47\linewidth}
\begin{center}
      \begin{tikzpicture}
    % x node set with absolute coordinates
    \node[] (A) at (0,0) {$A$};
    \node[] (Y) [right =of A] {$Y$};
    \node[] (X) [above =of A] {$X$};
    \node[] (S) [left =of A] {$S$};
    \node[] (U) [below =of A] {$U$};
    % Directed edge
    \path (A) edge (Y);
    \path (X) edge (Y);
    \path (X) edge (A);
    \path (X) edge (S);
    \path (U) edge (S);
    \path (U) edge (Y);
    \path (S) edge[bend right=30] (Y);
\end{tikzpicture}
  \subcaption{DAG that satisfies (A4), case 3}
  \label{fig:DAGA4Case3}
\end{center}
\end{minipage}
\caption{A figure with four subfigures}
\label{fig:roc}
\end{figure}

\begin{table}[h]
\footnotesize
    \centering
    \begin{tabular}{ccccc}
    \hline
  Case   &Quantity& &$a=0$ & $a=1$ \\\\
      \hline\hline\\
       1  &$\text{log}\{\E(Y^a|X, U)\}$& &$m(X)+f(U)g_{1}(X)$ & $m(X)+f(U)g_{1}(X)+r(X)$\\\\\\
       \multirow{3}{*}{2} & \multirow{3}{*}{$\text{log}\{\E(Y^{s, a}|X)\}$} &s=0 & $m(X)$ & $m(X)+r(X)$\\\\
        && s=1 & $m(X)+g_{2}(X)$ & $m(X)+g_{2}(X)+r(X)$\\\\
       \multirow{3}{*}{3}& \multirow{3}{*}{$\text{log}\{\E(Y^{s, a}|X, U)\}$} & s=0 & $m(X)+f(U)g_{1}(X)$ & $m(X)+f(U)g_{1}(X)+r(X)$\\\\
        && s=1& $m(X)+f(U)g_{1}(X)+g_{2}(X)$ & $m(X)+f(U)g_{1}(X)+g_{2}(X)+r(X)$\\\\
        \hline
       \end{tabular}
    \caption{Hypothetical data generating mechanisms for a binary outcome.}
    \label{tab:DGM}
\end{table}
Figures \ref{fig:DAGA4*} and \ref{fig:DAGA4Case1} show the Directed Acyclic Graphs (DAGs) that are allowed by the condition (A4*) and (A4), respectively. Condition (A4) allows unmeasured confounders to be a common cause of $S$ and $Y$, which is not allowed by condition (A4*). 

However, there is a restriction in Figures \ref{fig:DAGA4Case1} to satisfy condition (A4). The restriction is that condition (A4) requires that the impact of the unmeasured confounder $U$ on the outcome in the trial and the target population must be canceled in a multiplicative way, in both the trial and the target population. 

To illustrate this restriction more concretely, we present a possible data-generating mechanism (DGM) that aligns with
Figures \ref{fig:DAGA4Case1} and satisfies condition (A4) but not condition (A4*). Suppose that the outcome of interest is binary; consider case 1 in Table \ref{tab:DGM}. Let $S=h(X, U)$ and denote $U_{0}(x)$ and $U_{1}(x)$ as the solutions to $0=h(x, U)$ and $1=h(x, U)$ for a given $x$, respectively. The DGM (case 1 in Table \ref{tab:DGM}) aligns with Figures \ref{fig:DAGA4Case1} because $U$ causes $X$ and $Y$. It also satisfies (A4) because
\begin{align*}
&\dfrac{\E ( Y^1 | X = x , S = 1)}{\E ( Y^0 | X = x , S = 1)} = \dfrac{\exp [m(x)+f\{U_{1}(x)\}g_{1}(x)+r(x)]}{\exp [m(x)+f\{U_{1}(x)\}g_{1}(x)]} =\exp\{r(x)\},\\
& \dfrac{\E ( Y^1  | X = x , S = 0)}{\E ( Y^0 | X = x , S = 0)} = \dfrac{\exp [m(x)+f\{U_{0}(x)\}g_{1}(x)+r(x)]}{\exp [m(x)+f\{U_{0}(x)\}g_{1}(x)]}=\exp\{r(x)\},\\
\implies
& \dfrac{\E ( Y^1 | X = x , S = 1)}{\E ( Y^0 | X = x , S = 1)} = \dfrac{\E ( Y^1  | X = x , S = 0)}{\E ( Y^0 | X = x , S = 0)} =\exp\{r(x)\}.
\end{align*}
That is, condition (A4) allows unmeasured confounders in Figure \ref{fig:DAGA4Case1} such that in the same population (trial or target), maybe together with $X$, its impact on $Y^1$ and $Y^0$ are the same. Such an unmeasured $U$ can be, for example, adherence.

In addition, condition (A4) also allows for the causal structure where $S$ is a direct cause of $Y$, which is shown in Figure \ref{fig:DAGA4Case2}. To see the reasons, we first introduce the counterfactual outcomes $Y^{s, a}$ where $Y$ is allowed to be intervened by $S$. Next, we introduce two sets of conditions that are equivalent to (imply and can be implied by) condition (A4). That is, $\{(A4a), (A4b)\} \Leftrightarrow (A4) \Leftrightarrow \{(A4c), (A4d)\}$.

First, consider

\noindent
\emph{(A4a):}  $\E ( Y^{s=0, a=0} | X = x) \neq 0,   \E ( Y^{s=1, a=0} | X = x) \neq 0,$ and
\begin{equation*}
        \dfrac{\E ( Y^{s=0, a=1} | X = x )}{\E ( Y^{s=0, a=0} | X = x)} = \dfrac{\E ( Y^{s=1, a=1}  | X = x )}{\E ( Y^{s=1, a=0} | X = x )}.
\end{equation*} 
\noindent

\noindent
\emph{(A4b):}  $\E ( Y^{s, a} |X = x) = \E ( Y^{a} | X = x, S=s)$.

\noindent
In terms of identification of the target parameters, we can replace the condition (A4) with the conditions (A4a) and (A4b). Now let us see what does (A4a) imply via a possible DGM, case 2 in Table \ref{tab:DGM}. In such a case,
\begin{align*}
    &\dfrac{\E ( Y^{s=0, a=1} | X = x )}{\E ( Y^{s=0, a=0} | X = x)}=
    \dfrac{\exp\{ m(x)+r(x)\}}{\exp\{ m(x)\}}=\exp\{ r(x)\},\\
    &\dfrac{\E ( Y^{s=1, a=1} | X = x )}{\E ( Y^{s=1, a=0} | X = x)}=
    \dfrac{\exp\{ m(x)+g_{2}(x)+r(x)\}}{\exp\{ m(x)g_{2}(x)\}}=\exp\{ r(x)\},\\
    \implies
    &\dfrac{\E ( Y^{s=0, a=1} | X = x )}{\E ( Y^{s=0, a=0} | X = x)}=
    \dfrac{\E ( Y^{s=1, a=1} | X = x )}{\E ( Y^{s=1, a=0} | X = x)}=\exp\{ r(x)\},
\end{align*}
which satisfies the condition (A4a).  That is, the DGMs that featured the following characteristics can satisfy the condition (A4a), it allows $X$ to affect $Y^a$ differently in different populations (trial and target), as long as in a same population, $X$ affects $Y^1$ and $Y^0$ to the same extent in the multiplicative fashion. The DAG in Figure \ref{fig:DAGA4Case2} is compatible with such DGMs. Thus, we conclude that, the condition (A4) allows $S$ to be a direct causa of $Y$.

On the other hand, consider the other set of conditions that is equal to condition (A4).

\noindent
\emph{(A4c):}  $\E ( Y^{s=0, a=0} | X = x, S=0) \neq 0,   \E ( Y^{s=1, a=0} | X = x, S=1) \neq 0,$ and
\begin{equation*}
        \dfrac{\E ( Y^{s=0, a=1} | X = x, S=0 )}{\E ( Y^{s=0, a=0} | X = x, S=0 )} = \dfrac{\E ( Y^{s=1, a=1}  | X = x, S=1)}{\E ( Y^{s=1, a=0} | X = x, S=1)}.
\end{equation*} 
\noindent

\noindent
\emph{(A4d):}  $Y^{s, a}=Y^{a}$ if $S=s$.
\noindent

Note that, for identification, the condition (A4c) is enough, whereas both (A4a) and (A4b) are required. The DGM depicted in case 2 in Table \ref{tab:DGM} also satisfies the condition (A4c). Apply the same logic, we arrive at the same conclusion, the condition (A4) allows $S$ to be a direct causa of $Y$.

Naturally, the DAG presented in Figure \ref{fig:DAGA4Case3} and its associated possible DGM, the case 3 in Table \ref{tab:DGM} satisfied (A4) as well.
\clearpage

%%%%%%%%%%%%%%%%%%%%%%%%%%%%%%%%%%%%%%%%%%%%%%%%%%%%%%%%%%%%
\section{Proof of Theorem \ref{thm:identification1} and other identification stratities}\label{appendix:identification1}
%%%%%%%%%%%%%%%%%%%%%%%%%%%%%%%%%%%%%%%%%%%%%%%%%%%%%%%%%%%%
% \begin{restatable}[]{theorem}{identification1}
% \label{thm:identification1}
% When all the target population can only access to control, under conditions (A1) through (A5), $\alpha_{1}=\E(Y^1 | S = 0)$ can be identified as following estimable functions
% \begin{align*} 
%     \alpha_{1}
%     =&\E \left\{ \dfrac{\E(Y | X, S = 1, A = 1)}{\E(Y | X, S = 1, A = 0)} \E(Y | X, S = 0) \Big | S = 0 \right\}\\
%     =&\E \left\{ \dfrac{\E(Y | X, S = 1, A = 1)}{\E(Y | X, S = 1, A = 0)} Y \Big | S = 0 \right\}
% \end{align*}
% \end{restatable}
We will now show that under conditions (A1) through (A6), listed in the main text, the target parameters can be identified by $\alpha_1$.
\begin{proof}
\begin{align*}
    \E(Y^1 | S = 0) &= \dfrac{1}{\Pr(S=0)}\E[Y^1 I(S=0)]\\
    &=\dfrac{1}{\Pr(S=0)}\E\big[\E\{Y^1 I(S=0)|X\}\big]\\
    &=\dfrac{1}{\Pr(S=0)}\E\big[\E(Y^1|X, S=0)\Pr(S=0|X)\big]\\
    &=\dfrac{1}{\Pr(S=0)}\E\left[\dfrac{\E(Y^1 | X, S = 0)}{\E(Y^0 | X, S = 0)} \E(Y^0 | X, S = 0) \Pr(S=0)\right]\\
    &= \dfrac{1}{\Pr(S=0)}\E \left[ \dfrac{\E(Y^1 | X, S = 1)}{\E(Y^0 | X, S = 1)} \E(Y^0 | X, S = 0) \Pr(S=0)\right]\\ 
    &= \dfrac{1}{\Pr(S=0)}\E \left[ \dfrac{\E(Y^1 | X, S = 1, A = 1)}{\E(Y^0 | X, S = 1, A = 0)} \E(Y^0 | X, S = 0) \Pr(S=0)\right] \\
    &=\dfrac{1}{\Pr(S=0)}\E \left[\dfrac{\E(Y | X, S = 1, A = 1)}{\E(Y | X, S = 1, A = 0)} \E(Y | X, S = 0) \Pr(S=0|X)\right]\\
    &=\dfrac{1}{\Pr(S=0)}\E\left[\dfrac{\E(Y | X, S = 1, A = 1)}{\E(Y | X, S = 1, A = 0)} \E\{YI(S=0)|X\}\right]\\
    &=\dfrac{1}{\Pr(S=0)}\E\left[\dfrac{\E(Y | X, S = 1, A = 1)}{\E(Y | X, S = 1, A = 0)} YI(S=0)\right]\\
    &=\E\left[\dfrac{\E(Y | X, S = 1, A = 1)}{\E(Y | X, S = 1, A = 0)} Y\Big|S=0 \right]\\
    &=\alpha_1.
\end{align*}
The fourth equation follows from multiplying and dividing by $\E[Y^0 | X, S = 0]$ which is allowed by the positivity component of condition (A4); the fifth from the exchangeability component of condition (A4); the sixth from condition (A2), which allows us to write $$\dfrac{\E(Y^1 | X =x, S = 1)}{\E(Y^0 | X=x, S = 1)} = \dfrac{\E(Y^1 | X=x, S = 1, A = 1)}{\E(Y^0 | X=x, S = 1, A = 0)}, \mbox{ for all } x \mbox{ with } f(x, S = 0) \neq 0;$$ and the seventh step follows because conditions (A1) and the special case considered in this section, $S=0\implies A=0$ allow us to write 
\begin{align*}
&\E(Y^0 | X=x, S = 0) = \E(Y^0 | X=x, S = 0, A = 0) \\= &\E(Y | X=x, S = 0, A = 0) = \E(Y | X=x, S = 0),
\end{align*}
and condition (A1) allows us to write $$\dfrac{\E(Y^1 | X = x, S = 1, A = 1)}{\E(Y^0 | X=x, S = 1, A = 0)} = \dfrac{\E(Y | X=x, S = 1, A = 1)}{\E(Y | X=x, S = 1, A = 0)},$$ for all $x$ with $f(x, S = 0) \neq 0$.
\end{proof}

Next, we give other identification strategies in the following propositions.
\begin{prop}
Under conditions (A1)-(A6), 
\begin{align*}
    \alpha_1=&\E\left[\dfrac{\E(Y | X, S = 1, A = 1)}{\E(Y | X, S = 1, A = 0)} \E(Y | X, S = 1, A = 1)\Big|S=0 \right]\\
    =&\dfrac{1}{\Pr(S=0)}\E \left[ \dfrac{\Pr(S=0|X)}{\Pr(S=1|X)}\dfrac{\E(Y | X, S = 0) }{\E(Y | X, S = 1, A = 0)} \dfrac{SAY}{\Pr(A=1|X, S=1)} \right]\\
    =&\dfrac{1}{\Pr(S=0)}\E\left[ \dfrac{\E(Y | X, S = 1, A = 1)}{\E(Y | X, S = 1, A = 0)} \dfrac{\E(Y | X, S = 0)}{\E(Y | X, S = 1, A = 0)} \dfrac{\Pr(S=0|X)}{\Pr(S=1|X)}\dfrac{S(1-A)Y}{\Pr(A=0|X, S=1)} \right].
\end{align*}

\end{prop}
\begin{proof}
Equation 1 holds because
\begin{align*}\label{eq:id1}
        \E(Y^1 | S = 0) &= \E \big\{ \E(Y^1 | X, S = 0) \big | S = 0 \big\} \\ 
         &= \E \left\{ \dfrac{\E(Y^1 | X, S = 0)}{\E(Y^0 | X, S = 0)} \E(Y^0 | X, S = 0) \Big | S = 0 \right\} \\ 
         &= \E \left\{ \dfrac{\E(Y^1 | X, S = 1)}{\E(Y^0 | X, S = 1)} \E(Y^0 | X, S = 0) \Big | S = 0 \right\} \\ 
          &= \E \left\{ \dfrac{\E(Y^1 | X, S = 1, A = 1)}{\E(Y^0 | X, S = 1, A = 0)} \E(Y^0 | X, S = 0) \Big | S = 0 \right\} \\ 
          &= \E \left\{ \dfrac{\E(Y | X, S = 1, A = 1)}{\E(Y | X, S = 1, A = 0)} \E(Y | X, S = 0) \Big | S = 0 \right\}.
          % \\&= \E \left\{ r_{A}(X)\mu_{0, 0}(X) \Big | S = 0 \right\}.
\end{align*}
Equation 2 holds because
\begin{align*}
        \E(Y^1 | S = 0) 
          &= \E \left\{ \dfrac{\E(Y | X, S = 1, A = 1)}{\E(Y | X, S = 1, A = 0)} \E(Y | X, S = 0) \Big | S = 0 \right\}\\
          &=\E \left\{ \dfrac{\E(Y | X, S = 0) }{\E(Y | X, S = 1, A = 0)} \E\Big\{\dfrac{SAY}{\Pr(S=1|X)\Pr(A=1|X, S=1)} \Big|  X\Big\} \Big | S = 0 \right\}\\
          &=\dfrac{1}{\Pr(S=0)}\E \left\{ (1-S)\dfrac{\E(Y | X, S = 0) }{\E(Y | X, S = 1, A = 0)} \E\Big\{\dfrac{SAY}{\Pr(S=1|X)\Pr(A=1|X, S=1)} \Big|  X\Big\} \right\}\\
          &=\dfrac{1}{\Pr(S=0)}\E \left\{ \dfrac{\Pr(S=0|X)}{\Pr(S=1|X)}\dfrac{\E(Y | X, S = 0) }{\E(Y | X, S = 1, A = 0)} \dfrac{SAY}{\Pr(A=1|X, S=1)} \right\}.
\end{align*}
Equation 3 holds because
\begin{align*}
        \E(Y^1 | S = 0) 
          =& \E \left\{ \dfrac{\E(Y | X, S = 1, A = 1)}{\E(Y | X, S = 1, A = 0)} \dfrac{\E(Y | X, S = 0)}{\E(Y | X, S = 1, A = 0)}\E(Y | X, S = 1, A=0) \Big | S = 0 \right\}\\
          =&\E \Big\{ \dfrac{\E(Y | X, S = 1, A = 1)}{\E(Y | X, S = 1, A = 0)} \dfrac{\E(Y | X, S = 0)}{\E(Y | X, S = 1, A = 0)} \\
          &\E\Big\{\dfrac{S(1-A)Y}{\Pr(S=1|X)\Pr(A=0|X, S=1)} \Big|  X\Big\} \Big | S = 0 \Big\}\\
          =&\dfrac{1}{\Pr(S=0)}\E \Big\{ (1-S)\dfrac{\E(Y | X, S = 1, A = 1)}{\E(Y | X, S = 1, A = 0)} \dfrac{\E(Y | X, S = 0)}{\E(Y | X, S = 1, A = 0)} \\
          &\E\Big\{\dfrac{S(1-A)Y}{\Pr(S=1|X)\Pr(A=0|X, S=1)} \Big|  X\Big\} \Big\}\\
          =&\dfrac{1}{\Pr(S=0)}\E \Big\{ \dfrac{\E(Y | X, S = 1, A = 1)}{\E(Y | X, S = 1, A = 0)} \dfrac{\E(Y | X, S = 0)}{\E(Y | X, S = 1, A = 0)}\\
          & \dfrac{\Pr(S=0|X)}{\Pr(S=1|X)}\dfrac{S(1-A)Y}{\Pr(A=0|X, S=1)} \Big\}.
\end{align*}
\end{proof}

Since $S=0\implies A=0$, the outcomes in the target population are unconfounded, that is, $Y^0=Y$ in $\{S=0\}$. Therefore, $\E(Y^0|S=0)=\E(Y|S=0)$.
\clearpage

\section{Proof of Lemma \ref{thm:IF1}}\label{appendix:IF1}
%%%%%%%%%%%%%%%%%%%%%%%%%%%%%%%%%%%%%%%%%%%%%%%%%%%%%%%%%%%%

% \begin{restatable}[]{theorem}{IF}
% \label{thm:IF1}
% The influence functions of $\alpha_{1}$, $A(O, \boldsymbol{\eta})$ is
% \begin{align*}
%     \kappa\Bigg(
%     (1-S)
%     \{r_A(X)Y-\alpha_{1}\}+
%     S\tau(X)r_S(X)\Big[
%     \dfrac{A}{q(X)}\{Y-\mu_{1, 1}(X)\}-\dfrac{1-A}{1-q(X)}r_A(X)\{Y-\mu_{1, 0}(X)\}
%     \Big]
%     \Bigg).
% \end{align*}
% It is also the semiparametric efficient influence function of $\alpha_{1}$ when the propensity score in the target population $\Pr(A=1|X, S=0)$ is known.
% \end{restatable}
\begin{proof}
Recall that under the nonparametric model, $\mathcal M_{\text{\tiny np}}$, for the observable data $O$, the density of the law of the observable data can be written as 
\begin{equation*}
    p(r,x,a,y) = p(s) p(x|s) p(a|x, s) p(y|a, x, s).
\end{equation*}
Under this model, the tangent space is the Hilbert space of mean zero random variables with finite variance. It can be decomposed as $L_2^0 = \Lambda_{S} \oplus \Lambda_{X|S} \oplus   \Lambda_{A|S,X} \oplus  \Lambda_{Y|S, X, A} $. We now derive the influence function for $\alpha_{1} $ under this nonparametric model. 

\noindent
Recall that
\begin{equation*} 
  \alpha_{1}  \equiv \E \left\{ \dfrac{\E(Y | X, S = 1, A = 1)}{\E(Y | X, S = 1, A = 0)} \E(Y | X, S = 0) \Big | S = 0 \right\}=\E\Big\{\dfrac{\mu_{1, 1}(X)}{\mu_{1, 0}(X)}\mu_{0, 0}(X)\Big|S=0\Big\}.
\end{equation*}
We will use the path differentiability of $\alpha_{1} $ to obtain the efficient influence function under the non-parametric model for the observed data \citep{bickel1993efficient}. To do so, we examine the derivative of $\alpha_{1, p_t}(X)$ with respect to $t$; where the subscript $p_t$ denotes the dependence of $\alpha_{1} $ on a one-dimensional parametric sub-model $p_t$, indexed by $t \in [0,1)$, with $t = 0$ denoting the ``true'' data law. 
% Furthermore, denote the function $u_{p}(\cdot)$ as the score function of the observed data with the (conditional) likelihood $p$.\\

Utilizing the law of total expectation, the tricks of $\E\big\{\E(u|w)v\big\}=\E\big\{\E(v|w)u\big\}$, subtracting the mean zero terms, and $\E(A|B)=\E\Big\{\dfrac{I(B=b)}{\Pr(B=b)}A\Big\}$ when $B$ is discrete repeatedly, we have
\begin{align*}
     &\dfrac{\partial \alpha_{1, p_t}}{\partial t}\Big|_{t=0}\\
    =&\dfrac{\partial}{\partial t}\E_{p_t}\Bigg\{\dfrac{\mu_{1, 1, p_t}(X)}{\mu_{1, 0, p_t}(X)}\mu_{0, p_t}(X)\Big|S=0\Bigg\}\Big|_{t=0}\\
    =&\underbrace{\dfrac{\partial}{\partial t}\E_{p_{t}}\Bigg\{\dfrac{\mu_{1, 1, p_0}(X)}{\mu_{1, 0, p_0}(X)}\mu_{0 , p_0}(X)\Big|S=0\Bigg\}\Big|_{t=0}}_{(1)}+
    \underbrace{\E_{p_{0}}\Bigg\{\dfrac{\dfrac{\partial}{\partial t}\mu_{1, 1, p_t}(X)\Big|_{t=0}}{\mu_{1, 0, p_0}(X)}\mu_{0 , p_0}(X)\Big|S=0\Bigg\}}_{(2)}+\\
    &\underbrace{\E_{p_{0}}\Bigg\{\dfrac{\mu_{1, 1, p_0}(X)}{\dfrac{\partial}{\partial t}\mu_{1, 0, p_t}(X)\Big|_{t=0}}\mu_{0 , p_0}(X)\Big|S=0\Big\}}_{(3)}+
    \underbrace{\E_{p_{0}}\Big\{\dfrac{\mu_{1, 1, p_0}(X)}{\mu_{1, 0, p_0}(X)}\dfrac{\partial}{\partial t}\mu_{0, p_t}(X)\Big|_{t=0}\Big|S=0\Bigg\}}_{(4)}.
\end{align*}
Next, we calculate each term by turn. Denote the score function by $g(\cdot)$.
\begin{align*}
    (1)=&\E_{p_{0}}\Big\{\dfrac{1-S}{\Pr(S=0)}\dfrac{\mu_{1, 1, p_0}(X)}{\mu_{1, 0, p_0}(X)}\mu_{0 , p_0}(X) g_{X|S}(O)\Big\}\\
    =&\E_{p_{0}}\Big\{\dfrac{1-S}{\Pr(S=0)}\dfrac{\mu_{1, 1, p_0}(X)}{\mu_{1, 0, p_0}(X)}\mu_{0 , p_0}(X) g(O)\Big\}\\
    =&\E_{p_{0}}\Bigg[\dfrac{1-S}{\Pr(S=0)}\Big\{\dfrac{\mu_{1, 1, p_0}(X)}{\mu_{1, 0, p_0}(X)}\mu_{0 , p_0}(X)-\alpha_{1} \Big\}g(O)\Bigg].
\end{align*}
The second equation follows by the fact that $\dfrac{1-S}{\Pr(S=0)}\dfrac{\mu_{1, 1, p_0}(X)}{\mu_{1, 0, p_0}(X)}\mu_{0 , p_0}(X)\{g_{S}(O)+g_{A|S, X}(O)+g_{Y|S, X, A}(O)\}$ is mean zero, and the third one follows because
$\alpha_{1} $ is perpendicular to $g(O)$ given $S$.
The derivation of terms (2)-(4) are similar, we take the derivation of (2) as an example below
\begin{align*}
    (2)=&\E_{p_{0}}\Bigg\{\dfrac{\mu_{0 , p_0}(X)}{\mu_{1, 0, p_0}(X)}\dfrac{\partial}{\partial t}\mu_{1, 1, p_t}(X)\Big|_{t=0}\Big|S=0\Bigg\}\\
    =&\E_{p_{0}}\Bigg(\dfrac{\mu_{0 , p_0}(X)}{\mu_{1, 0, p_0}(X)}\E_{p_{0}}\Big[\{Y-\mu_{1, 1, p_0}(X)\}g_{Y|X, A=1, S=1}(O)\Big|X, A=1, S=1\Big]\Big|S=0\Bigg)\\
    =&\E_{p_{0}}\Bigg(\dfrac{1-S}{\Pr(S=0)}\dfrac{\mu_{0 , p_0}(X)}{\mu_{1, 0, p_0}(X)}\E_{p_{0}}\Big[\dfrac{AS}{q(X)\Pr(S=1|X)}\{Y-\mu_{1, 1, p_0}(X)\}g_{Y|X, A, S}(O)\Big|X\Big]\Bigg)\\
    =&\E_{p_{0}}\Bigg[\dfrac{\Pr(S=0|X)}{\Pr(S=0)}\dfrac{\mu_{0 , p_0}(X)}{\mu_{1, 0, p_0}(X)}\dfrac{AS}{q(X)\Pr(S=1|X)}\{Y-\mu_{1, 1, p_0}(X)\}g_{Y|X, A, S}(O)\Bigg]\\
    =&\E_{p_{0}}\Bigg[\dfrac{\Pr(S=0|X)}{\Pr(S=0)}\dfrac{\mu_{0 , p_0}(X)}{\mu_{1, 0, p_0}(X)}\dfrac{AS}{q(X)\Pr(S=1|X)}\{Y-\mu_{1, 1, p_0}(X)\}g(O)\Bigg]\\
    =&\E_{p_{0}}\Bigg[\dfrac{1}{\Pr(S=0)}S\tau(X)\dfrac{\mu_{0, p_0}(X)}{\mu_{1, 0, p_0}(X)}\dfrac{A}{q(X)}\{Y-\mu_{1, 1, p_0}(X)\}g(O)\Bigg].
\end{align*}
Similar to (2),
$
    (4)=\E_{p_{0}}\Bigg[\dfrac{1-S}{\Pr(S=0)}\dfrac{\mu_{1, 1, p_0}(X)}{\mu_{1, 0, p_0}(X)}\{Y-\mu_{0 , p_0}(X)\}g(O)\Bigg].
$
Applying the chain rule of derivative, and similar to above, (3) is,
\begin{align*}
    -\E_{p_{0}}\Bigg[\dfrac{1}{\Pr(S=0)}S\tau(X)\dfrac{\mu_{1, 1, p_0}(X)\mu_{0 , p_0}(X)}{\{\mu_{1, 0, p_0}(X)\}^{2}}\dfrac{(1-A)}{\{1-q(X)\}}\{Y-\mu_{1, 0, p_0}(X)\}g(O)\Bigg].
\end{align*}
After algebra re-arrangement, we conclude that the influence function of $\alpha_{1} $ is $A(O, \boldsymbol{\eta_1})$. The same result can be obtained using other identification results.

By applying basic nonparametric theory and chain rules, it is trivial to derive other influence functions.
\end{proof}
\clearpage

%%%%%%%%%%%%%%%%%%%%%%%%%%%%%%%%%%%%%%%%%%%%%%%%%%%%%%%%%%%%
\section{Properties of $\widehat \kappa$}\label{appendix: kappa}
%%%%%%%%%%%%%%%%%%%%%%%%%%%%%%%%%%%%%%%%%%%%%%%%%%%%%%%%%%%%

\begin{restatable}{lemma}{lemmagamma}
\label{lemma:gamma}
For the above defined $\widehat \kappa(\widetilde{x}) = \{n^{-1} \sum_{i=1}^n I(\widetilde{X}=\widetilde{x}, S_i = s)\}^{-1}$,
\begin{enumerate}
    \item[(i)] $I(\widehat \kappa(\widetilde{x})^{-1}=0)=o_p(n^{-1/2})$
    \item[(ii)] $\widehat \kappa(\widetilde{x})\xrightarrow{P} \kappa(\widetilde{x})$.
\end{enumerate}
\end{restatable}
\begin{proof}
For any $\varepsilon>0$,
\begin{align*}
    &\Pr\{|\sqrt{n}I(\widehat\kappa(\widetilde{x})^{-1}=0)-0|>\varepsilon\}\\
    \leqslant 
    & \Pr\{\widehat\kappa(\widetilde{x})^{-1}=0\}=1-\{\kappa(\widetilde{x})^{-1}\}^{n}\\
    \rightarrow 
    & \text{ } 0, \quad \mathrm{ as } \quad n\rightarrow \infty \quad \textrm{(By condition \emph{A6})}.
\end{align*}
By Weak Law of Large Numbers, $\forall \varepsilon$, we have
$
    \Pr\{|\widehat\kappa(\widetilde{x})^{-1}-\kappa(\widetilde{x})^{-1}|>\varepsilon\},
$
as $n\rightarrow\infty$.
Next, fix any $\varepsilon>0$, set $\kappa\equiv\Pr\{\widetilde{X}=\widetilde{x}, S=s\}/2$, we have
\begin{align*}
    &\Pr\{|\widehat\kappa(\widetilde{x})-\kappa(\widetilde{x})|>\varepsilon\}\\
    \leqslant
    & \Pr\Big\{\widehat\kappa(\widetilde{x})^{-1}\geqslant\kappa, \Big|\widehat\kappa(\widetilde{x})-\kappa(\widetilde{x})\Big|>\varepsilon\Big\}+\Pr\{\widehat\kappa(\widetilde{x})^{-1}<\kappa\}\\
    \leqslant
    & \Pr\{|\widehat\kappa(\widetilde{x})^{-1}-\kappa(\widetilde{x})^{-1}|\geqslant2\kappa^{2}\varepsilon\}+\Pr\{|\widehat\kappa(\widetilde{x})^{-1}-\kappa(\widetilde{x})^{-1}|>\kappa\}\\
    \rightarrow
    & \text{ } 0, \quad \mathrm{ as } \quad n\rightarrow \infty
\end{align*}
\end{proof}
\clearpage

%%%%%%%%%%%%%%%%%%%%%%%%%%%%%%%%%%%%%%%%%%%%%%%%%%%%%%%%%%%%
\section{Proof of Theorem \ref{thm:inference1}}\label{appendix:inference1}
\begin{proof}
For general functions $\boldsymbol{\widetilde \eta_1}$ define 
\begin{align*}
H(\boldsymbol{\widetilde \eta_1}) =  &\widetilde\kappa\Bigg(
    (1-S)
    \dfrac{\widetilde\mu_{1, 1}(X)}{\widetilde\mu_{1, 0}(X)}Y+
    S\widetilde\tau(X)\dfrac{\widetilde\mu_{0, 0}(X)}{\widetilde\mu_{1, 0}(X)}\Bigg[
    \dfrac{A}{\widetilde q(X)}\big\{Y-\widetilde\mu_{1, 1}(X)\big\}-\\&
    \dfrac{\widetilde\mu_{1,1}(X)}{\widetilde\mu_{1, 0}(X)}\dfrac{1-A}{1-\widetilde q(X)}\big\{Y-\widetilde\mu_{1, 0}(X)\big\}
    \Bigg]
    \Bigg).
\end{align*}
Therefore, $\widehat \alpha_{1}  = \mathbb{P}_n\big\{H(\boldsymbol{\widehat\eta_1})\big\}$.\\
We observe that we can decompose $\widehat \alpha_{1}  - \alpha_{1} $ into three parts as below.
\begin{align*}
     &\widehat \alpha_{1}  - \alpha_{1} =\mathbb{P}_n\big\{H(\boldsymbol{\widehat \eta_1})\big\}-\mathbb{P}\big\{H(\boldsymbol{\eta_1})\big\}\\
    =&\underbrace{
    (\mathbb{P}_n-\mathbb{P})\big\{H(\boldsymbol{\widehat \eta_1})-H(\boldsymbol{\eta_1})\big\} }_{1}+
    \underbrace{
    \mathbb{P}\big\{H(\boldsymbol{\widehat \eta_1})-H(\boldsymbol{\eta_1})\big\} }_{2}+
    \underbrace{
    (\mathbb{P}_n-\mathbb{P})H(\boldsymbol{\eta_1})}_{3}.
\end{align*}
Working on term 1, note that
\begin{align*}
    &||H(\boldsymbol{\widehat \eta_1})-H(\boldsymbol{\eta_1})||\\
    = 
    & \Bigg|\Bigg|\widehat\kappa\Bigg(
    (1-S)
    \dfrac{\widehat\mu_{1, 1}(X)}{\widehat\mu_{1, 0}(X)}Y\\
    &+
    S\widehat\tau(X)\dfrac{\widehat\mu_{0, 0}(X)}{\widehat\mu_{1, 0}(X)}\Bigg[
    \dfrac{A}{\widehat q(X)}\big\{Y-\widehat\mu_{1, 1}(X)\big\}-\dfrac{1-A}{1-\widehat q(X)}\dfrac{\widehat\mu_{1, 1}(X)}{\widehat\mu_{1, 0}(X)}\big\{Y-\widehat\mu_{1, 0}(X)\big\}
    \Bigg]
    \Bigg)-\\
    &\kappa\Bigg(
    (1-S)
    \dfrac{\mu_{1, 1}(X)}{\mu_{1, 0}(X)}Y\\
    &+
    S\tau(X)\dfrac{\mu_{0, 0}(X)}{\mu_{1, 0}(X)}\Bigg[
    \dfrac{A}{q(X)}\big\{Y-\mu_{1, 1}(X)\big\}-\dfrac{1-A}{1-q(X)}\dfrac{\mu_{0, 0}(X)}{\mu_{1, 0}(X)}\big\{Y-\mu_{1, 0}(X)\big\}
    \Bigg]
    \Bigg)\Bigg|\Bigg|\\
    \leqslant &
    \Bigg|\Bigg|
    \widehat\kappa\Bigg((1-S)
    \Bigg[\Big\{\dfrac{\widehat\mu_{1, 1}(X)}{\widehat\mu_{1, 0}(X)}-\dfrac{\mu_{1, 1}(X)}{\mu_{1, 0}(X)}\Big\}Y+\\
    &S\widehat\tau(X)\dfrac{\widehat\mu_{0, 0}(X)}{\widehat\mu_{1, 0}(X)}\Big[
    \dfrac{A}{\widehat q(X)}\{\widehat\mu_{1, 1}(X)-\mu_{1, 1}(X)\}-\dfrac{1-A}{1-\widehat q(X)}\dfrac{\widehat\mu_{1, 1}(X)}{\widehat\mu_{0, 0}(X)}\{\widehat\mu_{1, 0}(X)-\mu_{1, 0}(X)\}
    \Big]
    \Bigg]    
    \Bigg)
    \Bigg|\Bigg|\\
    &+\Bigg|\Bigg|
    (\widehat\kappa-\kappa)\Bigg(
    (1-S)
    \dfrac{\mu_{1, 1}(X)}{\mu_{1, 0}(X)}Y+\\
    &
    S\tau(X)\dfrac{\mu_{0, 0}(X)}{\mu_{1, 0}(X)}\Bigg[
    \dfrac{A}{q(X)}\big\{Y-\mu_{1, 1}(X)\big\}-\dfrac{1-A}{1-q(X)}\dfrac{\mu_{0, 0}(X)}{\mu_{1, 0}(X)}\big\{Y-\mu_{1, 0}(X)\big\}
    \Bigg]
    \Bigg)
    \Bigg|\Bigg|+\\
    &\Bigg|\Bigg|
    \widehat \kappa \Bigg(
    AS\Big\{\dfrac{\widehat\mu_{0, 0}(X)\widehat\tau(X)}{\widehat\mu_{1, 0}(X)\widehat q(X)}-\dfrac{\mu_{0, 0}(X)\tau(X)}{\mu_{1, 0}(X) q(X)}\Big\}\{Y-\mu_{1, 1}(X)\}
    +\\
    &(1-A)S\Big[\dfrac{\widehat\mu_{1, 1}(X)\widehat\mu_{0, 0}(X)\widehat\tau(X)}{\widehat\mu_{1, 0}(X)^{2}\{1-\widehat q(X)\}}-\dfrac{\mu_{1, 1}(X)\mu_{0, 0}(X)\tau(X)}{\mu_{1, 0}(X)^{2}\{1- q(X)\}}\Big]\{Y-\mu_{1, 0}(X)\}
    \Bigg)
    \Bigg|\Bigg|\\
    \lesssim
    &||\widehat\mu_{0, 0}(X)-\mu_{0, 0}(X)||+||\widehat\mu_{1, 0}(X)-\mu_{1, 0}(X)||+||\widehat\mu_{1, 1}(X)-\mu_{1, 1}(X)||+\\
    &||\widehat \tau(X)-\tau(X)||+||\widehat q(X)-q(X)||+||\widehat\kappa-\kappa||.
\end{align*}
Working on term 2, we have
\begin{align*}
     &\E\{H(\boldsymbol{\widehat \eta_1})-H(\boldsymbol{\eta_1})\}\\
    =&(\widehat\kappa-\kappa)\alpha_{1}/\kappa + \widehat\kappa\\
    &\E\Bigg(\Pr(S=0|X)
    \Bigg[
    \dfrac{\widehat\mu_{1, 1}(X)}{\widehat\mu_{1, 0}(X)}\mu_{0, 0}(X)-\dfrac{\mu_{1, 1}(X)}{\mu_{1, 0}(X)}\mu_{0, 0}(X)\Bigg]
    +
    \Pr(S=1|X)\widehat\tau(X)\dfrac{\widehat\mu_{0, 0}(X)}{\widehat\mu_{1, 0}(X)}\\
    &\Bigg[
    \dfrac{q(X)}{\widehat q(X)}\{\mu_{1, 1}(X)-\widehat\mu_{1, 1}(X)\}+
    \dfrac{1-q(X)}{1-\widehat q(X)}\dfrac{\widehat\mu_{1, 1}(X)}{\widehat\mu_{1, 0}(X)}\{\mu_{1, 0}(X)-\widehat\mu_{1, 0}(X)\}
    \Bigg]
    \Bigg)\\
    =&(\widehat\kappa-\kappa)\alpha_{1}/\kappa + \widehat\kappa\E\Bigg(\Pr(S=0|X)
    \Bigg[
    \Bigg\{\dfrac{\widehat\mu_{1, 1}(X)}{\widehat\mu_{1, 0}(X)}-\dfrac{\mu_{1, 1}(X)}{\mu_{1, 0}(X)}\Bigg\}\mu_{0, 0}(X)+\\
    &
    \dfrac{\widehat\tau(X) q(X)\widehat\mu_{0, 0}(X)}{\tau(X) \widehat q(X)\widehat\mu_{1, 0}(X)}\{\mu_{1, 1}(X)-\widehat\mu_{1, 1}(X)\}+\dfrac{1-q(X)}{1-\widehat q(X)}
    \dfrac{\widehat\tau(X) \widehat\mu_{0, 0}(X)\widehat\mu_{1, 1}(X)}{\tau(X) \widehat\mu_{1, 0}(X)^{2}}\{\widehat\mu_{1, 0}(X)-\mu_{1, 0}(X)\}
    \Bigg]
    \Bigg)\\
    =&(\widehat\kappa-\kappa)\alpha_{1}/\kappa + \widehat\kappa\\
    &\E\Bigg(\Pr(S=0|X)
    \Bigg[
    \dfrac{\mu_{0, 0}(X)}{\mu_{1, 0}(X)}
    \{\widehat\mu_{1, 1}(X)-\mu_{1, 1}(X)\}+
    \dfrac{\mu_{0, 0}(X)}{\mu_{1, 0}(X)}\dfrac{\widehat\mu_{1, 1}(X)}{\widehat\mu_{1, 0}(X)}\{\mu_{1, 0}(X)-\widehat\mu_{1, 0}(X)\}+\\
    &\dfrac{\widehat\tau(X) q(X)\widehat\mu_{0, 0}(X)}{\tau(X) \widehat q(X)\widehat\mu_{1, 0}(X)}
    \{\mu_{1, 1}(X)-\widehat\mu_{1, 1}(X)\}+
    \dfrac{1-q(X)}{1-\widehat q(X)}
    \dfrac{\widehat\tau(X) \widehat\mu_{0, 0}(X)\widehat\mu_{1, 1}(X)}{\tau(X) \widehat\mu_{1, 0}(X)^{2}}\{\widehat\mu_{1, 0}(X)-\mu_{1, 0}(X)\}
    \Bigg]
    \Bigg)\\
    =&(\widehat\kappa-\kappa)\alpha_{1}/\kappa + \widehat\kappa\E\Bigg(\Pr(S=0|X)
    \Bigg[
    \Bigg\{
    \dfrac{\mu_{0, 0}(X)}{\mu_{1, 0}(X)}-\dfrac{\widehat\tau(X) q(X)\widehat\mu_{0, 0}(X)}{\tau(X) \widehat q(X)\widehat\mu_{1, 0}(X)}
    \Bigg\}
    \{\widehat\mu_{1, 1}(X)-\mu_{1, 1}(X)\}+\\
    &\dfrac{\widehat\mu_{1, 1}(X)}{\widehat\mu_{1, 0}(X)}
    \Bigg\{
    \dfrac{1-q(X)}{1-\widehat q(X)}\dfrac{\widehat\tau(X) \widehat\mu_{0, 0}(X)}{\tau(X) \widehat\mu_{1, 0}(X)}-\dfrac{\mu_{0, 0}(X)}{\mu_{1, 0}(X)}
    \Bigg\}
    \{\widehat\mu_{1, 0}(X)-\mu_{1, 0}(X)\}
    \Bigg]
    \Bigg)\\
    =&(\widehat\kappa-\kappa)\alpha_{1}/\kappa + R_{1, n}\\
    \lesssim& (\widehat\kappa-\kappa)\alpha_{1}/\kappa + \\
    &O_p\Big\{
    \big(
    \left\lVert \widehat \mu_{1, 1}(X)-\mu_{1, 1}(X) \right\rVert+
    \left\lVert \widehat \mu_{1, 0}(X)-\mu_{1, 0}(X) \right\rVert\big)
    \big(
    \left\lVert \widehat r_S(X)-r_S(X) \right\rVert+
    \left\lVert \widehat \tau(X)-\tau(X) \right\rVert\big)
    \Big\}.
\end{align*}
Alternatively, we have
\begin{align*}
     &\E\{H(\boldsymbol{\widehat \eta_1})-H(\boldsymbol{\eta_1})\}\\
    =&(\widehat\kappa-\kappa)\alpha_{1}/\kappa + \widehat\kappa\\
    &\E\Bigg(\Pr(S=0|X)
    \Bigg[
    \Bigg\{\dfrac{\widehat\mu_{1, 1}(X)}{\widehat\mu_{1, 0}(X)}-\dfrac{\mu_{1, 1}(X)}{\mu_{1, 0}(X)}\Bigg\}\mu_{0, 0}(X)+
    \Bigg\{\dfrac{\mu_{1, 1}(X)}{\mu_{1, 0}(X)}-\dfrac{\widehat\mu_{1, 1}(X)}{\widehat\mu_{1, 0}(X)}\Bigg\}
    \dfrac{\widehat\tau(X) q(X)}{\tau(X) \widehat q(X)}\widehat\mu_{0, 0}(X)+\\
    &\dfrac{\mu_{1, 1}(X)}{\mu_{1, 0}(X)}\dfrac{\widehat\tau(X) q(X)}{\tau(X) \widehat q(X)}\widehat\mu_{0, 0}(X)
    -\dfrac{\mu_{1, 1}(X)}{\mu_{1, 0}(X)}\dfrac{\widehat\tau(X) q(X)}{\tau(X) \widehat q(X)}\widehat\mu_{0, 0}(X)+\\
    &\Bigg\{\dfrac{\mu_{1, 1}(X)}{\mu_{1, 0}(X)}-\dfrac{\widehat\mu_{1, 1}(X)}{\widehat\mu_{1, 0}(X)}\Bigg\}\dfrac{1-q(X)}{1-\widehat q(X)}
    \dfrac{\widehat\tau(X) \widehat\mu_{0, 0}(X)}{\tau(X) \widehat\mu_{1, 0}(X)}\{\widehat\mu_{1, 0}(X)-\mu_{1, 0}(X)\}\\
    &+\dfrac{\mu_{1, 1}(X)}{\mu_{1, 0}(X)}\dfrac{1-q(X)}{1-\widehat q(X)}
    \dfrac{\widehat\tau(X) \widehat\mu_{0, 0}(X)}{\tau(X) \widehat\mu_{1, 0}(X)}\{\widehat\mu_{1, 0}(X)-\mu_{1, 0}(X)\}
    \Bigg]
    \Bigg)\\
    =&(\widehat\kappa-\kappa)\alpha_{1}/\kappa + \widehat\kappa\\
    &\E\Bigg(\Pr(S=0|X)
    \Bigg[
   \Bigg\{\dfrac{\widehat\mu_{1, 1}(X)}{\widehat\mu_{1, 0}(X)}-\dfrac{\mu_{1, 1}(X)}{\mu_{1, 0}(X)}\Bigg\}\\
   &\Bigg[\mu_{0, 0}(X)-
    \dfrac{\widehat\tau(X) q(X)}{\tau(X) \widehat q(X)}\widehat\mu_{0, 0}(X)+
    \dfrac{\widehat\tau(X) }{\tau(X)}\dfrac{1-q(X)}{1-\widehat q(X)}\widehat\mu_{0, 0}(X)\Big\{1-\dfrac{\mu_{1, 0}(X)}{\widehat\mu_{1, 0}(X)}\Big\}\Bigg]\\
    &+\dfrac{\mu_{1, 1}(X)}{\mu_{1, 0}(X)}\dfrac{q(X)}{\widehat q(X)}
    \dfrac{\widehat\tau(X) \widehat\mu_{0, 0}(X)}{\tau(X) \widehat\mu_{1, 0}(X)}\{\widehat\mu_{1, 0}(X)-\mu_{1, 0}(X)\}\Bigg\{
    1-\dfrac{1-q(X)}{1-\widehat q(X)}\dfrac{\widehat q(X)}{q(X)}
    \Bigg\}
    \Bigg]
    \Bigg)\\
    =&(\widehat\kappa-\kappa)\alpha_{1}/\kappa + R_{1, n}\\
    \lesssim&(\widehat\kappa-\kappa)\alpha_{1}/\kappa + \\
    &\left\lVert \dfrac{\widehat\mu_{1, 1}(X)}{\widehat\mu_{1, 0}(X)}-\dfrac{\mu_{1, 1}(X)}{\mu_{1, 0}(X)}\right\rVert
    \Big\{
    \left\lVert \widehat\mu_{0, 0}(X)-\mu_{0, 0}(X)\right\rVert+
    \left\lVert\widehat\mu_{1, 0}(X)-\mu_{1, 0}(X)\right\rVert+\left\lVert\widehat\tau(X)-\tau(X)\right\rVert
    \Big\}\\
    &+\left\lVert\widehat q(X)-q(X)\right\rVert\Bigg\{
    \left\lVert \dfrac{\widehat\mu_{1, 1}(X)}{\widehat\mu_{1, 0}(X)}-\dfrac{\mu_{1, 1}(X)}{\mu_{1, 0}(X)}\right\rVert+\left\lVert\widehat\mu_{1, 0}(X)-\mu_{1, 0}(X)\right\rVert
    \Bigg\}\\
    \lesssim&(\widehat\kappa-\kappa)\alpha_{1}/\kappa + \\
    &O_p\Big\{
    \big(\left\lVert \widehat r_A(X)-r_A(X)\right\rVert \big)
    \big(
    \left\lVert \widehat \mu_{0, 0}(X)-\mu_{0, 0}(X) \right\rVert+
    \left\lVert \widehat \mu_{1, 0}(X)-\mu_{1, 0}(X) \right\rVert+
    \left\lVert \widehat \tau(X)-\tau(X) \right\rVert
    \big)
    \Big\}
\end{align*}

Because $I(\widehat \kappa^{-1}=0)=o_p(n^{-1/2})$, and we also assume $||\widehat q(X)-q(X)||=o_p(n^{-1/2})$. Therefore, 
by conditions \emph{(a3)}, we have $||H(\boldsymbol{\widehat \eta_1})-H(\boldsymbol{\eta_1})|| = o_p(1)$, so that by Lemma 2 of \citet{kennedy2020sharp}, term 1, $(\mathbb{P}_n-\mathbb{P})\big\{H(\boldsymbol{\widehat \eta_1})-H(\boldsymbol{\eta_1})\big\}=o_p(n^{-1/2})$. Combine terms 2 and 3, we have
\begin{align*}
     &\mathbb{P}\big\{H(\boldsymbol{\widehat \eta_1})-H(\boldsymbol{\eta_1})\big\}+
    (\mathbb{P}_n-\mathbb{P})H(\boldsymbol{\eta_1})\\
    =&\mathbb{P}_n\{A(O, \boldsymbol{\eta_1})\}+R_{1, n}+\Big\{\dfrac{\widehat\kappa}{\kappa} -1\Big\}\alpha_{1} +
    \Big\{\dfrac{\kappa}{\widehat\kappa} -1\Big\}\alpha_{1},
\end{align*}
Factoring and simplifying the last two terms yields $-\{\widehat\kappa -\kappa\}^2\alpha_{1}/\widehat\kappa\kappa $, which is
$o_{p}(n^{-1/2})$ by the central limit theorem. 

Therefore, combining terms 1, 2 and 3, we conclude
\begin{align*}
     \widehat \alpha_{1} -\alpha_{1} 
    =\mathbb{P}_n\{A(O, \boldsymbol{\eta_1})\}+R_{1, n}+o_p(n^{-1/2}).
\end{align*}
In particular, if $R_{1, n}= o_p(n^{-1/2})$, then 
\begin{align*}
    \sqrt{n}\big(\widehat \alpha_{1} -\alpha_{1} \big)\rightsquigarrow\mathcal{N}\big[0, \E\{A(O, \boldsymbol{\eta_1})^{2}\}\big].
\end{align*}
That is, $\widehat \alpha_{1} $ is non-parametric efficient.

The estimator of $\beta_{1}$ is the empirical average of $Y(1-S)$. According to non-parametric theory, it is trivial to see that, without any conditions, 
$
    \sqrt{n}(\widehat \beta_{1}  -\beta_{1})\rightsquigarrow\mathcal{N}\big[0, \E\{B(O)^{2}\}\big].
$
Since we just use chain rules to derive the influence functions of $\phi_{1}$ and $\psi_{1}$, with the same conditions hold for $\widehat \alpha_{1}$, we have $ 
      \sqrt{n}(\widehat \phi_{1}  -\phi_{1})\rightsquigarrow\mathcal{N}\big[0, \E\{\Phi_{1}(O, \boldsymbol{\eta_1})^{2}\}\big],
$
and 
$
      \sqrt{n}(\widehat \psi_{1}  -\psi_{1})\rightsquigarrow\mathcal{N}\big[0, \E\{\Psi_{1}(O, \boldsymbol{\eta_1})^{2}\}\big].
$
\end{proof}
\clearpage
\section{Efficient estimation under (A4*)}\label{appendix:A4*1}
The following theorem gives the identification results.
\begin{theorem}
\label{thm:identification1*}
Under conditions (A1) through (A3), (A4*), and (A5) through (A6), $\E[Y^1|S=0]$ can be identified by $\alpha'_1\equiv\E \left\{\mu_{1, 1}(X) | S = 0 \right\}$ and $\alpha'_1=\E \left\{S\tau(X)A/q(X)Y | S = 0 \right\}$.
In addition, $\E[Y^0|S=0]$ can be identified by $\beta_1\equiv\mathbb{P}_{n}(Y| S = 0)$. Then $\E[Y^1/Y^0|S=0]$ and $\E[Y^1-Y^0|S=0]$ can be identified by $\phi'_1\equiv\alpha'_{1}/\beta_{1}$ and $\psi'_1\equiv\alpha'_{1}-\beta_{1}$ respectively.
\end{theorem}
The following theorem gives the influence functions of the target parameters in the above theorem.
\begin{theorem}
\label{thm:IF1*}
Under conditions (A1) through (A3), (A4*), and (A5) through (A6), the non-parametric influence function of $\alpha_{1}$, $A'_{1}(O, \boldsymbol{\eta_{1}})$ is
\begin{align*}
A'_{1}\{O;\mu_{1, 1}(X),\tau(X),q(X)\}=    
\kappa\Bigg(
    (1-S)
    \{\mu_{1, 1}(X)-\alpha_{1}\}+
    S\tau(X)
    \dfrac{A}{q(X)}\{Y-\mu_{1, 1}(X)\}
    \Bigg).
\end{align*}
The influence function of $\beta_{1}$ is $B_{1}(O)=Y(1-S)-\beta_{1}$. In addition, the influence functions of $\phi'_{1}$ and $\psi'_{1}$ are $\Phi'_{1}(O, \boldsymbol{\eta_{1}})=\frac{1}{\beta_{1}}\{A'_{1}(O, \boldsymbol{\eta_{1}})-\phi_{1} B_{1}(O)\}$ and $\Psi_{1}'(O, \boldsymbol{\eta_{1}})=A'_{1}(O, \boldsymbol{\eta_{1}})-B_{1}(O)$ respectively.
\end{theorem}
Then the proposed estimator is $
\widehat \alpha_{1}'=
\kappa\mathbb{P}_{n}\Big[
(1-S)\widehat\mu_{1, 1}(X)+
S\widehat\tau(X)\dfrac{A}{\widehat q(X)}\{Y-\widehat\mu_{1, 1}(X)\}
    \Big]
$.
The following theorem gives the asymptotic properties of $\widehat\alpha_{1}'$.
\begin{theorem}
\label{thm:inference1*}
If conditions \emph{(a1)} through \emph{(a3)} hold, then $\widehat \alpha'_1$ is consistent with rate of convergence $O_p(R'_{1, n}+n^{-1/2})$, where  
$
    R'_{1, n}=\left\lVert \widehat \mu_{1, 1}(X)-\mu_{1, 1}(X) \right\rVert 
    \left\lVert \widehat \tau(X)-\tau(X) \right\rVert.
$
When $R'_{1, n}= o_p(n^{-1/2})$, $\widehat \alpha'_1$ is asymptotically normal and non-parametric efficient.
\end{theorem}
The proofs of the above theorems are trivial given the proofs of Theorems \ref{thm:identification1} to \ref{thm:inference1}.
\clearpage

%%%%%%%%%%%%%%%%%%%%%%%%%%%%%%%%%%%%%%%%%%%%%%%%%%%%%%%%%%%%
\section{Proof of Theorem \ref{thm:identification2}}\label{appendix:identification2}
%%%%%%%%%%%%%%%%%%%%%%%%%%%%%%%%%%%%%%%%%%%%%%%%%%%%%%%%%%%%

We will now show that under conditions (A1) through (A5), and (B1) through (B2) listed in the main text, the target parameters are identifiable. 
\begin{proof}
\begin{equation*}
    \begin{split}
        \alpha_2=\E(Y^1 | S = 0) &= \E \big\{ \E(Y^1 | X, S = 0) \big | S = 0 \big\} \\ 
         &= \E \left\{ \dfrac{\E(Y^1 | X, S = 0)}{\E(Y^0 | X, S = 0)} \E(Y^0 | X, S = 0) \Big | S = 0 \right\} \\ 
         &= \E \left\{ \dfrac{\E(Y^1 | X, S = 1)}{\E(Y^0 | X, S = 1)} \E(Y^0 | X, S = 0) \Big | S = 0 \right\} \\ 
          &= \E \left\{ \dfrac{\E(Y^1 | X, S = 1, A = 1)}{\E(Y^0 | X, S = 1, A = 0)} \E(Y^0 | X, S = 0) \Big | S = 0 \right\} \\ 
          &= \E \left\{ \dfrac{\E(Y^1 | X, S = 1, A = 1)}{\E(Y^0 | X, S = 1, A = 0)} \E(Y^0 | X, S = 0, A =0) \Big | S = 0 \right\} \\ 
          &= \E \left\{ \dfrac{\E(Y | X, S = 1, A = 1)}{\E(Y | X, S = 1, A = 0)} \E(Y | X, S = 0, A = 0) \Big | S = 0 \right\} ,
    \end{split}
\end{equation*}
where the first step follows from the law of total expectation; the second from multiplying and dividing by $\E(Y^0 | X, S = 0)$ , which is allowed by the positivity component of condition (A4); the third from the exchangeability component of condition (A4); the fourth from condition (A2); the fifth from condition (B1); and the last from condition (A1).
Next, we give other identification strategies in the following propositions.

Furthermore, using the law of total expectation and conditions (A1), (B1), and (B2) we can write, 
\begin{equation*}
    \begin{split}
        \beta_{2}=\E(Y^0 | S = 0) &= \E \big\{  \E(Y^0 | X, S = 0)  \big | S = 0 \big\} \\ 
        &= \E \big\{  \E(Y^0 | X, S = 0, A = 0)  \big | S = 0 \big\} \\ 
        &= \E \big\{  \E(Y | X, S = 0, A = 0)  \big | S = 0 \big\}.
    \end{split}
\end{equation*}

Taken together, the above results show that, under conditions (A1) through (A5), (B1) and (B2) we can identify the causal mean ratio in the target population by $\phi_{2}=\alpha_2/\beta_{2}$ and the causal mean difference in the target population by $\psi_{2}=\alpha_2-\beta_{2}$.
\end{proof}
\begin{prop}
The target parameter $\alpha_2$ can be also written as 
\begin{align*} 
   &\E \left\{ \dfrac{\E[Y | X, S = 1, A = 1]}{\E[Y | X, S = 1, A = 0]} (1-A)/\Pr[A=0|X, S=0] Y| S = 0 \right\}, \\
   &\dfrac{1}{\Pr[S=0]}\E \left\{S\dfrac{\Pr[S=0|X]}{\Pr[S=1|X]}\dfrac{A}{\Pr[A=a|X, S=1]}\dfrac{\E[Y | X, S = 0, A = 0]}{\E[Y | X, S = 1, A = 0]} Y  \right\},\\
   &   \dfrac{1}{\Pr[S=0]}\E \left[S\dfrac{\Pr[S=0|X]}{\Pr[S=1|X]}\dfrac{1-A}{1-\Pr[A=a|X, S=1]}\dfrac{\E[Y | X, S = 1, A = 1]}{\E[Y | X, S = 1, A = 0]}\dfrac{\E[Y | X, S = 0, A = 0]}{\E[Y | X, S = 1, A = 0]}Y \right].
\end{align*}
\end{prop}
The derivation of the weighting estimators follows the same logic as \ref{appendix:identification1}. 
\clearpage

%%%%%%%%%%%%%%%%%%%%%%%%%%%%%%%%%%%%%%%%%%%%%%%%%%%%%%%%%%%%
\section{Proof of Lemma \ref{thm:IF2}}\label{appendix:IF2}
%%%%%%%%%%%%%%%%%%%%%%%%%%%%%%%%%%%%%%%%%%%%%%%%%%%%%%%%%%%%

\begin{proof}
Recall that under the nonparametric model, $\mathcal M_{\text{\tiny np}}$, for the observable data $O$, the density of the law of the observable data can be written as 
\begin{equation*}
    p(r,x,a,y) = p(s) p(x|s) p(a|x, s) p(y|a, x, s).
\end{equation*}
Under this model, the tangent space is the Hilbert space of mean zero random variables with finite variance. It can be decomposed as $L_2^0 = \Lambda_{S} \oplus \Lambda_{X|S} \oplus   \Lambda_{A|S,X} \oplus  \Lambda_{Y|S, X, A} $.

\noindent
Recall that
\begin{equation*} 
  \alpha_2  \equiv \E \left\{ \dfrac{\E(Y | X, S = 1, A = 1)}{\E(Y | X, S = 1, A = 0)} \E(Y | X, S = 0, A=0) \Big | S = 0 \right\}=\E\Big\{\dfrac{\mu_{1, 1}(X)}{\mu_{1, 0}(X)}\mu_{0, 0}(X)\Big|S=0\Big\}.
\end{equation*}
We will use the path differentiability of $\alpha_2 $ to obtain the efficient influence function under the non-parametric model for the observed data \citep{bickel1993efficient}. To do so, we examine the derivative of $\alpha_{2, p_t}(X)$ with respect to $t$; where the subscript $p_t$ denotes the dependence of $\alpha_2 $ on a one-dimensional parametric sub-model $p_t$, indexed by $t \in [0,1)$, with $t = 0$ denoting the ``true'' data law. 
% Furthermore, denote the function $u_{p}(\cdot)$ as the score function of the observed data with the (conditional) likelihood $p$.\\

Utilizing the law of total expectation, the tricks of $\E\big\{\E(u|w)v\big\}=\E\big\{\E(v|w)u\big\}$, subtracting the mean zero terms, and $\E(A|B)=\E\Big\{\dfrac{I(B=b)}{\Pr(B=b)}A\Big\}$ when $B$ is discrete repeatedly, we have
\begin{align*}
     &\dfrac{\partial \alpha_{2, p_t}}{\partial t}\Big|_{t=0}\\
    =&\dfrac{\partial}{\partial t}\E_{p_t}\Bigg\{\dfrac{\mu_{1, 1, p_t}(X)}{\mu_{1, 0, p_t}(X)}\mu_{0, 0, p_t}(X)\Big|S=0\Bigg\}\Big|_{t=0}\\
    =&\underbrace{\dfrac{\partial}{\partial t}\E_{p_{t}}\Bigg\{\dfrac{\mu_{1, 1, p_0}(X)}{\mu_{1, 0, p_0}(X)}\mu_{0 , 0, p_0}(X)\Big|S=0\Bigg\}\Big|_{t=0}}_{(1)}+
    \underbrace{\E_{p_{0}}\Bigg\{\dfrac{\dfrac{\partial}{\partial t}\mu_{1, 1, p_t}(X)\Big|_{t=0}}{\mu_{1, 0, p_0}(X)}\mu_{0 , 0, p_0}(X)\Big|S=0\Bigg\}}_{(2)}+\\
    &\underbrace{\E_{p_{0}}\Bigg\{\dfrac{\mu_{1, 1, p_0}(X)}{\dfrac{\partial}{\partial t}\mu_{1, 0, p_t}(X)\Big|_{t=0}}\mu_{0 , 0, p_0}(X)\Big|S=0\Big\}}_{(3)}+
    \underbrace{\E_{p_{0}}\Big\{\dfrac{\mu_{1, 1, p_0}(X)}{\mu_{1, 0, p_0}(X)}\dfrac{\partial}{\partial t}\mu_{0, 0, p_t}(X)\Big|_{t=0}\Big|S=0\Bigg\}}_{(4)}.
\end{align*}
Next, we calculate each term by turn. Denote the score function by $g(\cdot)$.
\begin{align*}
    (1)=&\E_{p_{0}}\Big\{\dfrac{1-S}{\Pr(S=0)}\dfrac{\mu_{1, 1, p_0}(X)}{\mu_{1, 0, p_0}(X)}\mu_{0 , 0, p_0}(X) g_{X|S}(O)\Big\}\\
    =&\E_{p_{0}}\Big\{\dfrac{1-S}{\Pr(S=0)}\dfrac{\mu_{1, 1, p_0}(X)}{\mu_{1, 0, p_0}(X)}\mu_{0 , 0, p_0}(X) g(O)\Big\}\\
    =&\E_{p_{0}}\Bigg[\dfrac{1-S}{\Pr(S=0)}\Big\{\dfrac{\mu_{1, 1, p_0}(X)}{\mu_{1, 0, p_0}(X)}\mu_{0 , 0, p_0}(X)-\alpha_2 \Big\}g(O)\Bigg].
\end{align*}
The second equation follows by the fact that $\dfrac{1-S}{\Pr(S=0)}\dfrac{\mu_{1, 1, p_0}(X)}{\mu_{1, 0, p_0}(X)}\mu_{0 , 0, p_0}(X)\{g_{S}(O)+g_{A|S, X}(O)+g_{Y|S, X, A}(O)\}$ is mean zero, and the third one follows because
$\alpha_2 $ is perpendicular to $g(O)$ given $S$.
Next, take the derivation of (2) as an example below
\begin{align*}
    (2)=&\E_{p_{0}}\Bigg\{\dfrac{\mu_{0 , 0, p_0}(X)}{\mu_{1, 0, p_0}(X)}\dfrac{\partial}{\partial t}\mu_{1, 1, p_t}(X)\Big|_{t=0}\Big|S=0\Bigg\}\\
    =&\E_{p_{0}}\Bigg(\dfrac{\mu_{0 , 0, p_0}(X)}{\mu_{1, 0, p_0}(X)}\E_{p_{0}}\Big[\{Y-\mu_{1, 1, p_0}(X)\}g_{Y|X, A=1, S=1}(O)\Big|X, A=1, S=1\Big]\Big|S=0\Bigg)\\
    =&\E_{p_{0}}\Bigg(\dfrac{1-S}{\Pr(S=0)}\dfrac{\mu_{0 , 0, p_0}(X)}{\mu_{1, 0, p_0}(X)}\E_{p_{0}}\Big[\dfrac{AS}{q(X)\Pr(S=1|X)}\{Y-\mu_{1, 1, p_0}(X)\}g_{Y|X, A, S}(O)\Big|X\Big]\Bigg)\\
    =&\E_{p_{0}}\Bigg[\dfrac{\Pr(S=0|X)}{\Pr(S=0)}\dfrac{\mu_{0 , 0, p_0}(X)}{\mu_{1, 0, p_0}(X)}\dfrac{AS}{q(X)\Pr(S=1|X)}\{Y-\mu_{1, 1, p_0}(X)\}g_{Y|X, A, S}(O)\Bigg]\\
    =&\E_{p_{0}}\Bigg[\dfrac{\Pr(S=0|X)}{\Pr(S=0)}\dfrac{\mu_{0 , 0, p_0}(X)}{\mu_{1, 0, p_0}(X)}\dfrac{AS}{q(X)\Pr(S=1|X)}\{Y-\mu_{1, 1, p_0}(X)\}g(O)\Bigg]\\
    =&\E_{p_{0}}\Bigg[\dfrac{1}{\Pr(S=0)}S\tau(X)\dfrac{\mu_{0, p_0}(X)}{\mu_{1, 0, p_0}(X)}\dfrac{A}{q(X)}\{Y-\mu_{1, 1, p_0}(X)\}g(O)\Bigg].
\end{align*}
Similar to (2),
$
    (4)=\E_{p_{0}}\Bigg[\dfrac{1-S}{\Pr(S=0)}\dfrac{\mu_{1, 1}(X)}{\mu_{1, 0}(X)}\dfrac{1-A}{\pi(X)}\big\{Y-\mu_{0, 0}(X)\big\}g(O)\Bigg].
$
Applying the chain rule of derivative, and similar to above, (3) is,
\begin{align*}
    -\E_{p_{0}}\Bigg[\dfrac{1}{\Pr(S=0)}S\tau(X)\dfrac{\mu_{1, 1, p_0}(X)\mu_{0 , 0, p_0}(X)}{\{\mu_{1, 0, p_0}(X)\}^{2}}\dfrac{(1-A)}{\{1-q(X)\}}\{Y-\mu_{1, 0, p_0}(X)\}g(O)\Bigg].
\end{align*}
After some algebra re-arrangement, we conclude that the influence function of $\alpha_2 $ is $\alpha_2(O, \boldsymbol{\eta})$.

It is trivial to derive the influence function of $\beta_{2}$ given the above proof. By the chain rule, it is easy to derive the influence function of $\phi_{2}$ and $\psi_{2}$.

We first show the semiparametric efficient influence function is also $\alpha_2(O, \boldsymbol{\eta})$.
Under the non-parametric model, the Hilbert space can be decomposed as $L_2^0 = \Lambda_{S} \oplus \Lambda_{X|S} \oplus  \Lambda_{A|X, S} \oplus  \Lambda_{Y|A, X, S}$. Now after incorporating the knowledge of the probability $\Pr(A=0|X, S=0)$ will change the decomposition to $L_2^0 = \Lambda_{S} \oplus \Lambda_{X|S} \oplus\Lambda_{Y|A, X, S}$.

According to the total law of expectation, the first term $\kappa (1-S)\dfrac{\mu_{1, 1}(X)}{\mu_{1, 0}(X)}\{\mu_{0, 0}(X)-\alpha_2\}$ is a function of $S$ and $X$ and is mean zero. Therefore, it belongs to $\Lambda_{S} \oplus\Lambda_{X|S}$. All the other terms are functions of $(S, X, A, Y)$ and are mean zero conditional on $(X, A, S)$. Thus, they belong to $\Lambda_{Y|X, A, S}$. We then conclude that $\alpha_2(O, \boldsymbol{\eta})$ belongs to $\Lambda_{S} \oplus\Lambda_{X|S}\oplus\Lambda_{Y|X, A, S}$, which indicates that the unique influence function $\alpha_2(O, \boldsymbol{\eta})$ under the non-parametric model is also the efficient influence function under the semiparametric model which incorporate the restrictions.
\end{proof}
\clearpage

%%%%%%%%%%%%%%%%%%%%%%%%%%%%%%%%%%%%%%%%%%%%%%%%%%%%%%%%%%%%
\section{Proof of Theorem \ref{thm:inference2}}\label{appendix:inference2}
%%%%%%%%%%%%%%%%%%%%%%%%%%%%%%%%%%%%%%%%%%%%%%%%%%%%%%%%%%%%

\begin{proof}
For general functions $\widetilde\kappa, \widetilde\mu_{0, 0}(X)$ and $\widetilde\pi(X)$,  define 
\begin{align*}
G(\boldsymbol{\widetilde\eta_2}') =  
    \widetilde\kappa(1-S)  \Bigg[
    \widetilde\mu_{0, 0}(X)+\dfrac{1-A}{\widetilde\pi(X)}\big\{Y-\widetilde\mu_{0, 0}(X)\big\} 
    \Bigg].
\end{align*}
We observe that we can decompose $\widehat \beta_{2}   - \beta_{2}  $ into three parts as below.
\begin{align*}
     &\widehat \beta_{2}   - \beta_{2}  =\mathbb{P}_n\big\{G(\boldsymbol{\widehat \eta_2}')\big\}-\mathbb{P}\big\{G(\boldsymbol{\eta_2}')\big\}\\
    =&\underbrace{
    (\mathbb{P}_n-\mathbb{P})\big\{G(\boldsymbol{\widehat \eta_2}')-G(\boldsymbol{\eta_2}')\big\} }_{1}+
    \underbrace{
    \mathbb{P}\big\{G(\boldsymbol{\widehat \eta_2}')-G(\boldsymbol{\eta_2}')\big\} }_{2}+
    \underbrace{
    (\mathbb{P}_n-\mathbb{P})G(\boldsymbol{\eta_2}')}_{3}.
\end{align*}
Working on term 1, note that
\begin{align*}
    &||G(\boldsymbol{\widehat \eta_2}')-G(\boldsymbol{\eta_2}')||\\
    = 
    & \Bigg|\Bigg|\widehat\kappa
    (1-S)\Bigg[\widehat\mu_{0, 0}(X)+\dfrac{1-A}{\widehat\pi(X)}\big\{Y-\widehat\mu_{0, 0}(X)\big\}\Bigg] -\\
    &\kappa
    (1-S)\Bigg[\mu_{0, 0}(X)+\dfrac{1-A}{\pi(X)}\big\{Y-\mu_{0, 0}(X)\big\}
    \Bigg]
    \Bigg|\Bigg|\\
    \leqslant 
    & \Bigg|\Bigg|\widehat\kappa
    (1-S)\Bigg[\{\widehat\mu_{0, 0}(X)-\mu_{0, 0}(X)\}+\dfrac{1-A}{\widehat\pi(X)}\big\{\mu_{0, 0}(X)-\widehat\mu_{0, 0}(X)\big\}\Bigg]\Bigg|\Bigg| +\\
    &\Bigg|\Bigg|(\widehat\kappa-\kappa)
    (1-S)\Bigg[\mu_{0, 0}(X)+\dfrac{1-A}{\pi(X)}\big\{Y-\mu_{0, 0}(X)\big\}
    \Bigg]
    \Bigg|\Bigg|+\\
   & \Bigg|\Bigg|\widehat\kappa
    (1-S)(1-A)\big\{Y-\mu_{0, 0}(X)\big\}\Big\{\dfrac{1}{\widehat\pi(X)}-\dfrac{1}{\pi(X)}\Big\}
    \Bigg|\Bigg|\\
    \lesssim&
    ||\widehat\mu_{0, 0}(X)-\mu_{0, 0}(X)||+||\widehat \pi(X)-\pi(X)||.
\end{align*}
By conditions \emph{(b3)}, we have $||G(\boldsymbol{\widehat \eta_2}')-G(\boldsymbol{\eta_2}')|| = o_p(1)$, so that by Lemma 2 of \citet{kennedy2020sharp}, $(\mathbb{P}_n-\mathbb{P})\big\{G(\boldsymbol{\widehat \eta_2}')-G(\boldsymbol{\eta_2}')\big\}=o_p(n^{-1/2})$.\\
Working on term 2, we have
\begin{align*}
     &\E\{G(\boldsymbol{\widehat \eta_2}')-G(\boldsymbol{\eta_2}')\}\\
    =&(\widehat\kappa-\kappa)\beta_{2}  + \widehat\kappa\Bigg[
    \widehat\mu_{0, 0}(X)-\mu_{0, 0}(X)+\dfrac{\pi(X)}{\widehat\pi(X)}\big\{\mu_{0, 0}(X)-\widehat\mu_{0, 0}(X)\big\} 
    \Bigg]\\
    =&(\widehat\kappa-\kappa)\beta_{2}  + \widehat\kappa\Big\{1-\dfrac{\pi(X)}{\widehat\pi(X)}\Big\}\big\{\mu_{0, 0}(X)-\widehat\mu_{0, 0}(X)\big\} \\
    \lesssim & (\widehat\kappa-\kappa)\beta_{2}  + ||\widehat\mu_{0, 0}(X)-\mu_{0, 0}(X)||\cdot||\widehat \pi(X)-\pi(X)||.
\end{align*}
Combine the first term of 2 and 3, we have
\begin{align*}
     &\mathbb{P}\big\{G(\boldsymbol{\widehat \eta_2}')-G(\boldsymbol{\eta_2}')\big\}+
    (\mathbb{P}_n-\mathbb{P})G(\boldsymbol{\eta_2}')\\
    =&\mathbb{P}_n\{B(O, \boldsymbol{\eta_2}')\}+\Big\{\dfrac{\widehat\kappa}{\kappa} -1\Big\}\beta_{2} +
    \Big\{\dfrac{\kappa}{\widehat\kappa} -1\Big\}\beta_{2}.  
\end{align*}
Factoring and simplifying the last two terms yields $\beta_{2}  \left\{\widehat\kappa- \kappa\right\}\left\{\frac{1}{\kappa} - \frac{1}{\widehat\kappa}\right\}$, which is
$o_{p}(n^{-1/2})$ by the central limit theorem. 
Therefore, combining terms 1, 2 and 3, we conclude
\begin{align*}
     \widehat \beta_{2}  -\beta_{2}  
    =O_{p}(R_{2, n})+\mathbb{P}_n\{B(O, \boldsymbol{\eta_2}')\}+o_p(n^{-1/2}),
\end{align*}
In particular, if $R_{2, n}= o_p(n^{-1/2})$, then 
\begin{align*}
    \sqrt{n}\big(\widehat \beta_{2}  -\beta_{2}  \big)\rightsquigarrow\mathcal{N}\big[0, \E\{B(O, \boldsymbol{\eta_2}')^{2}\}\big].
\end{align*}
That is, $\widehat \beta_{2}  $ is non-parametric efficient.

For general functions $\boldsymbol{\widetilde\eta_2}$ define 
\begin{align*}
H(\boldsymbol{\widetilde\eta_2}) =  &\widetilde\kappa\Bigg(
    (1-S)\Bigg[
    \dfrac{\widetilde\mu_{1, 1}(X)}{\widetilde\mu_{1, 0}(X)}\widetilde\mu_{0, 0}(X)+\dfrac{1-A}{\widetilde\pi(X)}\dfrac{\widetilde\mu_{1, 1}(X)}{\widetilde\mu_{1, 0}(X)}\big\{Y-\widetilde\mu_{0, 0}(X)\big\}
    \Bigg]\\
    &+
    S\widetilde\tau(X)\dfrac{\widetilde\mu_{0, 0}(X)}{\widetilde\mu_{1, 0}(X)}\Bigg[
    \dfrac{A}{\widetilde q(X)}\big\{Y-\widetilde\mu_{1, 1}(X)\big\}-\dfrac{\widetilde\mu_{1, 1}(X)}{\widetilde\mu_{1, 0}(X)}\dfrac{1-A}{1-\widetilde q(X)}\big\{Y-\widetilde\mu_{1, 0}(X)\big\}
    \Bigg]
    \Bigg).
\end{align*}

We observe that we can decompose $\widehat \alpha_{2}  - \alpha_{2} $ into three parts as below.
\begin{align*}
     &\widehat \alpha_{2}  - \alpha_{2} =\mathbb{P}_n\big\{H(\boldsymbol{\widehat \eta_2})\big\}-\mathbb{P}\big\{H(\boldsymbol{\eta_2})\big\}\\
    =&\underbrace{
    (\mathbb{P}_n-\mathbb{P})\big\{H(\boldsymbol{\widehat \eta_2})-H(\boldsymbol{\eta_2})\big\} }_{1}+
    \underbrace{
    \mathbb{P}\big\{H(\boldsymbol{\widehat \eta_2})-H(\boldsymbol{\eta_2})\big\} }_{2}+
    \underbrace{
    (\mathbb{P}_n-\mathbb{P})H(\boldsymbol{\eta_2})}_{3}.
\end{align*}
Working on term 1, note that
\begin{align*}
    &||H(\boldsymbol{\widehat \eta_2})-H(\boldsymbol{\eta_2})||\\
    = 
    & \Bigg|\Bigg|\widehat\kappa\Bigg(
    (1-S)\Bigg[
    \dfrac{\widehat\mu_{1, 1}(X)}{\widehat\mu_{1, 0}(X)}\widehat\mu_{0, 0}(X)+\dfrac{1-A}{\widehat\pi(X)}\dfrac{\widehat\mu_{1, 1}(X)}{\widehat\mu_{1, 0}(X)}\big\{Y-\widehat\mu_{0, 0}(X)\big\}
    \Bigg]\\
    &+
    S\widehat\tau(X)\dfrac{\widehat\mu_{0, 0}(X)}{\widehat\mu_{1, 0}(X)}\Bigg[
    \dfrac{A}{\widehat q(X)}\big\{Y-\widehat\mu_{1, 1}(X)\big\}-\dfrac{1-A}{1-\widehat q(X)}\dfrac{\widehat\mu_{1, 1}(X)}{\widehat\mu_{1, 0}(X)}\big\{Y-\widehat\mu_{1, 0}(X)\big\}
    \Bigg]
    \Bigg)-\\
    &\kappa\Bigg(
    (1-S)\Bigg[
    \dfrac{\mu_{1, 1}(X)}{\mu_{1, 0}(X)}\mu_{0, 0}(X)+\dfrac{1-A}{\pi(X)}\dfrac{\mu_{1, 1}(X)}{\mu_{1, 0}(X)}\big\{Y-\mu_{0, 0}(X)\big\}
    \Bigg]\\
    &+
    S\tau(X)\dfrac{\mu_{0, 0}(X)}{\mu_{1, 0}(X)}\Bigg[
    \dfrac{A}{q(X)}\big\{Y-\mu_{1, 1}(X)\big\}-\dfrac{1-A}{1-q(X)}\dfrac{\mu_{0, 0}(X)}{\mu_{1, 0}(X)}\big\{Y-\mu_{1, 0}(X)\big\}
    \Bigg]
    \Bigg)\Bigg|\Bigg|\\
    \leqslant &
    \Bigg|\Bigg|
    \widehat\kappa\Bigg((1-S)
    \Bigg[\Big\{\dfrac{\widehat\mu_{1, 1}(X)}{\widehat\mu_{1, 0}(X)}\widehat\mu_{0, 0}(X)-\dfrac{\mu_{1, 1}(X)}{\mu_{1, 0}(X)}\mu_{0, 0}(X)\Big\}+\\
    &\dfrac{1-A}{1-\widehat \pi(X)}\dfrac{\widehat\mu_{1, 1}(X)}{\widehat\mu_{1, 0}(X)}\{\mu_{0, 0}(X)-\widehat\mu_{0, 0}(X)\}+\\
    &S\widehat\tau(X)\dfrac{\widehat\mu_{0, 0}(X)}{\widehat\mu_{1, 0}(X)}\Big[
    \dfrac{A}{\widehat q(X)}\{\widehat\mu_{1, 1}(X)-\mu_{1, 1}(X)\}-\dfrac{1-A}{1-\widehat q(X)}\dfrac{\widehat\mu_{1, 1}(X)}{\widehat\mu_{0, 0}(X)}\{\widehat\mu_{1, 0}(X)-\mu_{1, 0}(X)\}
    \Big]
    \Bigg]    
    \Bigg)
    \Bigg|\Bigg|\\
    &+\Bigg|\Bigg|
    (\widehat\kappa-\kappa)\Bigg(
    (1-S)\Bigg[
    \dfrac{\mu_{1, 1}(X)}{\mu_{1, 0}(X)}\mu_{0, 0}(X)+\dfrac{1-A}{\pi(X)}\dfrac{\mu_{1, 1}(X)}{\mu_{1, 0}(X)}\big\{Y-\mu_{0, 0}(X)\big\}
    \Bigg]+\\
    &
    S\tau(X)\dfrac{\mu_{0, 0}(X)}{\mu_{1, 0}(X)}\Bigg[
    \dfrac{A}{q(X)}\big\{Y-\mu_{1, 1}(X)\big\}-\dfrac{1-A}{1-q(X)}\dfrac{\mu_{0, 0}(X)}{\mu_{1, 0}(X)}\big\{Y-\mu_{1, 0}(X)\big\}
    \Bigg]
    \Bigg)
    \Bigg|\Bigg|+\\
    &\Bigg|\Bigg|
    \widehat \kappa \Bigg((1-A)(1-S)
    \Big[\dfrac{\widehat\mu_{1, 1}(X)}{\widehat\mu_{1, 0}(X)\{1-\widehat\pi(X)\}}-\dfrac{\mu_{1, 1}(X)}{\mu_{1, 0}(X)\{1-\pi(X)\}}
    \Big]\{Y-\mu_{0, 0}(X)\}+\\
    &AS\Big\{\dfrac{\widehat\mu_{0, 0}(X)\widehat\tau(X)}{\widehat\mu_{1, 0}(X)\widehat q(X)}-\dfrac{\mu_{0, 0}(X)\tau(X)}{\mu_{1, 0}(X) q(X)}\Big\}\{Y-\mu_{1, 1}(X)\}
    +\\
    &(1-A)S\Big[\dfrac{\widehat\mu_{1, 1}(X)\widehat\mu_{0, 0}(X)\widehat\tau(X)}{\widehat\mu_{1, 0}(X)^{2}\{1-\widehat q(X)\}}-\dfrac{\mu_{1, 1}(X)\mu_{0, 0}(X)\tau(X)}{\mu_{1, 0}(X)^{2}\{1- q(X)\}}\Big]\{Y-\mu_{1, 0}(X)\}
    \Bigg)
    \Bigg|\Bigg|\\
    \lesssim
    &||\widehat\mu_{0, 0}(X)-\mu_{0, 0}(X)||+||\widehat\mu_{1, 0}(X)-\mu_{1, 0}(X)||+||\widehat\mu_{1, 1}(X)-\mu_{1, 1}(X)||+\\
    &||\widehat \tau(X)-\tau(X)||+||\widehat q(X)-q(X)||+||\widehat \pi(X)-\pi(X)||+||\widehat\kappa-\kappa||.
\end{align*}
Because $\kappa^{-1}$ is an empirical average of $\Pr(S=0)$, it is trivial to see that
i) $I(\widehat \kappa^{-1}=0)=o_p(n^{-1/2})$, and ii) $\widehat \kappa\xrightarrow{P} \kappa$.
Note that we also assume $||\widehat q(X)-q(X)||=o_p(n^{-1/2})$. Therefore, 
by conditions \emph{(b3)}, we have $||H(\boldsymbol{\widehat \eta_2})-H(\boldsymbol{\eta_2})|| = o_p(1)$, so that by Lemma 2 of \citet{kennedy2020sharp}, $(\mathbb{P}_n-\mathbb{P})\big\{H(\boldsymbol{\widehat \eta_2})-H(\boldsymbol{\eta_2})\big\}=o_p(n^{-1/2})$.\\
Working on term 2, we have
\begin{align*}
     &\E\{H(\boldsymbol{\widehat \eta_2})-H(\boldsymbol{\eta_2})\}\\
    =&(\widehat\kappa-\kappa)\alpha_{2}/\kappa + \widehat\kappa\\
    &\E\Bigg(\Pr(S=0|X)
    \Bigg[
    \dfrac{\widehat\mu_{1, 1}(X)}{\widehat\mu_{1, 0}(X)}\widehat\mu_{0, 0}(X)-\dfrac{\mu_{1, 1}(X)}{\mu_{1, 0}(X)}\mu_{0, 0}(X)+\dfrac{\pi(X)}{\widehat\pi(X)}\dfrac{\widehat\mu_{1, 1}(X)}{\widehat\mu_{1, 0}(X)}\\
    &\{\mu_{0, 0}(X)-\widehat\mu_{0, 0}(X)\}\Bigg]+
    \Pr(S=1|X)\widehat\tau(X)\dfrac{\widehat\mu_{0, 0}(X)}{\widehat\mu_{1, 0}(X)}
    \Bigg[
    \dfrac{q(X)}{\widehat q(X)}\{\mu_{1, 1}(X)-\widehat\mu_{1, 1}(X)\}+\\
    &\dfrac{1-q(X)}{1-\widehat q(X)}\dfrac{\widehat\mu_{1, 1}(X)}{\widehat\mu_{1, 0}(X)}\{\mu_{1, 0}(X)-\widehat\mu_{1, 0}(X)\}
    \Bigg]
    \Bigg)\\
    =&(\widehat\kappa-\kappa)\alpha_{2}/\kappa + \widehat\kappa\\
    &\E\Bigg(\Pr(S=0|X)
    \Bigg[
    \dfrac{\widehat\mu_{1, 1}(X)}{\widehat\mu_{1, 0}(X)}\{\widehat\mu_{0, 0}(X)
    -\mu_{0, 0}(X)\}+
    \Bigg\{\dfrac{\widehat\mu_{1, 1}(X)}{\widehat\mu_{1, 0}(X)}-\dfrac{\mu_{1, 1}(X)}{\mu_{1, 0}(X)}\Bigg\}\mu_{0, 0}(X)+\\
    &\dfrac{\pi(X)}{1-\widehat \pi(X)}\dfrac{\widehat\mu_{1, 1}(X)}{\widehat\mu_{1, 0}(X)}\{\mu_{0, 0}(X)-\widehat\mu_{0, 0}(X)\}+
    \dfrac{\widehat\tau(X) q(X)\widehat\mu_{0, 0}(X)}{\tau(X) \widehat q(X)\widehat\mu_{1, 0}(X)}\{\mu_{1, 1}(X)-\widehat\mu_{1, 1}(X)\}+\\
    &\dfrac{1-q(X)}{1-\widehat q(X)}
    \dfrac{\widehat\tau(X) \widehat\mu_{0, 0}(X)\widehat\mu_{1, 1}(X)}{\tau(X) \widehat\mu_{1, 0}(X)^{2}}\{\widehat\mu_{1, 0}(X)-\mu_{1, 0}(X)\}
    \Bigg]
    \Bigg)\\
    =&(\widehat\kappa-\kappa)\alpha_{2}/\kappa + \widehat\kappa\\
    &\E\Bigg(\Pr(S=0|X)
    \Bigg[
    \dfrac{\widehat\mu_{1, 1}(X)}{\widehat\mu_{1, 0}(X)}\dfrac{\pi(X)-\widehat\pi(X)}{\pi(X)}\{\widehat\mu_{0, 0}(X)
    -\mu_{0, 0}(X)\}+\dfrac{\mu_{0, 0}(X)}{\mu_{1, 0}(X)}\\
    \end{align*}
\begin{align*}
    &\{\widehat\mu_{1, 1}(X)-\mu_{1, 1}(X)\}+
    \dfrac{\mu_{0, 0}(X)}{\mu_{1, 0}(X)}\dfrac{\widehat\mu_{1, 1}(X)}{\widehat\mu_{1, 0}(X)}\{\mu_{1, 0}(X)-\widehat\mu_{1, 0}(X)\}+
    \dfrac{\widehat\tau(X) q(X)\widehat\mu_{0, 0}(X)}{\tau(X) \widehat q(X)\widehat\mu_{1, 0}(X)}\\
    &\{\mu_{1, 1}(X)-\widehat\mu_{1, 1}(X)\}+
    \dfrac{1-q(X)}{1-\widehat q(X)}
    \dfrac{\widehat\tau(X) \widehat\mu_{0, 0}(X)\widehat\mu_{1, 1}(X)}{\tau(X) \widehat\mu_{1, 0}(X)^{2}}\{\widehat\mu_{1, 0}(X)-\mu_{1, 0}(X)\}
    \Bigg]
    \Bigg)\\
    =&(\widehat\kappa-\kappa)\alpha_{2}/\kappa + \widehat\kappa\\
    &\E\Bigg(\Pr(S=0|X)
    \Bigg[
    \dfrac{\widehat\mu_{1, 1}(X)}{\widehat\mu_{1, 0}(X)}\dfrac{\pi(X)-\widehat\pi(X)}{\pi(X)}\{\widehat\mu_{0, 0}(X)
    -\mu_{0, 0}(X)\}+\\
    &\Bigg\{
    \dfrac{\mu_{0, 0}(X)}{\mu_{1, 0}(X)}-\dfrac{\widehat\tau(X) q(X)\widehat\mu_{0, 0}(X)}{\tau(X) \widehat q(X)\widehat\mu_{1, 0}(X)}
    \Bigg\}
    \{\widehat\mu_{1, 1}(X)-\mu_{1, 1}(X)\}+\\
    &\dfrac{\widehat\mu_{1, 1}(X)}{\widehat\mu_{1, 0}(X)}
    \Bigg\{
    \dfrac{1-q(X)}{1-\widehat q(X)}\dfrac{\widehat\tau(X) \widehat\mu_{0, 0}(X)}{\tau(X) \widehat\mu_{1, 0}(X)}-\dfrac{\mu_{0, 0}(X)}{\mu_{1, 0}(X)}
    \Bigg\}
    \{\widehat\mu_{1, 0}(X)-\mu_{1, 0}(X)\}
    \Bigg]
    \Bigg)\\
    \lesssim
    &\left\lVert \widehat\mu_{0, 0}(X)-\mu_{0, 0}(X)\right\rVert\cdot
    \left\lVert \widehat\pi(X)-\pi(X)\right\rVert+
    \Big\{
    \left\lVert \widehat\mu_{1, 1}(X)-\mu_{1, 1}(X)\right\rVert+\left\lVert\widehat\mu_{1, 0}(X)-\mu_{1, 0}(X)\right\rVert
    \Big\}\\
    &+\Bigg\{
    \left\lVert \dfrac{\widehat\mu_{0, 0}(X)}{\widehat\mu_{1, 0}(X)}-\dfrac{\mu_{0, 0}(X)}{\mu_{1, 0}(X)}\right\rVert+\left\lVert\widehat q(X)-q(X)\right\rVert+\left\lVert\widehat\tau(X)-\tau(X)\right\rVert
    \Bigg\}.
\end{align*}
Alternatively, we have
\begin{align*}
     &\E\{H(\boldsymbol{\widehat \eta_2})-H(\boldsymbol{\eta_2})\}\\
    =&(\widehat\kappa-\kappa)\alpha_{2}/\kappa + \widehat\kappa\\
    &\E\Bigg(\Pr(S=0|X)
    \Bigg[
    \dfrac{\widehat\mu_{1, 1}(X)}{\widehat\mu_{1, 0}(X)}\dfrac{\pi(X)-\widehat\pi(X)}{\pi(X)}\{\widehat\mu_{0, 0}(X)
    -\mu_{0, 0}(X)\}+\Bigg\{\dfrac{\widehat\mu_{1, 1}(X)}{\widehat\mu_{1, 0}(X)}-\dfrac{\mu_{1, 1}(X)}{\mu_{1, 0}(X)}\Bigg\}\\
    &\mu_{0, 0}(X)+\dfrac{\widehat\mu_{1, 1}(X)}{\widehat\mu_{1, 0}(X)}
    \dfrac{\widehat\tau(X) q(X)}{\tau(X) \widehat q(X)}\widehat\mu_{0, 0}(X)-\dfrac{\widehat\mu_{1, 1}(X)}{\widehat\mu_{1, 0}(X)}
    \dfrac{\widehat\tau(X) q(X)}{\tau(X) \widehat q(X)}\widehat\mu_{0, 0}(X)+\\
    &\dfrac{1-q(X)}{1-\widehat q(X)}
    \dfrac{\widehat\tau(X) \widehat\mu_{0, 0}(X)\widehat\mu_{1, 1}(X)}{\tau(X) \widehat\mu_{1, 0}(X)^{2}}\{\widehat\mu_{1, 0}(X)-\mu_{1, 0}(X)\}
    \Bigg]
    \Bigg)
    \end{align*}
\begin{align*}
    =&(\widehat\kappa-\kappa)\alpha_{2}/\kappa + \widehat\kappa\\
    &\E\Bigg(\Pr(S=0|X)
    \Bigg[
    \dfrac{\widehat\mu_{1, 1}(X)}{\widehat\mu_{1, 0}(X)}\dfrac{\pi(X)-\widehat\pi(X)}{\pi(X)}\{\widehat\mu_{0, 0}(X)
    -\mu_{0, 0}(X)\}+\Bigg\{\dfrac{\widehat\mu_{1, 1}(X)}{\widehat\mu_{1, 0}(X)}-\dfrac{\mu_{1, 1}(X)}{\mu_{1, 0}(X)}\Bigg\}\\
    &\mu_{0, 0}(X)+\Bigg\{\dfrac{\mu_{1, 1}(X)}{\mu_{1, 0}(X)}-\dfrac{\widehat\mu_{1, 1}(X)}{\widehat\mu_{1, 0}(X)}\Bigg\}
    \dfrac{\widehat\tau(X) q(X)}{\tau(X) \widehat q(X)}\widehat\mu_{0, 0}(X)+\dfrac{\mu_{1, 1}(X)}{\mu_{1, 0}(X)}\dfrac{\widehat\tau(X) q(X)}{\tau(X) \widehat q(X)}\widehat\mu_{0, 0}(X)\\
    &-\dfrac{\mu_{1, 1}(X)}{\mu_{1, 0}(X)}\dfrac{\widehat\tau(X) q(X)}{\tau(X) \widehat q(X)}\widehat\mu_{0, 0}(X)+\Bigg\{\dfrac{\mu_{1, 1}(X)}{\mu_{1, 0}(X)}-\dfrac{\widehat\mu_{1, 1}(X)}{\widehat\mu_{1, 0}(X)}\Bigg\}\dfrac{1-q(X)}{1-\widehat q(X)}
    \dfrac{\widehat\tau(X) \widehat\mu_{0, 0}(X)}{\tau(X) \widehat\mu_{1, 0}(X)}\\
    &\{\widehat\mu_{1, 0}(X)-\mu_{1, 0}(X)\}+\dfrac{\mu_{1, 1}(X)}{\mu_{1, 0}(X)}\dfrac{1-q(X)}{1-\widehat q(X)}
    \dfrac{\widehat\tau(X) \widehat\mu_{0, 0}(X)}{\tau(X) \widehat\mu_{1, 0}(X)}\{\widehat\mu_{1, 0}(X)-\mu_{1, 0}(X)\}
    \Bigg]
    \Bigg)\\
    =&(\widehat\kappa-\kappa)\alpha_{2}/\kappa + \widehat\kappa\\
    &\E\Bigg(\Pr(S=0|X)
    \Bigg[
    \dfrac{\widehat\mu_{1, 1}(X)}{\widehat\mu_{1, 0}(X)}\dfrac{\pi(X)-\widehat\pi(X)}{\pi(X)}\{\widehat\mu_{0, 0}(X)
    -\mu_{0, 0}(X)\}+\\
    &\Bigg\{\dfrac{\widehat\mu_{1, 1}(X)}{\widehat\mu_{1, 0}(X)}-\dfrac{\mu_{1, 1}(X)}{\mu_{1, 0}(X)}\Bigg\}\dfrac{\widehat\tau(X) q(X)}{\tau(X) \widehat q(X)}\{\widehat\mu_{0, 0}(X)-\mu_{0, 0}(X)\}+\\
    &\Bigg\{\dfrac{\widehat\mu_{1, 1}(X)}{\widehat\mu_{1, 0}(X)}-\dfrac{\mu_{1, 1}(X)}{\mu_{1, 0}(X)}\Bigg\}\mu_{0, 0}(X)\Bigg[
    \dfrac{\widehat\tau(X) q(X)}{\tau(X) \widehat q(X)}-1-
    \dfrac{\widehat\tau(X) }{\tau(X)}\dfrac{1-q(X)}{1-\widehat q(X)}\dfrac{\widehat\mu_{0, 0}(X)}{\mu_{0, 0}(X)}\Big\{1-\dfrac{\mu_{1, 0}(X)}{\widehat\mu_{1, 0}(X)}\Big\}\Bigg]\\
    &+\dfrac{\mu_{1, 1}(X)}{\mu_{1, 0}(X)}\dfrac{q(X)}{\widehat q(X)}
    \dfrac{\widehat\tau(X) \widehat\mu_{0, 0}(X)}{\tau(X) \widehat\mu_{1, 0}(X)}\{\widehat\mu_{1, 0}(X)-\mu_{1, 0}(X)\}\Bigg\{
    1-\dfrac{1-q(X)}{1-\widehat q(X)}\dfrac{\widehat q(X)}{q(X)}
    \Bigg\}
    \Bigg]
    \Bigg)\\
    =&(\widehat\kappa-\kappa)\alpha_{2}/\kappa + \widehat\kappa\\
    &\E\Bigg\{\Pr(S=0|X)
    \Bigg(
    \{\widehat\mu_{0, 0}(X)-\mu_{0, 0}(X)\}\Bigg[
    \dfrac{\widehat\mu_{1, 1}(X)}{\widehat\mu_{1, 0}(X)}\Bigg\{1-\dfrac{\pi(X)}{\widehat\pi(X)}\Bigg\}+\\
    &\dfrac{\widehat\tau(X) }{\tau(X)}\dfrac{q(X)}{\widehat q(X)}\Bigg\{\dfrac{\mu_{1, 1}(X)}{\mu_{1, 0}(X)}-\dfrac{\widehat\mu_{1, 1}(X)}{\widehat\mu_{1, 0}(X)}\Bigg\}
    \Bigg]+\Bigg\{\dfrac{\mu_{1, 1}(X)}{\mu_{1, 0}(X)}-\dfrac{\widehat\mu_{1, 1}(X)}{\widehat\mu_{1, 0}(X)}
    \Bigg\}\mu_{0, 0}(X)\\
    &\Bigg[
    \dfrac{\widehat\tau(X) q(X)}{\tau(X) \widehat q(X)}-1-
    \dfrac{\widehat\tau(X) }{\tau(X)}\dfrac{1-q(X)}{1-\widehat q(X)}\dfrac{\widehat\mu_{0, 0}(X)}{\mu_{0, 0}(X)}\Big\{1-\dfrac{\mu_{1, 0}(X)}{\widehat\mu_{1, 0}(X)}\Big\}\Bigg]+\\
    &\dfrac{\mu_{1, 1}(X)}{\mu_{1, 0}(X)}\dfrac{q(X)}{\widehat q(X)}
    \dfrac{\widehat\tau(X) \widehat\mu_{0, 0}(X)}{\tau(X) \widehat\mu_{1, 0}(X)}\{\widehat\mu_{1, 0}(X)-\mu_{1, 0}(X)\}\Bigg\{
    1-\dfrac{1-q(X)}{1-\widehat q(X)}\dfrac{\widehat q(X)}{q(X)}
    \Bigg\}
    \Bigg)
    \Bigg\}\end{align*}
\begin{align*}
    \lesssim&(\widehat\kappa-\kappa)\alpha_{2}/\kappa + \\
    &\left\lVert \widehat\mu_{0, 0}(X)-\mu_{0, 0}(X)\right\rVert\cdot
    \left\lVert \widehat\pi(X)-\pi(X)\right\rVert+\\
    &\left\lVert \dfrac{\widehat\mu_{1, 1}(X)}{\widehat\mu_{1, 0}(X)}-\dfrac{\mu_{1, 1}(X)}{\mu_{1, 0}(X)}\right\rVert
    \Big\{
    \left\lVert \widehat\mu_{0, 0}(X)-\mu_{0, 0}(X)\right\rVert+
    \left\lVert\widehat\mu_{1, 0}(X)-\mu_{1, 0}(X)\right\rVert+\left\lVert\widehat\tau(X)-\tau(X)\right\rVert
    \Big\}\\
    &+\left\lVert\widehat q(X)-q(X)\right\rVert\Bigg\{
    \left\lVert \dfrac{\widehat\mu_{1, 1}(X)}{\widehat\mu_{1, 0}(X)}-\dfrac{\mu_{1, 1}(X)}{\mu_{1, 0}(X)}\right\rVert+\left\lVert\widehat\mu_{1, 0}(X)-\mu_{1, 0}(X)\right\rVert
    \Bigg\}.
\end{align*}
Combine the first term of 2 and 3, we have
\begin{align*}
     &\mathbb{P}\big\{H(\boldsymbol{\widehat \eta_2})-H(\boldsymbol{\eta_2})\big\}+
    (\mathbb{P}_n-\mathbb{P})H(\boldsymbol{\eta_2})\\
    =&\mathbb{P}_n\{A_2(O, \boldsymbol{\eta_2})\}+\Big\{\dfrac{\widehat\kappa}{\kappa} -1\Big\}\alpha_{1} +
    \Big\{\dfrac{\kappa}{\widehat\kappa} -1\Big\}\alpha_{1}. 
\end{align*}
Factoring and simplifying the last two terms yields $\alpha_{2} \left\{\widehat\kappa- \kappa\right\}\left\{\frac{1}{\kappa} - \frac{1}{\widehat\kappa}\right\}$, which is
$o_{p}(n^{-1/2})$ by the central limit theorem. 
Therefore, combining terms 1, 2 and 3, we conclude
\begin{align*}
     \widehat \alpha_{2} -\alpha_{2} 
    =O_{p}(R_{1, n})+\mathbb{P}_n\{A_2(O, \boldsymbol{\eta_2})\}+o_p(n^{-1/2}),
\end{align*}
In particular, if $R_{1, n}= o_p(n^{-1/2})$, then 
\begin{align*}
    \sqrt{n}\big(\widehat \alpha_{2} -\alpha_{2} \big)\rightsquigarrow\mathcal{N}\big[0, \E\{A_2(O, \boldsymbol{\eta_2})^{2}\}\big].
\end{align*}
That is, $\widehat \alpha_{2}$ is non-parametric efficient.

The derivation of the rest of the theorem is  an application of chain rules.
\end{proof}
\clearpage

\section{Theoretical results and discussion of Section \ref{sec:scenario3}}\label{appendix:scenario3}

Theorem \ref{thm:identification3} provides the identification results; the proof along with other identification strategies are given in \ref{appendix:identification3}.

\begin{restatable}{theorem}{identification3}
\label{thm:identification3}
Under conditions (A1) through (A5), (C1), and (C2), $\E[Y^1|S=0]$ can be identified by 
\begin{align*} 
   \alpha_{3}\equiv\E \left[\dfrac{\E[Y|X, S=1, A=1]}{\E[Y|X, S=1, A=0]}\E[Y | X, W, S = 0, A = 0]\right],
\end{align*}
and $\E[Y^0|S=0]$ can be identified by 
\begin{align*} 
   \beta_{3}\equiv\E\left[\E[Y | X, W, S = 0, A = 0]|S=0\right].
\end{align*}
Then $\E[Y^1|S=0]/\E[Y^0|S=0]$ and $\E[Y^1-Y^0|S=0]$ can be identified by estimable functionals  $\phi_{3}\equiv\alpha_{3}/\beta_{3}$ and $\psi_{3}\equiv\alpha_{3}-\beta_{3}$ respectively.
\end{restatable}

Define additional nuisance parameters $\mu'_{0, 0}(X, W)=\E[Y | X, W, S = 0, A = 0]$, $M(X)=\E \big[\E[Y | X, W, S = 0, A = 0] \big | X, S = 0 \big]$, and $\pi'(X, W)=\Pr[A=0|X, W, S=0]$.
% In addition, write $r'_S(X)=M(X)/\mu_{1, 0}(X)$. 
Denote the nuisance functions by $\boldsymbol{\eta_3}=\{\kappa, \mu'_{0, 0}(X, W), M(X), \mu_{1, 1}(X), \mu_{1, 0}(X), \pi'(X, W), q(X)\}$, $\boldsymbol{\eta_3}'=\{\kappa, \mu'_{0, 0}(X, W), \pi'(X, W)\}$, and their estimates by $\boldsymbol{\widehat \eta_3}$ and $\boldsymbol{\widehat \eta'_3}$ respectively. 
The following lemma gives the influence functions of the target parameters. The proof is given in \ref{appendix:IF3}.
\begin{lemma}
\label{thm:IF3}
The influence functions of $\alpha_{3}$, $\beta_{3}$, $\phi_{3}$, and $\psi_{3}$ are
\begin{align*}
    A_3(O, \boldsymbol{\eta_3})=&\kappa\Bigg(
    (1-S)
    \dfrac{\mu_{1, 1}(X)}{\mu_{1, 0}(X)}\Bigg[\mu'_{0, 0}(X, W)-\alpha_3 +\dfrac{1-A}{\pi'(X, W)}\big\{Y-\mu'_{0, 0}(X, W)\big\}
    \Bigg]\\
    &+
    S\tau(X)\dfrac{M(X)}{\mu_{1, 0}(X)}\Bigg[
    \dfrac{A}{q(X)}\big\{Y-\mu_{1, 1}(X)\big\}-\dfrac{1-A}{1-q(X)}\dfrac{\mu_{1, 1}(X)}{\mu_{1, 0}(X)}\big\{Y-\mu_{1, 0}(X)\big\}
    \Bigg]
    \Bigg).\\
    B_3(O, \boldsymbol{\eta_3}') =&\kappa(1-S)  \Bigg[
    \mu'_{0, 0}(X, W)-\beta_3  +\dfrac{1-A}{\pi'(X, W)}\big\{Y-\mu'_{0, 0}(X, W)\big\} 
    \Bigg],\\
    \Phi_{3}(O; \boldsymbol{\eta_{3}})=&\dfrac{1}{\beta_{3}}\{A_{3}(O; \boldsymbol{\eta_{3}})-\phi_{2} B_{3}(O; \boldsymbol{\eta'_{3}})\},\quad\quad
    \Psi_{3}'(O; \boldsymbol{\eta_{3}})=A_{3}(O; \boldsymbol{\eta_{3}})-B_{3}(O; \boldsymbol{\eta'_{3}}).
\end{align*}
These influence functions are also the semiparametric efficient influence functions when $\pi'(X, W)$ is known. 
\end{lemma}
Once influence functions are known, we can construct estimators in the same way as described in Section \ref{sec:estimation}. To establish asymptotic properties of estimators, we first provide the following conditions.
\begin{enumerate}
\item[(c1)] $\exists \varepsilon>0, \quad s.t. \quad \Pr[\varepsilon\leqslant \kappa\leqslant1-\varepsilon]=
\Pr[\varepsilon\leqslant \widehat \kappa\leqslant1-\varepsilon]=\Pr[\varepsilon\leqslant q(X)\leqslant1-\varepsilon]=
\Pr[\varepsilon\leqslant \widehat q(X)\leqslant1-\varepsilon]=
\Pr[\varepsilon\leqslant \tau(X)]=
\Pr[\varepsilon\leqslant \widehat \tau(X)]=
\Pr\{\varepsilon\leqslant \pi'(X)\leqslant1-\varepsilon\}=
\Pr\{\varepsilon\leqslant \widehat \pi'(X)\leqslant1-\varepsilon\}=1$,
\item[(c2)] $ \E[Y^{2}]<\infty$,
\item[(c3)] $\left\lVert \boldsymbol{\widehat \eta_3}-\boldsymbol{\eta_3} \right\rVert =o_p(1)$ 
%All the nuisance functions can be estimated with a rate $o_p(1)$.
\end{enumerate}
The following theorem give the asymptotic properties of the estimators; the proof is given in \ref{appendix:inference3}.
\begin{theorem}\label{thm:inference3}
If conditions \emph{(c1)} through \emph{(c3)} hold, then $\widehat \beta_3$ is consistent with rate of convergence $O_p(R_{3, n}^{\beta}+n^{-1/2})$, where  
$
    R_{2, n}^{\beta}=\left\lVert \widehat\mu'_{0, 0}(X, W)-\mu'_{0, 0}(X, W)\right\rVert
     \left\lVert \widehat\pi'(X, W)-\pi'(X, W)\right\rVert.
$
When $R_{3, n}^{\beta}= o_p(n^{-1/2})$, $\widehat \beta_3$ is asymptotically normal and non-parametric efficient. In addition, 
$\widehat \alpha_3, \widehat \phi_3, \widehat \psi_3$ are consistent with rate of convergence $O_p(R_{3, n}^{\alpha}+n^{-1/2})$, where  
$
    R_{3, n}^{\alpha}=\left\lVert \widehat \mu_{1, 1}(X)/\widehat \mu_{1, 0}(X)-\mu_{1, 1}(X)/ \mu_{1, 0}(X)\right\rVert
    (
    || \widehat M(X)-M(X) ||+
    \left\lVert \widehat \mu_{1, 0}(X)-\mu_{1, 0}(X) \right\rVert+
    || \widehat \tau(X)-\tau(X) ||
    )+R_{2, n}^{\beta}.
$
When $R_{3, n}^{\alpha}= o_p(n^{-1/2})$, $\widehat \alpha_3, \widehat \phi_3, \widehat \psi_3$ are asymptotically normal and non-parametric efficient.
\end{theorem}
Other alternative representations of $R_{3, n}^{\alpha}$ is also given in \ref{appendix:inference3}, which admits the same conclusion for robustness and efficiency when assumption \ref{assumption} holds.
\clearpage

%%%%%%%%%%%%%%%%%%%%%%%%%%%%%%%%%%%%%%%%%%%%%%%%%%%%%%%%%%%%
\section{Proof of Theorem \ref{thm:identification3}}\label{appendix:identification3}
%%%%%%%%%%%%%%%%%%%%%%%%%%%%%%%%%%%%%%%%%%%%%%%%%%%%%%%%%%%%

\begin{proof}
\begin{equation*}
    \begin{split}
        \E(Y^1 | S = 0) &= \E \big[\E(Y^1 | X, S = 0) \big | S = 0 \big] \\ 
         &= \E \left[ \dfrac{\E(Y^1 | X, S = 0)}{\E(Y^0 | X, S = 0)} \E(Y^0 | X, S = 0) \Big | S = 0 \right] \\ 
         &= \E \left[ \dfrac{\E(Y^1 | X, S = 1)}{\E(Y^0 | X, S = 1)} \E(Y^0 | X, S = 0) \Big | S = 0 \right] \\ 
          &= \E \left[ \dfrac{\E(Y^1 | X, S = 1, A = 1)}{\E(Y^0 | X, S = 1, A = 0)} \E(Y^0 | X, S = 0) \Big | S = 0 \right] \\ 
         &= \E \left[ \dfrac{\E(Y^1 | X, S = 1, A = 1)}{\E(Y^0 | X, S = 1, A = 0)} \E\big\{ \E(Y^0 | X, W, S = 0) \big | X, S = 0 \big\} \Big | S = 0 \right]\\
         &= \E \left[ \dfrac{\E(Y^1 | X, S = 1, A = 1)}{\E(Y^0 | X, S = 1, A = 0)} \E\big\{ \E(Y^0 | X, W, S = 0, A = 0) \big | X, S = 0 \big\} \Big | S = 0 \right] \\
         &= \E \left[\dfrac{\E(Y | X, S = 1, A = 1)}{\E(Y | X, S = 1, A = 0)} \E\big\{ \E(Y | X, W, S = 0, A = 0) \big | X, S = 0 \big\} \Big | S = 0 \right]\\
    &=\E \left[\E\left[\dfrac{\E(Y | X, S = 1, A = 1)}{\E(Y | X, S = 1, A = 0)} \E(Y | X, W, S = 0, A = 0) \big | X, S = 0 \right] \Big | S = 0 \right]\\
    &=\E\left[\dfrac{\E(Y | X, S = 1, A = 1)}{\E(Y | X, S = 1, A = 0)} \E(Y | X, W, S = 0, A = 0) \Big | S = 0 \right]\\
    \end{split}
\end{equation*}
Here, the first step follows from the law of total expectation; the second from multiplying and dividing by $\E(Y^0 | X, S = 0)$, which is allowed by the positivity component of condition (A4); the third from the exchangeability component of condition (A4); the fourth from condition (A2); the fifth from the law of total expectation; the sixth by condition (C1); the seventh from condition (A1).

Furthermore, using the law of total expectation and conditions (A1), (C1), and (C2) we can write, 
\begin{equation*}
    \begin{split}
        \E(Y^0 | S = 0) &= \E \big\{  \E(Y^0 | X, W,  S = 0)  \big | S = 0 \big\} \\
        &= \E \big\{  \E(Y^0 | X, W,  S = 0, A = 0)  \big | S = 0 \big\} \\
        &= \E \big\{  \E(Y | X, W, S = 0, A = 0)  \big | S = 0 \big\}.
    \end{split}
\end{equation*}
\noindent
The target parameter $\beta_3$ can also be written as
$\E \left[ \dfrac{1-A}{\Pr[A=0|X, W, S=0]}Y| S = 0 \right]$
The derivation of the weighting estimators follows the same logic as \ref{appendix:identification1}.

We can then identify the causal mean ratio in the target population by $\phi_{3}\equiv\alpha_3/\beta_{3}$ and the causal mean difference in the target population by $\psi_{3}\equiv\alpha_3-\beta_{3}$.
\end{proof}

\begin{prop}
    The target parameter $\alpha_3$ can also be written as 
    \begin{align*}
&\dfrac{1}{\Pr(S=0)}\E \left[ \dfrac{\E(Y | X, S = 1, A = 1)}{\E(Y | X, S = 1, A = 0)}\dfrac{1-A}{\Pr[A=0|X, W, S=0]} Y \right], \\
   &\dfrac{1}{\Pr(S=0)}\E \left[S\dfrac{\Pr[S=0|X]}{\Pr[S=1|X]}\dfrac{A}{\Pr[A=1|X, S=1]}\dfrac{\E \left[ \E[Y | X, W, S = 0, A = 0] \big | X, S = 0 \right]}{\E[Y | X, S = 1, A = 0]}Y  \right], \\
   &\dfrac{1}{\Pr(S=0)}\E \Big[S\dfrac{\Pr[S=0|X]}{\Pr[S=1|X]}\dfrac{1-A}{1-\Pr[A=1|X, S=1]}\dfrac{\E(Y | X, S = 1, A = 1)}{\E(Y | X, S = 1, A = 0)}\\&\hspace{3.5cm}\dfrac{\E \left[ \E[Y | X, W, S = 0, A = 0] \big | X, S = 0 \right]}{\E[Y | X, S = 1, A = 0]}Y  \Big].
\end{align*}
\end{prop}

The derivation of the weighting estimators follows the same logic as \ref{appendix:identification1}.
\clearpage

%%%%%%%%%%%%%%%%%%%%%%%%%%%%%%%%%%%%%%%%%%%%%%%%%%%%%%%%%%%%
\section{Proof of Lemma \ref{thm:IF3}}
%%%%%%%%%%%%%%%%%%%%%%%%%%%%%%%%%%%%%%%%%%%%%%%%%%%%%%%%%%%%

\label{appendix:IF3}
\begin{proof}
Recall that
\begin{equation*} 
  \alpha_3=\E\Big\{\dfrac{\mu_{1, 1}(X)}{\mu_{1, 0}(X)}M(X)\Big|S=0\Big\}.
\end{equation*}
We then have
\begin{align*}
     &\dfrac{\partial \alpha_{3, p_t}}{\partial t}\Big|_{t=0}\\
    =&\dfrac{\partial}{\partial t}\E_{p_t}\Bigg\{\dfrac{\mu_{1, 1, p_t}(X)}{\mu_{1, 0, p_t}(X)}M_{p_t}(X)\Big|S=0\Bigg\}\Big|_{t=0}\\
    =&\underbrace{\dfrac{\partial}{\partial t}\E_{p_{0}}\Bigg\{\dfrac{\mu_{1, 1, p_0}(X)}{\mu_{1, 0, p_0}(X)}M_{p_0}(X)\Big|S=0\Bigg\}\Big|_{t=0}}_{(1)}+
    \underbrace{\E_{p_{0}}\Bigg\{\dfrac{\dfrac{\partial}{\partial t}\mu_{1, 1, p_t}(X)\Big|_{t=0}}{\mu_{1, 0, p_0}(X)}M_{p_0}(X)\Big|S=0\Bigg\}}_{(2)}+\\
    &\underbrace{\E_{p_{t}}\Bigg\{\dfrac{\mu_{1, 1, p_0}(X)}{\dfrac{\partial}{\partial t}\mu_{1, 0, p_t}(X)\Big|_{t=0}}M_{p_0}(X)\Big|S=0\Big\}}_{(3)}+
    \underbrace{\E_{p_{0}}\Big\{\dfrac{\mu_{1, 1, p_0}(X)}{\mu_{1, 0, p_0}(X)}\dfrac{\partial}{\partial t}M_{p_t}(X)\Big|_{t=0}\Big|S=0\Bigg\}}_{(4)}.
\end{align*}
Next, we calculate each term by turn. The derivation of terms (1), (2), and (3) are similar to ones in \ref{appendix:IF2}, and (1)+(2)+(3) is
\begin{align*}
    &\E\Bigg\{\kappa\Bigg(
    (1-S)
    \dfrac{\mu_{1, 1}(X)}{\mu_{1, 0}(X)}\Big\{M(X)-\alpha_3 
    \Big\}+
    S\tau(X)\dfrac{M(X)}{\mu_{1, 0}(X)}\\
    &\Bigg[
    \dfrac{A}{q(X)}\big\{Y-\mu_{1, 1}(X)\big\}-\dfrac{1-A}{1-q(X)}\dfrac{\mu_{1, 1}(X)}{\mu_{1, 0}(X)}\big\{Y-\mu_{1, 0}(X)\big\}
    \Bigg]
    \Bigg)g(O)\Bigg\}.
\end{align*}
For the term (4), it is 
\begin{align*}
    &\E_{p_{0}}\Bigg\{\dfrac{\mu_{1, 1, p_0}(X)}{\mu_{1, 0, p_0}(X)}\dfrac{\partial}{\partial t}\E_{p_{t}} \big\{ \mu'_{0, 0, p_0}(X, W)\big | X, S = 0 \big\}\Big|_{t=0}\Big|S=0\Bigg\}\\
    +&\E_{p_{0}}\Bigg\{\dfrac{\mu_{1, 1, p_0}(X)}{\mu_{1, 0, p_0}(X)}\E_{p_{0}} \Big\{\dfrac{\partial}{\partial t} \mu'_{0, 0, p_t}(X, W)\Big|_{t=0}\big | X, S = 0 \Big\}\Big|S=0\Bigg\}\\
    =&\E\Big\{g(O)\kappa(1-S)
    \dfrac{\mu_{1, 1}(X)}{\mu_{1, 0}(X)}\Big[\mu'_{0, 0, p_0}(X, W)-M(X)+\dfrac{1-A}{\pi'(X, W)}\big\{Y-\mu'_{0, 0}(X, W)\big\}\Big]\Big\}.
\end{align*}
After some algebra, we conclude that the influence function of $\alpha_3 $ is $A(O, \boldsymbol{\eta_3})$.

It is trivial to derive the influence function of $\beta_{3}$ given the above proof. By the chain rule, it is easy to derive the influence function of $\phi_{3}$ and $\psi_{3}$. The rest of proof follows from the proof for Lemma \ref{thm:IF2}.
\end{proof}
\clearpage

%%%%%%%%%%%%%%%%%%%%%%%%%%%%%%%%%%%%%%%%%%%%%%%%%%%%%%%%%%%%
\section{Proof of Theorem \ref{thm:inference3}}\label{appendix:inference3}
%%%%%%%%%%%%%%%%%%%%%%%%%%%%%%%%%%%%%%%%%%%%%%%%%%%%%%%%%%%%

\begin{proof}
For general functions $\boldsymbol{\widetilde\eta_3}$ define 
\begin{align*}
H(\boldsymbol{\widetilde\eta_3}) =  &\kappa\Bigg(
    (1-S)
    \dfrac{\widetilde\mu_{1, 1}(X)}{\widetilde\mu_{1, 0}(X)}\Bigg[\widetilde\mu'_{0, 0}(X, W) +\dfrac{1-A}{1-\widetilde\pi'(X, W)}\big\{Y-\widetilde\mu'_{0, 0}(X, W)\big\}
    \Bigg]\\
    &+
    S\widetilde\tau(X)\dfrac{\widetilde M_{0, 0}(X, W)}{\widetilde\mu_{1, 0}(X)}\Bigg[
    \dfrac{A}{\widetilde q(X)}\big\{Y-\widetilde \mu_{1, 1}(X)\big\}-\dfrac{1-A}{1-\widetilde q(X)}\dfrac{\widetilde \mu_{1, 1}(X)}{\widetilde \mu_{1, 0}(X)}\big\{Y-\widetilde \mu_{1, 0}(X)\big\}
    \Bigg]
    \Bigg).
\end{align*}
We observe that we can decompose $\widehat \alpha_3  - \alpha_3 $ into three parts as below.
\begin{align*}
     &\widehat \alpha_3  - \alpha_3 =\mathbb{P}_n\big\{H(\boldsymbol{\widehat \eta_3})\big\}-\mathbb{P}\big\{H(\boldsymbol{\eta_3})\big\}\\
    =&\underbrace{
    (\mathbb{P}_n-\mathbb{P})\big\{H(\boldsymbol{\widehat \eta_3})-H(\boldsymbol{\eta_3})\big\} }_{1}+
    \underbrace{
    \mathbb{P}\big\{H(\boldsymbol{\widehat \eta_3})-H(\boldsymbol{\eta_3})\big\} }_{2}+
    \underbrace{
    (\mathbb{P}_n-\mathbb{P})H(\boldsymbol{\eta_3})}_{3}.
\end{align*}
Working on term 1, note that
\begin{align*}
    &||H(\boldsymbol{\widehat \eta_3})-H(\boldsymbol{\eta_3})||\\
    = 
    & \Bigg|\Bigg|\widehat\kappa\Bigg(
    (1-S)\Bigg[
    \dfrac{\widehat\mu_{1, 1}(X)}{\widehat\mu_{1, 0}(X)}\widehat \mu'_{0, 0}(X, W)+\dfrac{1-A}{1-\widehat\pi'(X, W)}\dfrac{\widehat\mu_{1, 1}(X)}{\widehat\mu_{1, 0}(X, W)}\big\{Y-\widehat \mu'_{0, 0}(X)\big\}
    \Bigg]\\
    &+
    S\widehat\tau(X)\dfrac{\widehat M(X)}{\widehat\mu_{1, 0}(X)}\Bigg[
    \dfrac{A}{\widehat q(X)}\big\{Y-\widehat\mu_{1, 1}(X)\big\}-\dfrac{1-A}{1-\widehat q(X)}\dfrac{\widehat\mu_{1, 1}(X)}{\widehat\mu_{1, 0}(X)}\big\{Y-\widehat\mu_{1, 0}(X)\big\}
    \Bigg]
    \Bigg)-\\
    &\kappa\Bigg(
    (1-S)\Bigg[
    \dfrac{\mu_{1, 1}(X)}{\mu_{1, 0}(X)}\mu'_{0, 0}(X, W)+\dfrac{1-A}{1-\pi'(X, W)}\dfrac{\mu_{1, 1}(X)}{\mu_{1, 0}(X)}\big\{Y-\mu'_{0, 0}(X, W)\big\}
    \Bigg]\\
    &+
    S\tau(X)\dfrac{M(X)}{\mu_{1, 0}(X)}\Bigg[
    \dfrac{A}{q(X)}\big\{Y-\mu_{1, 1}(X)\big\}-\dfrac{1-A}{1-q(X)}\dfrac{\mu_{1, 1}(X)}{\mu_{1, 0}(X)}\big\{Y-\mu_{1, 0}(X)\big\}
    \Bigg]
    \Bigg)\Bigg|\Bigg|\\
    \leqslant &
    \Bigg|\Bigg|
    \widehat\kappa\Bigg((1-S)
    \Bigg[\Big\{\dfrac{\widehat\mu_{1, 1}(X)}{\widehat\mu_{1, 0}(X)}\widehat\mu'_{0, 0}(X, W)-\dfrac{\mu_{1, 1}(X)}{\mu_{1, 0}(X)}\mu'_{0, 0}(X, W)\Big\}+\\
    &\dfrac{1-A}{1-\widehat \pi'(X, W)}\dfrac{\widehat\mu_{1, 1}(X)}{\widehat\mu_{1, 0}(X)}\{\mu'_{0, 0}(X, W)-\widehat\mu'_{0, 0}(X, W)\}+\\
    &S\widehat\tau(X)\dfrac{\widehat M(X)}{\widehat\mu_{1, 0}(X)}\Big[
    \dfrac{A}{\widehat q(X)}\{\widehat\mu_{1, 1}(X)-\mu_{1, 1}(X)\}-\dfrac{1-A}{1-\widehat q(X)}\dfrac{\widehat\mu_{1, 1}(X)}{\widehat\mu_{0, 0}(X)}\{\widehat\mu_{1, 0}(X)-\mu_{1, 0}(X)\}
    \Big]
    \Bigg]    
    \Bigg)
    \Bigg|\Bigg|\\
    &+\Bigg|\Bigg|
    (\widehat\kappa-\kappa)\Bigg(
    (1-S)\Bigg[
    \dfrac{\mu_{1, 1}(X)}{\mu_{1, 0}(X)}\mu'_{0, 0}(X, W)+\dfrac{1-A}{1-\pi'(X, W)}\dfrac{\mu_{1, 1}(X)}{\mu_{1, 0}(X)}\big\{Y-\mu'_{0, 0}(X, W)\big\}
    \Bigg]+\\
    &
    S\tau(X)\dfrac{M(X)}{\mu_{1, 0}(X)}\Bigg[
    \dfrac{A}{q(X)}\big\{Y-\mu_{1, 1}(X)\big\}-\dfrac{1-A}{1-q(X)}\dfrac{\mu_{1, 1}(X)}{\mu_{1, 0}(X)}\big\{Y-\mu_{1, 0}(X)\big\}
    \Bigg]
    \Bigg)
    \Bigg|\Bigg|+\\
    &\Bigg|\Bigg|
    \widehat \kappa \Bigg((1-A)(1-S)
    \Big[\dfrac{\widehat\mu_{1, 1}(X)}{\widehat\mu_{1, 0}(X)\{1-\widehat\pi'(X, W)\}}-\dfrac{\mu_{1, 1}(X)}{\mu_{1, 0}(X)\{1-\pi'(X, W)\}}
    \Big]\{Y-\mu'_{0, 0}(X, W)\}+\\
    &AS\Big\{\dfrac{\widehat M(X)\widehat\tau(X)}{\widehat\mu_{1, 0}(X)\widehat q(X)}-\dfrac{ M(X)\tau(X)}{\mu_{1, 0}(X) q(X)}\Big\}\{Y-\mu_{1, 1}(X)\}
    +\\
    &(1-A)S\Big[\dfrac{\widehat\mu_{1, 1}(X)\widehat M(X)\widehat\tau(X)}{\widehat\mu_{1, 0}(X)^{2}\{1-\widehat q(X)\}}-\dfrac{\mu_{1, 1}(X) M(X)\tau(X)}{\mu_{1, 0}(X)^{2}\{1- q(X)\}}\Big]\{Y-\mu_{1, 0}(X)\}
    \Bigg)
    \Bigg|\Bigg|        \end{align*}
    \begin{align*}
    \lesssim
    &||\widehat M(X)- M(X)||+||\widehat\mu'_{0, 0}(X, W)-\mu'_{0, 0}(X, W)||+||\widehat\mu_{1, 0}(X)-\mu_{1, 0}(X)||+\\
    &||\widehat\mu_{1, 1}(X)-\mu_{1, 1}(X)||+||\widehat \tau(X)-\tau(X)||+||\widehat q(X)-q(X)||+||\widehat \pi'(X, W)-\pi'(X, W)||+||\widehat\kappa-\kappa||.
\end{align*}
Because $\kappa^{-1}$ is an empirical average of $\Pr(S=0)$, it is trivial to see that
i) $I(\widehat \kappa^{-1}=0)=o_p(n^{-1/2})$, and ii) $\widehat \kappa\xrightarrow{P} \kappa$.
Note that we also assume $||\widehat q(X)-q(X)||=o_p(n^{-1/2})$. Therefore, 
by conditions \emph{(a3)}, we have $||H(\boldsymbol{\widehat \eta_3})-H(\boldsymbol{\eta_3})|| = o_p(1)$, so that by Lemma 2 of \citet{kennedy2020sharp}, $(\mathbb{P}_n-\mathbb{P})\big\{H(\boldsymbol{\widehat \eta_3})-H(\boldsymbol{\eta_3})\big\}=o_p(n^{-1/2})$.\\
Working on term 2, noticing $\E\{\mu'_{0, 0}(X, W)\}=\E\{M(X)\}$, we have
\begin{align*}
     &\E\{H(\boldsymbol{\widehat \eta_3})-H(\boldsymbol{\eta_3})\}\\
    =&(\widehat\kappa-\kappa)\alpha_{1}/\kappa + \widehat\kappa\\
    &\E\Bigg(\Pr(S=0|X)
    \Bigg[
    \dfrac{\widehat\mu_{1, 1}(X)}{\widehat\mu_{1, 0}(X)}\widehat\mu'_{0, 0}(X, W)-\dfrac{\mu_{1, 1}(X)}{\mu_{1, 0}(X)}\mu'_{0, 0}(X, W)+\dfrac{1-\pi'(X, W)}{1-\widehat\pi'(X, W)}\dfrac{\widehat\mu_{1, 1}(X)}{\widehat\mu_{1, 0}(X)}\\
    &\{\mu'_{0, 0}(X, W)-\widehat\mu'_{0, 0}(X, W)\}\Bigg]+
    \Pr(S=1|X)\widehat\tau(X)\dfrac{\widehat M(X)}{\widehat\mu_{1, 0}(X)}
    \Bigg[
    \dfrac{q(X)}{\widehat q(X)}\{\mu_{1, 1}(X)-\widehat\mu_{1, 1}(X)\}+\\
    &\dfrac{1-q(X)}{1-\widehat q(X)}\dfrac{\widehat\mu_{1, 1}(X)}{\widehat\mu_{1, 0}(X)}\{\mu_{1, 0}(X)-\widehat\mu_{1, 0}(X)\}
    \Bigg]
    \Bigg)\\
    =&(\widehat\kappa-\kappa)\alpha_{1}/\kappa + \widehat\kappa\\
    &\E\Bigg(\Pr(S=0|X)
    \Bigg[
    \dfrac{\widehat\mu_{1, 1}(X)}{\widehat\mu_{1, 0}(X)}\{\widehat\mu'_{0, 0}(X, W)
    -\mu'_{0, 0}(X, W)\}+
    \Bigg\{\dfrac{\widehat\mu_{1, 1}(X)}{\widehat\mu_{1, 0}(X)}-\dfrac{\mu_{1, 1}(X)}{\mu_{1, 0}(X)}\Bigg\}\mu'_{0, 0}(X, W)+\\
    &\dfrac{1-\pi'(X, W)}{1-\widehat \pi'(X, W)}\dfrac{\widehat\mu_{1, 1}(X)}{\widehat\mu_{1, 0}(X)}\{\mu'_{0, 0}(X, W)-\widehat\mu'_{0, 0}(X, W)\}+
    \dfrac{\widehat\tau(X) q(X)\widehat M(X)}{\tau(X) \widehat q(X)\widehat\mu_{1, 0}(X)}\{\mu_{1, 1}(X)-\widehat\mu_{1, 1}(X)\}+\\
    &\dfrac{1-q(X)}{1-\widehat q(X)}
    \dfrac{\widehat\tau(X) \widehat M(X)\widehat\mu_{1, 1}(X)}{\tau(X) \widehat\mu_{1, 0}(X)^{2}}\{\widehat\mu_{1, 0}(X)-\mu_{1, 0}(X)\}
    \Bigg]
    \Bigg)            \end{align*}
    \begin{align*}
    =&(\widehat\kappa-\kappa)\alpha_{1}/\kappa + \widehat\kappa\\
    &\E\Bigg(\Pr(S=0|X)
    \Bigg[
    \dfrac{\widehat\mu_{1, 1}(X)}{\widehat\mu_{1, 0}(X)}\dfrac{\pi'(X, W)-\widehat\pi'(X, W)}{1-\pi'(X, W)}\{\widehat\mu'_{0, 0}(X, W)
    -\mu'_{0, 0}(X, W)\}+\dfrac{\mu'_{0, 0}(X, W)}{\mu_{1, 0}(X)}\\
    &\{\widehat\mu_{1, 1}(X)-\mu_{1, 1}(X)\}+
    \dfrac{\mu'_{0, 0}(X, W)}{\mu_{1, 0}(X)}\dfrac{\widehat\mu_{1, 1}(X)}{\widehat\mu_{1, 0}(X)}\{\mu_{1, 0}(X)-\widehat\mu_{1, 0}(X)\}+
    \dfrac{\widehat\tau(X) q(X)\widehat M(X)}{\tau(X) \widehat q(X)\widehat\mu_{1, 0}(X)}\\
    &\{\mu_{1, 1}(X)-\widehat\mu_{1, 1}(X)\}+
    \dfrac{1-q(X)}{1-\widehat q(X)}
    \dfrac{\widehat\tau(X) \widehat M(X)\widehat\mu_{1, 1}(X)}{\tau(X) \widehat\mu_{1, 0}(X)^{2}}\{\widehat\mu_{1, 0}(X)-\mu_{1, 0}(X)\}
    \Bigg]
    \Bigg)\\
    =&(\widehat\kappa-\kappa)\alpha_{1}/\kappa + \widehat\kappa\\
    &\E\Bigg(\Pr(S=0|X)
    \Bigg[
    \dfrac{\widehat\mu_{1, 1}(X)}{\widehat\mu_{1, 0}(X)}\dfrac{\pi'(X, W)-\widehat\pi'(X, W)}{1-\pi'(X, W)}\{\widehat\mu'_{0, 0}(X, W)
    -\mu'_{0, 0}(X, W)\}+\\
    &\Bigg\{
    \dfrac{\mu_{0, 0}(X)}{\mu_{1, 0}(X)}-\dfrac{\widehat\tau(X) q(X)\widehat\mu_{0, 0}(X)}{\tau(X) \widehat q(X)\widehat\mu_{1, 0}(X)}
    \Bigg\}
    \{\widehat\mu_{1, 1}(X)-\mu_{1, 1}(X)\}+\\
    &\dfrac{\widehat\mu_{1, 1}(X)}{\widehat\mu_{1, 0}(X)}
    \Bigg\{
    \dfrac{1-q(X)}{1-\widehat q(X)}\dfrac{\widehat\tau(X) \widehat\mu_{0, 0}(X)}{\tau(X) \widehat\mu_{1, 0}(X)}-\dfrac{\mu_{0, 0}(X)}{\mu_{1, 0}(X)}
    \Bigg\}
    \{\widehat\mu_{1, 0}(X)-\mu_{1, 0}(X)\}
    \Bigg]
    \Bigg)\\
    \lesssim
    &\left\lVert \widehat\mu_{0, 0}(X)-\mu_{0, 0}(X)\right\rVert\cdot
    \left\lVert \widehat\pi'(X, W)-\pi'(X, W)\right\rVert+
    \Big\{
    \left\lVert \widehat\mu_{1, 1}(X)-\mu_{1, 1}(X)\right\rVert+\left\lVert\widehat\mu_{1, 0}(X)-\mu_{1, 0}(X)\right\rVert
    \Big\}\\
    &+\Bigg\{
    \left\lVert \dfrac{\widehat\mu_{0, 0}(X)}{\widehat\mu_{1, 0}(X)}-\dfrac{\mu_{0, 0}(X)}{\mu_{1, 0}(X)}\right\rVert+\left\lVert\widehat q(X)-q(X)\right\rVert+\left\lVert\widehat\tau(X)-\tau(X)\right\rVert
    \Bigg\}.
\end{align*}
Alternatively, we have
\begin{align*}
     &\E\{H(\boldsymbol{\widehat \eta_3})-H(\boldsymbol{\eta_3})\}\\
    =&(\widehat\kappa-\kappa)\alpha_{1}/\kappa + \widehat\kappa\\
    &\E\Bigg(\Pr(S=0|X)
    \Bigg[
    \dfrac{\widehat\mu_{1, 1}(X)}{\widehat\mu_{1, 0}(X)}\dfrac{\pi'(X, W)-\widehat\pi'(X, W)}{1-\pi'(X, W)}\{\widehat\mu_{0, 0}(X)
    -\mu_{0, 0}(X)\}+\Bigg\{\dfrac{\widehat\mu_{1, 1}(X)}{\widehat\mu_{1, 0}(X)}-\dfrac{\mu_{1, 1}(X)}{\mu_{1, 0}(X)}\Bigg\}\\
    &\mu_{0, 0}(X)+\dfrac{\widehat\mu_{1, 1}(X)}{\widehat\mu_{1, 0}(X)}
    \dfrac{\widehat\tau(X) q(X)}{\tau(X) \widehat q(X)}\widehat\mu_{0, 0}(X)-\dfrac{\widehat\mu_{1, 1}(X)}{\widehat\mu_{1, 0}(X)}
    \dfrac{\widehat\tau(X) q(X)}{\tau(X) \widehat q(X)}\widehat\mu_{0, 0}(X)+
                \end{align*}
    \begin{align*}
    &\dfrac{1-q(X)}{1-\widehat q(X)}
    \dfrac{\widehat\tau(X) \widehat\mu_{0, 0}(X)\widehat\mu_{1, 1}(X)}{\tau(X) \widehat\mu_{1, 0}(X)^{2}}\{\widehat\mu_{1, 0}(X)-\mu_{1, 0}(X)\}
    \Bigg]
    \Bigg)\\
    =&(\widehat\kappa-\kappa)\alpha_{1}/\kappa + \widehat\kappa\\
    &\E\Bigg(\Pr(S=0|X)
    \Bigg[
    \dfrac{\widehat\mu_{1, 1}(X)}{\widehat\mu_{1, 0}(X)}\dfrac{\pi'(X, W)-\widehat\pi'(X, W)}{1-\pi'(X, W)}\{\widehat\mu_{0, 0}(X)
    -\mu_{0, 0}(X)\}+\Bigg\{\dfrac{\widehat\mu_{1, 1}(X)}{\widehat\mu_{1, 0}(X)}-\dfrac{\mu_{1, 1}(X)}{\mu_{1, 0}(X)}\Bigg\}\\
    &\mu_{0, 0}(X)+\Bigg\{\dfrac{\mu_{1, 1}(X)}{\mu_{1, 0}(X)}-\dfrac{\widehat\mu_{1, 1}(X)}{\widehat\mu_{1, 0}(X)}\Bigg\}
    \dfrac{\widehat\tau(X) q(X)}{\tau(X) \widehat q(X)}\widehat\mu_{0, 0}(X)+\dfrac{\mu_{1, 1}(X)}{\mu_{1, 0}(X)}\dfrac{\widehat\tau(X) q(X)}{\tau(X) \widehat q(X)}\widehat\mu_{0, 0}(X)\\
    &-\dfrac{\mu_{1, 1}(X)}{\mu_{1, 0}(X)}\dfrac{\widehat\tau(X) q(X)}{\tau(X) \widehat q(X)}\widehat\mu_{0, 0}(X)+\Bigg\{\dfrac{\mu_{1, 1}(X)}{\mu_{1, 0}(X)}-\dfrac{\widehat\mu_{1, 1}(X)}{\widehat\mu_{1, 0}(X)}\Bigg\}\dfrac{1-q(X)}{1-\widehat q(X)}
    \dfrac{\widehat\tau(X) \widehat\mu_{0, 0}(X)}{\tau(X) \widehat\mu_{1, 0}(X)}\\
    &\{\widehat\mu_{1, 0}(X)-\mu_{1, 0}(X)\}+\dfrac{\mu_{1, 1}(X)}{\mu_{1, 0}(X)}\dfrac{1-q(X)}{1-\widehat q(X)}
    \dfrac{\widehat\tau(X) \widehat\mu_{0, 0}(X)}{\tau(X) \widehat\mu_{1, 0}(X)}\{\widehat\mu_{1, 0}(X)-\mu_{1, 0}(X)\}
    \Bigg]
    \Bigg)\\
    =&(\widehat\kappa-\kappa)\alpha_{1}/\kappa + \widehat\kappa\\
    &\E\Bigg(\Pr(S=0|X)
    \Bigg[
    \dfrac{\widehat\mu_{1, 1}(X)}{\widehat\mu_{1, 0}(X)}\dfrac{\pi'(X, W)-\widehat\pi'(X, W)}{1-\pi'(X, W)}\{\widehat\mu_{0, 0}(X)
    -\mu_{0, 0}(X)\}+\\
    &\Bigg\{\dfrac{\widehat\mu_{1, 1}(X)}{\widehat\mu_{1, 0}(X)}-\dfrac{\mu_{1, 1}(X)}{\mu_{1, 0}(X)}\Bigg\}\dfrac{\widehat\tau(X) q(X)}{\tau(X) \widehat q(X)}\{\widehat\mu_{0, 0}(X)-\mu_{0, 0}(X)\}+\\
    &\Bigg\{\dfrac{\widehat\mu_{1, 1}(X)}{\widehat\mu_{1, 0}(X)}-\dfrac{\mu_{1, 1}(X)}{\mu_{1, 0}(X)}\Bigg\}\mu_{0, 0}(X)\Bigg[
    \dfrac{\widehat\tau(X) q(X)}{\tau(X) \widehat q(X)}-1-
    \dfrac{\widehat\tau(X) }{\tau(X)}\dfrac{1-q(X)}{1-\widehat q(X)}\dfrac{\widehat\mu_{0, 0}(X)}{\mu_{0, 0}(X)}\Big\{1-\dfrac{\mu_{1, 0}(X)}{\widehat\mu_{1, 0}(X)}\Big\}\Bigg]\\
    &+\dfrac{\mu_{1, 1}(X)}{\mu_{1, 0}(X)}\dfrac{q(X)}{\widehat q(X)}
    \dfrac{\widehat\tau(X) \widehat\mu_{0, 0}(X)}{\tau(X) \widehat\mu_{1, 0}(X)}\{\widehat\mu_{1, 0}(X)-\mu_{1, 0}(X)\}\Bigg\{
    1-\dfrac{1-q(X)}{1-\widehat q(X)}\dfrac{\widehat q(X)}{q(X)}
    \Bigg\}
    \Bigg]
    \Bigg)\\
    =&(\widehat\kappa-\kappa)\alpha_{1}/\kappa + \widehat\kappa\\
    &\E\Bigg\{\Pr(S=0|X)
    \Bigg(
    \{\widehat\mu_{0, 0}(X)-\mu_{0, 0}(X)\}\Bigg[
    \dfrac{\widehat\mu_{1, 1}(X)}{\widehat\mu_{1, 0}(X)}\Bigg\{1-\dfrac{1-\pi'(X, W)}{1-\widehat\pi'(X, W)}\Bigg\}+\\
    &\dfrac{\widehat\tau(X) }{\tau(X)}\dfrac{q(X)}{\widehat q(X)}\Bigg\{\dfrac{\mu_{1, 1}(X)}{\mu_{1, 0}(X)}-\dfrac{\widehat\mu_{1, 1}(X)}{\widehat\mu_{1, 0}(X)}\Bigg\}
    \Bigg]+\Bigg\{\dfrac{\mu_{1, 1}(X)}{\mu_{1, 0}(X)}-\dfrac{\widehat\mu_{1, 1}(X)}{\widehat\mu_{1, 0}(X)}
    \Bigg\}\mu_{0, 0}(X)            \end{align*}
    \begin{align*}
    &\Bigg[
    \dfrac{\widehat\tau(X) q(X)}{\tau(X) \widehat q(X)}-1-
    \dfrac{\widehat\tau(X) }{\tau(X)}\dfrac{1-q(X)}{1-\widehat q(X)}\dfrac{\widehat\mu_{0, 0}(X)}{\mu_{0, 0}(X)}\Big\{1-\dfrac{\mu_{1, 0}(X)}{\widehat\mu_{1, 0}(X)}\Big\}\Bigg]+\\
    &\dfrac{\mu_{1, 1}(X)}{\mu_{1, 0}(X)}\dfrac{q(X)}{\widehat q(X)}
    \dfrac{\widehat\tau(X) \widehat\mu_{0, 0}(X)}{\tau(X) \widehat\mu_{1, 0}(X)}\{\widehat\mu_{1, 0}(X)-\mu_{1, 0}(X)\}\Bigg\{
    1-\dfrac{1-q(X)}{1-\widehat q(X)}\dfrac{\widehat q(X)}{q(X)}
    \Bigg\}
    \Bigg)
    \Bigg\}\\
    \lesssim&(\widehat\kappa-\kappa)\alpha_{1}/\kappa + \\
    &\left\lVert \widehat\mu_{0, 0}(X)-\mu_{0, 0}(X)\right\rVert\cdot
    \left\lVert \widehat\pi'(X, W)-\pi'(X, W)\right\rVert+\\
    &\left\lVert \dfrac{\widehat\mu_{1, 1}(X)}{\widehat\mu_{1, 0}(X)}-\dfrac{\mu_{1, 1}(X)}{\mu_{1, 0}(X)}\right\rVert
    \Big\{
    \left\lVert \widehat\mu_{0, 0}(X)-\mu_{0, 0}(X)\right\rVert+
    \left\lVert\widehat\mu_{1, 0}(X)-\mu_{1, 0}(X)\right\rVert+\left\lVert\widehat\tau(X)-\tau(X)\right\rVert
    \Big\}\\
    &+\left\lVert\widehat q(X)-q(X)\right\rVert\Bigg\{
    \left\lVert \dfrac{\widehat\mu_{1, 1}(X)}{\widehat\mu_{1, 0}(X)}-\dfrac{\mu_{1, 1}(X)}{\mu_{1, 0}(X)}\right\rVert+\left\lVert\widehat\mu_{1, 0}(X)-\mu_{1, 0}(X)\right\rVert
    \Bigg\}.
\end{align*}
Combine the first term of 2 and 3, we have
\begin{align*}
     &\mathbb{P}\big\{H(\boldsymbol{\widehat \eta_3})-H(\boldsymbol{\eta_3})\big\}+
    (\mathbb{P}_n-\mathbb{P})H(\boldsymbol{\eta_3})\\
    =&\mathbb{P}_n\{A_3(O, \boldsymbol{\eta})\}+\Big\{\dfrac{\widehat\kappa}{\kappa} -1\Big\}\alpha_{3} +
    \Big\{\dfrac{\kappa}{\widehat\kappa} -1\Big\}\alpha_{3}
\end{align*}
Factoring and simplifying the last two terms yields $\alpha_3 \left\{\widehat\kappa- \kappa\right\}\left\{\frac{1}{\kappa} - \frac{1}{\widehat\kappa}\right\}$, which is
$o_{p}(n^{-1/2})$ by the central limit theorem. 
Therefore, combining terms 1, 2 and 3, we conclude
\begin{align*}
     \widehat \alpha_3 -\alpha_3 
    =O_{p}(R_{2, n})+\mathbb{P}_n\{A_3(O, \boldsymbol{\eta})\}+o_p(n^{-1/2}),
\end{align*}
In particular, if $R_{2, n}= o_p(n^{-1/2})$, then 
\begin{align*}
    \sqrt{n}\big(\widehat \alpha_3 -\alpha_3 \big)\rightsquigarrow\mathcal{N}\big[0, \E\{A_3(O, \boldsymbol{\eta_3})^{2}\}\big].
\end{align*}
The other proofs are committed here as they are trivial and similar to the proofs in \ref{appendix:inference2}.
\end{proof}
\clearpage

%%%%%%%%%%%%%%%%%%%%%%%%%%%%%%%%%%%%%%%%%%%%%%%%%%%%%%%%%%%%
\section{Simulation details}\label{appendix:Sim}
%%%%%%%%%%%%%%%%%%%%%%%%%%%%%%%%%%%%%%%%%%%%%%%%%%%%%%%%%%%%

In the simulation, we set
\begin{align*}
m(X)&=-0.5-0.4X_1-0.4X_2-0.4X_3-0.4X_4-0.4X_5\\
g(X)&=    0.4X_1+0.5X_2-0.5X_3-0.6X_4+0.6X_5\\
r(X)&=    0.5X_1-0.1X_2+0.3X_3+0.2X_4-0.3X_5.\\
\end{align*}
To correctly specify the models, we use logistic regression--regressing $Y$ on $\{X_{1},\dots, X_{5}\}$ using the subset $\{S=1, A=1\}$, $\{S=1, A=0\}$ and $\{S=0\}$ respectively to model $\mu_{1, 1}(X)$, $\mu_{1, 0}(X)$ and $\mu_{0, 0}(X)$. Similarly, we also use logistic regression to estimate $q(X)$ and $\tau(X)$.

To misspecify the model, we still use logistic regression, but we regress the outcome on $X_{2}$ only without the intercept.
\clearpage

%%%%%%%%%%%%%%%%%%%%%%%%%%%%%%%%%%%%%%%%%%%%%%%%%%%%%%%%%%%%
\section{Additional simulations}\label{appendix:AdditionalSim}
%%%%%%%%%%%%%%%%%%%%%%%%%%%%%%%%%%%%%%%%%%%%%%%%%%%%%%%%%%%%

Figure \ref{fig:Comp_Conv_mu10x} shows the comparison of $\widehat \alpha_{1}$ and $\widehat \alpha_{1}'$ when $\mu_{1,0}(X)$ is misspecified.
\begin{figure}
    \centering
    \includegraphics[width=\textwidth]{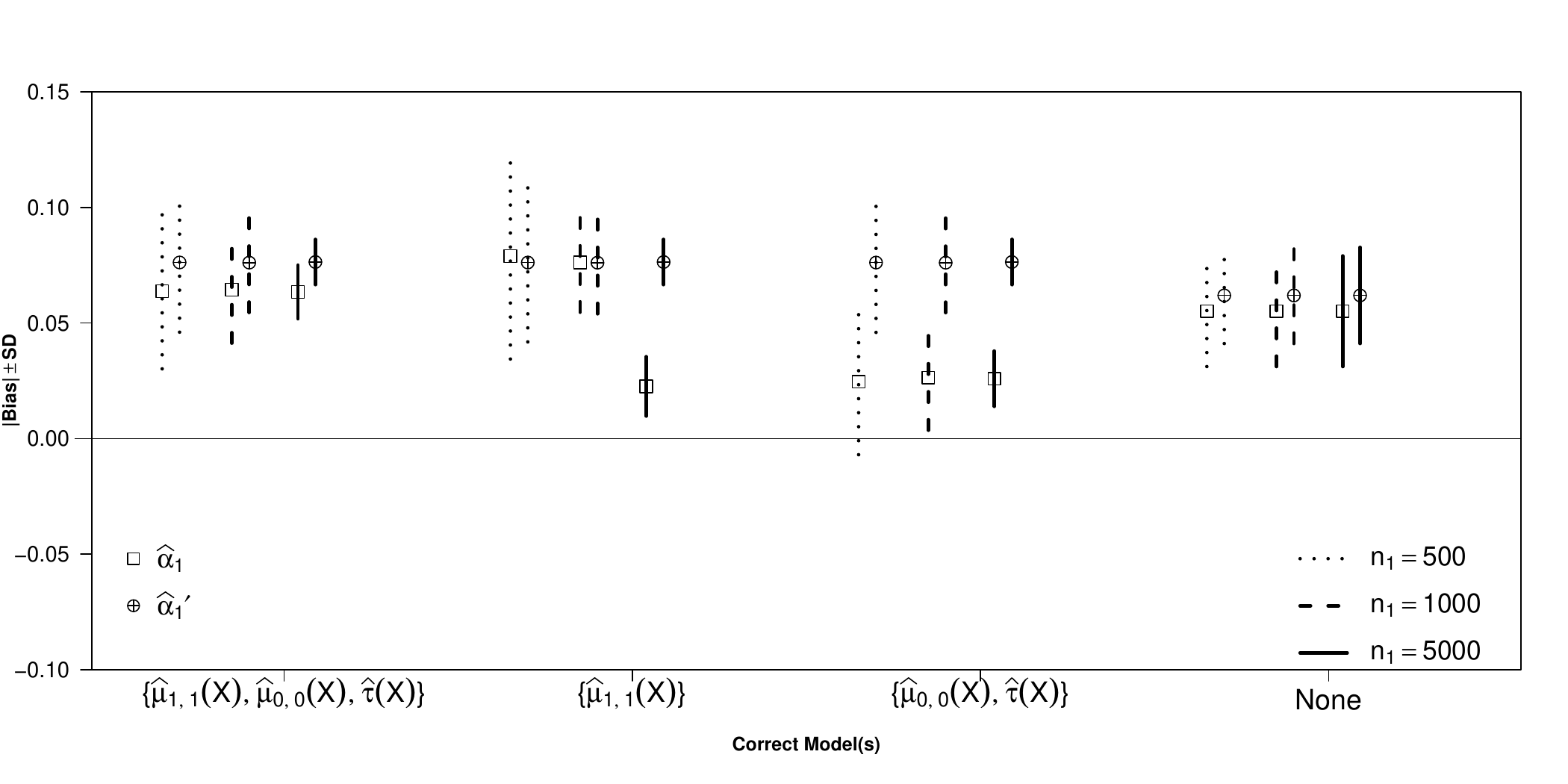}
    \caption{Absolute biases and standard deviations (SD) of the two estimators $\widehat \alpha_{1}$ and $\widehat \alpha_{1}'$ under different correctly specified nuisance functions and sample sizes. In all scenarios, $\widehat \mu_{1,0}(X)$ is misspecified.}  
    \label{fig:Comp_Conv_mu10x}
\end{figure}
\clearpage

%%%%%%%%%%%%%%%%%%%%%%%%%%%%%%%%%%%%%%%%%%%%%%%%%%%%%%%%%%%%
\section{Simulation: assessment of rate robustness}\label{appendix:sim_RateRobustness}
%%%%%%%%%%%%%%%%%%%%%%%%%%%%%%%%%%%%%%%%%%%%%%%%%%%%%%%%%%%%
To demonstrate the robustness of the proposed methods, we conduct a simple numerical experiment. We focus on estimation of $\alpha_1 = \mathrm{E}\{r_A(X) \mu_0(X) \mid S = 0\}$ defined in Section~\ref{sec:scenario1}. According to Theorem~\ref{thm:inference1}, the proposed influence function-based estimator $\widehat{\alpha}_1$ will be $\sqrt{n}$-consistent and nonparametrically asymptotically efficient when $\big(\left\lVert \widehat \mu_{1, 1}(X)-\mu_{1, 1}(X) \right\rVert+
    \left\lVert \widehat \mu_{1, 0}(X)-\mu_{1, 0}(X) \right\rVert\big)
    \big(
    \left\lVert \widehat r_S(X)-r_S(X) \right\rVert+
    \left\lVert \widehat \tau(X)-\tau(X) \right\rVert = o_p(n^{-1/2})$. 
Whereas simpler plug-in estimators (e.g., based on empirical versions of the identifying formulae given in Theorem~\ref{thm:identification1}) will in general converge only as fast as the component nuisance function estimates, the product asymptotic bias property of the proposed estimator yields faster convergence rates.

In order to illustrate the above advantage, consider the following data generating scenario: baseline confounders $X = (X_1, X_2)$ generated with $X_1 \sim \mathrm{Bernoulli}(0.5)$, independent of $X_2 \sim \mathrm{Uniform}(0,1)$, trial membership indicator $S \mid X_1, X_2 \sim \mathrm{Bernoulli}(t(X))$ with $t(X) = \mathrm{expit}(-0.2 + 0.5X_1 + 1.2 X_2)$, treatment status generated as $A \mid X_1, X_2, S = 0 \equiv 0$, $A \mid X_1, X_2, S = 1 \sim \mathrm{Bernoulli}(q(X))$ with $q(X) = \mathrm{expit}(0.3 + 0.9 X_1 - 0.8 X_2)$, and outcome generated as $Y \mid X_1, X_2, S, A \sim \mathcal{N}(m(X, S, A), 1)$, where
\[m(X, S, A) = \{0.75 (1 - S) + S\}\left\{ 5.2 + X_1 - 1.2 X_2 + A\left(1.2 -0.6 X_1\right)\right\}.\]

To ``estimate'' the nuisance functions, we added noise to perturb the true underlying functions as follows: $\widehat{t}(X) = \mathrm{expit}\left\{\mathrm{logit}(t(X)) + h\epsilon_{t}\right\}$, $\widehat{q}(X) = \mathrm{expit}\left\{\mathrm{logit}(q(X)) + h\epsilon_{q}\right\}$, $\widehat{\mu}_{s,a}(X) = m(X, s, a) + h\epsilon_{\mu, s, a}$, for $(s, a) \in \{(0,0), (1,0), (1,1)\}$, where $\epsilon_{t}, \epsilon_q, \epsilon_{\mu, 0,0}, \epsilon_{\mu, 1,0}, \epsilon_{\mu, 1,1} \overset{\mathrm{iid}}{\sim} \mathcal{N}(n^{-r}, n^{-2r})$, where we set $h = 2.2$ and $r \in (0, 0.5]$ over various scenarios. Next, we set $\widehat{\tau} = \frac{1 - \widehat{t}}{\widehat{t}}$, $\widehat{r}_A = \frac{\widehat{\mu}_{1,1}}{\widehat{\mu}_{1,0}}$, and $\widehat{r}_S = \frac{\widehat{\mu}_{0,0}}{\widehat{\mu}_{1,0}}$. By construction, we guarantee that $\left\lVert \widehat{\tau} - \tau\right\rVert = O_p(n^{-r})$, $\left\lVert \widehat{q} - q\right\rVert = O_p(n^{-r})$, $\left\lVert \widehat{\mu}_{s, a} - \mu_{s,a}\right\rVert = O_p(n^{-r})$, allowing us  to study the performance of the proposed estimator $\widehat{\alpha}_1$ with underlying nuisance errors on the order $O_p(n^{-r})$.

Setting $n \in \{1000,2000,5000\}$ and $r \in \{0.10 + 0.05k: k \in \{0, \ldots, 8\}\}$, we computed nuisance estimates as above, with 5,000 simulation replications for each scenario. In each replication, we computed the proposed influence function-based estimator $\widehat{\alpha}_1$, as well as a plug-in estimator $\frac{\mathbb{P}_n\left[\mathds{1}(S = 0)\widehat{r}_A(X) \widehat{\mu}_{0,0}(X)\right]}{\mathbb{P}_n[\mathds{1}(S = 0)]}$. To summarize results, we compared the root mean-square error (RMSE) of these two estimators: results are displayed in Figure~\ref{fig:sim_rates}.

\begin{figure}
    \centering
    \includegraphics[width = 0.7\linewidth]{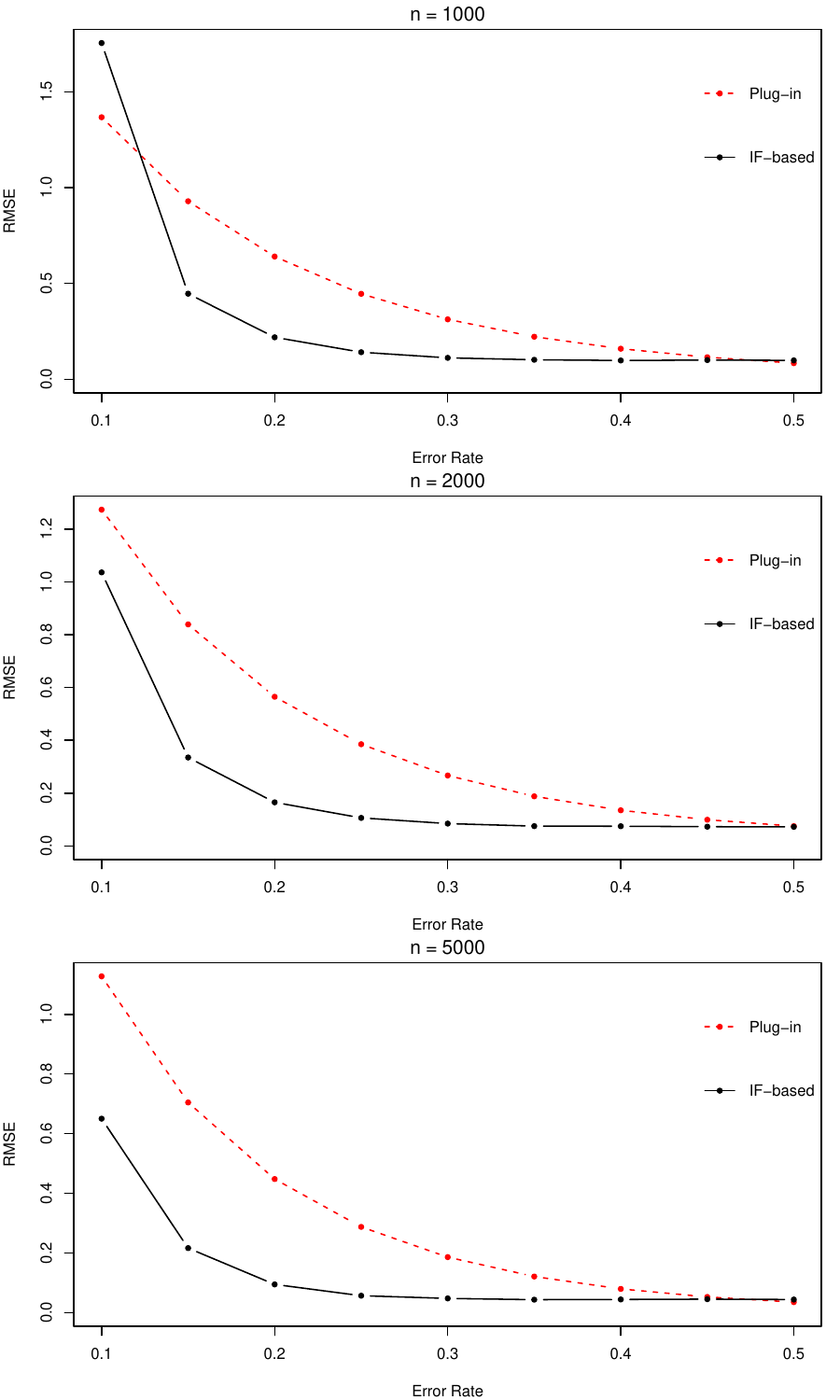}
    \caption{Simulation results for $\widehat {\alpha}_1$ for total sample size $n = 1000, 2000, 5000$, with 5000 iterations per scenario. Varying the error rate $r$ which controls the nuisance error at $O(n^{-r})$, the results include the root mean-square error (RMSE) of the proposed estimator (IF-based) and the estimator that only uses outcome models (plug-in).}
    \label{fig:sim_rates}
\end{figure}

According to the results in Figure~\ref{fig:sim_rates}, the proposed robust estimator outperforms the plug-in estimator, especially when nuisance error is of high order (i.e., the rate parameter $r$ is small). As predicted by Theorem~\ref{thm:inference1}, near-optimal performance of the doubly robust estimator is achieved once $r \geq 0.25$, whereas this is only achieved by the plug-in estimator when $r \approx 0.5$. One would only expect such low nuisance error, however, under correctly specified parametric models for $\widehat{\mu}_{s,a}$, which would be unlikely in any real data application.

%%%%%%%%%%%%%%%%%%%%%%%%%%%%%%%%%%%%%%%%%%%%%%%%%%%%%%%%%%%%
\section{Summary statistics of each trial}\label{appendix:application1}
%%%%%%%%%%%%%%%%%%%%%%%%%%%%%%%%%%%%%%%%%%%%%%%%%%%%%%%%%%%%

Table \ref{tab:TableOne} shows the baseline characteristics in the meta-trials, stratified by $S$. Results reported as mean (standard deviation) for continuous variables and count (percentage) for binary variables.\\
\begin{table}[h]
    \centering
    \begin{tabular}{ccc}
    \hline
    Source of data, S     &  0 & 1 \\
    \hline\hline
    Number of individuals &   284 & 448\\
    PANSS at 6 week $>$90 & 0.34 (0.48) & 0.30 (0.46)\\
    Paliperidone  ER & 0.68 (0.47) & 0.75 (0.43)\\
    PANSS Baseline & 93.26 (11.57) & 92.98 (12.46)\\
    Age &  41.87 (10.71) & 37.01 (10.63)\\
    Female & 0.29 (0.45) & 0.36 (0.48)\\
    Race, White &129 (45.4) & 219 (48.9)\\
    Race, Black &  151(53.2) & 90 (20.1)\\
    Race, Other & 4 (1.4) & 139 (31.0)\\
    \hline
    \end{tabular}
    \caption{Baseline characteristics in the two trials. Results reported as mean (standard deviation) for continuous variables and count (percentage) for binary variables.}
    \label{tab:TableOne}
\end{table}
\clearpage

%%%%%%%%%%%%%%%%%%%%%%%%%%%%%%%%%%%%%%%%%%%%%%%%%%%%%%%%%%%%
\section{Efficient estimation when the target population is a trial}\label{appendix:application_estimator}
%%%%%%%%%%%%%%%%%%%%%%%%%%%%%%%%%%%%%%%%%%%%%%%%%%%%%%%%%%%%

In this section, we provide efficient estimators corresponding to a special case of the treatment variation in the target population, as described in Section 3. The special case is the treatment option is the experimental treatment and the control that was assigned to the trial population. That is, both the trial and target population have the same treatment option. In this case, we can make an additional condition, for $a=\{0, 1\}, \E(Y^a|S=0, X)=\E(Y^a|S=0, A=a, X)$. Therefore, $\alpha_2=\E(Y^1|S=0)$ can be identified by $\E(\mu_{0, 1}(X)|S=0)$.

Similar to the previous developed, the influence function of $\alpha_2$, $\beta_2$, $\phi_2$ and $\psi_2$ are
\begin{align*}
    \widetilde A_{2}(O, \boldsymbol{\eta_{2}})=&\kappa(1-S)  \Bigg[
    \mu_{0, 1}(X)-\alpha_{2}  +\dfrac{A}{\pi(X)}\big\{Y-\mu_{0, 1}(X)\big\} 
    \Bigg],\\
    \widetilde B_{2}(O, \boldsymbol{\eta_{2}}) =&\kappa(1-S)  \Bigg[
    \mu_{0, 0}(X)-\beta_{2}  +\dfrac{1-A}{\pi(X)}\big\{Y-\mu_{0, 0}(X)\big\} 
    \Bigg],\\
    \widetilde \phi_{2}(O; \boldsymbol{\eta_{2}})=&\dfrac{1}{\beta_{2}}\{\widetilde A_{2}(O; \boldsymbol{\eta_{2}})-\phi_{2} \widetilde B_{2}(O; \boldsymbol{\eta_{2}})\},\quad\quad
    \widetilde \psi_{2}(O; \boldsymbol{\eta_{2}})=\widetilde A_{2}(O; \boldsymbol{\eta_{2}})-\widetilde B_{2}(O; \boldsymbol{\eta_{2}}),
\end{align*}
respectively.
The corresponding estimators can be obtained by solving the estimating equations, where the estimating equations are the empirical average of the above influence functions.
%%%%%%%%%%%%%%%%%%%%%%%%%REFERENCES%%%%%%%%%%%%%%%%%%%%%%%%%
%%%%%%%%%%%%%%%%%%%%%%%%%%%%%%%%%%%%%%%%%%%%%%%%%%%%%%%%%%%%

%%%%%%%%%%%%%%%%%%%%%%%%%%%%%%%%%%%%%%%%%%%%%%%%%%%%%%%%%%%%

\end{appendices}
\end{document}